\documentclass[11pt]{article}
\usepackage{amsmath}
\usepackage[dvips]{graphicx}
\usepackage{amsfonts}
\usepackage{amsopn}
\usepackage{amsthm} 
\usepackage{amssymb}
\usepackage{color}
\usepackage{hyperref}

\usepackage{latexsym}
\usepackage{amsgen}

\usepackage{enumerate}
\usepackage[toc,page,title,titletoc,header]{appendix}
\usepackage{graphicx}
\usepackage{indentfirst}
\usepackage{multicol}
\usepackage{booktabs}

\newtheorem{theorem}{Theorem}[section]
\newtheorem{definition}[theorem]{Definition}
\newtheorem{lemma}[theorem]{Lemma}
\newtheorem{corollary}[theorem]{Corollary}

\newtheorem{example}[theorem]{Example}
\newtheorem{remark}[theorem]{Remark}

\makeatletter

\newenvironment{figurehere}
  {\def\@captype{figure}}
  {}
\makeatother

\title{A notion of fractality for a class of states and noncommutative relative distance zeta functional}

\begin{document}

\author{
Yat Tin Chow  \footnote{Department of Mathematics, University of California, Riverside (yattinc@ucr.edu).  The author is partially supported by NSF DMS-2409903 and ONR N000142412661.}
}

\date{}
\maketitle

\begin{abstract}
In this work, we first recall the definition of the relative distance zeta function in \cite{LRZ_book,LRZ,LRZ3,LRZ5,LRZ6a} and slightly generalize this notion from sets to probability measures, and then move on to propose a novel definition a relative distance (and tube) zeta functional for a class of states over a C* algebra.  With such an extension, we look into the chance to define relative Minkowski dimensions in this context, and explore the notion of fractality for this class of states.  Relative complex dimensions as poles of this newly proposed relative distance zeta functional, as well as its geometric and transformation properties, decomposition rules and properties that respects tensor products are discussed.  We then explore some examples that possess fractal properties with this new zeta functional and propose functional equations similar to \cite{Hoffer,LapPe2,LapPeWi1,LRZ_book}.
\end{abstract}

\noindent { \footnotesize {\bf Mathematics Subject Classification
(MSC2020)}:
Primary: 11M41, 28A12, 28A75, 28A80, 46L51.  Secondary: 28B15, 42B20, 44A05, 44A15, 46L89
}

\noindent { \footnotesize {\bf Keywords}: Fractality, distance zeta functional, tube zeta functional, complex dimensions, C*-algebra, noncommutative measure and integration
}

\section{Introduction}

Fractality is traditionally captured via either a general notions of self-similarity, in an exact or approximate sense, or through that of dimensions, for instance in the sense of Minkowski or Hausdorff, so that in this context, a fractal is to be thought of as sets having non-integer dimensions.
However, as elaborated in \cite{LF_book,LRZ_book,LRZ3, LRZ5,LRZ6a,Richardson}, most of these notions, while being successful in capturing fractality in some occasions, fall short in capturing this properties at some other times.  For instance, the Cantor function, while widely agreed upon as a fractal, is rectifiable and has topological dimension $1$; space-filling curves, commonly agreed to be a fractal, has its image filling the $d$-dimensional space, and is thus of dimension $d$; meanwhile, $\alpha$-string generally has non-integer Minkowski dimension but does not possess the symmetry that the community would consider as a fractal.

With these observations, new notions are developed to refine the mathematical description of fractality, and a major class of methods that met with success is via various notions of zeta functions and complex dimensions as poles of the corresponding zeta functions concerned.  The first zeta function introduced for this purpose was that of fractal strings \cite{LF1,LF2,LF_book}, and then it is generalized in \cite{LRZ0,LRZ_book,LRZ,LRZ3,LRZ4,LRZ5,LRZ6a,LRZ6b,LRZ7} to tube and distance zeta functions to cover arbitrary fractal sets in the Euclidean spaces via developing appropriate relative fractal drums for this purpose.  In each of these cases, fractality is defined via the existence of non-real poles of the corresponding zeta functions.  These poles are referred to as complex dimensions and regarded as a generalized definition of dimensions, considering the fact that the upper Minkowski dimension is always a first pole.  The presence of the imaginary parts in the poles reflects the fact that a (generalized version of) Weyl's tube formula for the geometric object to be studied possesses some form of periodicity in the logarithmic scale (a.k.a. multiplicative periodicity), which are then converted to the imaginary parts of the poles in the zeta function via a Mellin transform of the tube formula.   It is this key observation that makes this approach's great success in capturing the symmetric property required to signify fractality, bypassing possible issues with solely using Minkowski or Hausdorff dimensions.  For a more detailed and thorough discussion of zeta functions and complex dimensions, we would refer the readers to consult \cite{LF1,LF2,LF_book,  LRZ0, LRZ_book, LRZ,LRZ3,LRZ4, LRZ5,LRZ6a,LRZ6b,LRZ7} and the references therein, along side with \cite{ElLapMacRo,Fal1,Fal2,
HL2,HL1,HerLap1,HerLap2,
LalLap1,LalLap2,
L1_ref,L2_ref,L3_ref,L4_ref,L5_ref,L6_ref,
LapLeRo,LapLu,
LLF1,LLF2,
LM1,LM2,
LapPe1,LapPe2,
LapPeWi1,LapPeWi2,
LP1,LP2,LP3,
LapRo,
LapRoZu,RaWi,
Tep1,Tep2}
for the development and detailed analysis on various versions of zeta functions and tube formulae for different specific occasions and its relevance to Reimann hypothesis.

In this work, we make a pioneering step to explore the possibility of extending the definition of the relative distance (and tube) zeta function in \cite{LF1,LF2,LF_book,LRZ0,LRZ_book,LRZ,LRZ3,LRZ4,LRZ5,LRZ6a,LRZ6b,LRZ7} to the noncommutative setting, more specifically to a class of states over a C* algebra. 
(A relevent but slightly different generalization of fractal zeta functions to (Ahlfors) metric measure spaces are also explored in \cite{Hen,LW}, see also \cite{LRZ}.)
We actualize this by first slightly extending this notion from sets to probability measures in the Euclidean space as in \eqref{zeta_communtative}, and view them as a special case of states over a commutative C* algebra.  Then we observe a consistent generalization can be made to provide the definition of a relative distance (and tube) zeta functional as in \eqref{zeta_non_communtative} for a special class of states.
Along this line, we may also define tubular neighborhood measures and functionals which help us define relative Minkowski dimensions in this context.
Relative complex dimensions can then be defined as poles of this newly proposed relative distance zeta functional, and can be regarded as a generaliztion of dimensions in view of the fact that under mild hypotheses the relative upper Minkowski dimension is always contained in the set of relative complex dimensions.  
Similar to the commutative case, the presence of the imaginary parts in the relative complex dimensions, i.e. the poles of the relative distance zeta functional, captures the periodicity in the logarithmic scale, i.e. the multiplicative periodicity, of the tubular neighborhood functionals.
This therefore gives us the right tools to explore the notion of fractality for our class of states over a C* algebra, which, being a limit of convex combinations of pure states via Choquet's theorem, may be regarded as a noncommutative generalization of geometric objects when we recall the fact that pure states over commutative C* algebra are always Dirac measures over points in a space.
Additional geometric and transformation properties, decomposition rules and properties that respects tensor products are explored such as functional equations similar to those introduced in \cite{Hoffer} could be obtained, which could possibly offer us a better description of the relative complex dimensions of the states considered.
Functional equations satisfied by the geometric, fractal and distance zeta function with rescaling as in \cite{Hoffer} were previously also introduced and used earlier in the following references, including \cite{KL93,LalLap1,LalLap2,Lal88,Lal89,L3_ref,LF1,LF2,LF_book,LapPe1,LapPe2,LapPeWi1,LapPeWi2,LRZ_book,LRZ,LRZ4}.

As an important remark to further understand the motivation of this work and its physical relevance, we would like to turn our attention to, and compare with, another major direction of fractal studies in the noncommutative setup, notably \cite{FLLP,LLL} in which the authors perform a quantization of fractals via the study of their spectral triple
 through the lens of Gromov-Hausdorff propinquity over quantum metric space \cite{Latremoliere_1, Latremoliere_2}.
Other works that aims for a quantization of fractals via the study of their spectral triple include \cite{GIL_new, L0_new, S_new, V_new}.
The difference between our work and thes aforementioned approaches is majorly on the respective motivations: their works were to look for a noncommutative characterization of (known) fractals (via a spectral triple, for instance), whereas our work and approach is to explore the other way around, asking the question of what might be considered as a fractal in a noncommutative setting, more concretely in the context of the C*-algebraic formulation of quantum mechanics where physical observables are their self-adjoint elements.
To understand more concretely what we mean here, let us recall that, in this setup, states of a C*-algebra can be regarded as a generalization of the notion of density matrices in quantum mechanics (which represent mixed  ``quantum states" in the physical system in quantum entanglement or statistical ensemble) and they now correspond to physical states as a mapping from physical observables (as self-adjoint elements) to their expected measurement outcomes (as real or complex numbers.)
Therefore asking for a characterization of fractality of states is equivalent to asking for such a characterization of a quantum system given only expected measurement outcomes of the relevant restricted set of physical observables (viewed as a geometric object).
A growing net of tubular neighborhood functionals then represents having increasingly more information of expected measurement outcomes over an increasingly larger set of physical observables which contains our relevant restricted set to be studied, the fractality of which is now described by the periodicity in the logarithmic scale/multiplicative periodicity to be observed when the set of physical observables from where information for expected measurement outcomes can be obtained grows.

Our work is organized as follows.  In Section \ref{sec_2_1}, we first recall the definition of the relative distance zeta function in  \cite{LRZ_book,LRZ,LRZ3,LRZ5,LRZ6a} and slightly generalize this notion from sets to probability measures.  Then we move on to defining relative distance (and tube) zeta functionals for a class of states over a C* algebra in Section \ref{sec_2_2}, and define relative Minkowski dimensions, relative complex dimensions and other notions of fractality in this context.
In Section \ref{sec_3_1} we explored the geometric and transformation properties of this newly proposed zeta functional; then in Section \ref{sec_3_2} the decomposition rules and in Section \ref{sec_3_3} the properties that respects tensor products.  Finally, we examine some examples that possess fractal properties with this new zeta functional in Section \ref{sec_4}.

\section{Noncommutative relative distance zeta functionals for states} \label{sec_2}

In this section, we would like to first recall the definition of the relative distance zeta function in \cite{Hen,LRZ_book,LRZ,LRZ3,LRZ5,LRZ6a,LW} for arbitrary sets in the Euclidean space and slightly generalize this to probability measures over Euclidean spaces.  Then we would move on to introducing a definition of the relative distance zeta functional for a class of states over a $C^*$-algebra.

\subsection{The commutative relative distance zeta function} \label{sec_2_1}

For $x \in \mathbb{R}^d$, 
we may consider, for a given Borel probability measure $\mu \in \mathcal{P} (\mathbb{R}^d)$, the following distance function:
\[
d (x, \mu ) := \lim_{p \rightarrow \infty}  \left( \int_{\mathbb{R}^d}  d(x,y)^{-2p} \mu (dy) \right)^{- \frac{1}{2p}} =  \inf \left\{  d(x, y) : y \in \text{spt} (\mu)  \right\} \,.
\]
where $d(x,y) = \sqrt{ \sum_{i=1}^d (x_i - y_i)^2 }$ is the Euclidean distance and $ \text{spt} (\mu)$ denotes the support of $\mu$.
%\[
% \text{spt} (\mu) := \bigcap \{ A  \in  \mathcal{B} \,:\,  \mu (\mathbb{R}^d \backslash A) = 0  \ \} 
%\] 
%where $\mathcal{B}$ is the Borel sigma algebra.
Then we may consider, as in  \cite{LRZ_book,LRZ,LRZ3,LRZ5,LRZ6a}, for another given Borel probability measure $\nu \in \mathcal{M} (\mathbb{R}^d)$, the relative distance zeta function.
\begin{definition}
The relative distance zeta function of a Borel probability measure $\mu \in \mathcal{P} (\mathbb{R}^d)$ given another measure $\nu \in \mathcal{M} (\mathbb{R}^d)$ is given by
\begin{equation}
\zeta_{\mu} (s, \nu) :=  \nu \left ( d (x, \mu )^{s-d}   \right) =  \int_{\mathbb{R}^d}  d (x, \mu )^{s-d} \nu (dx)
\label{zeta_communtative}
\end{equation}
for $s \in \mathbb{C}$ and $\Re (s) $ sufficiently large, and then meromorphically continued to a domain of definition $ \Omega_{\zeta_{\mu} (\,\cdot\,, \nu)}$ such that $\{ s \in \mathbb{C} \, : \, \Re (s) > \overline{\text{dim}}_B (\mu_y, \nu) \} \subset  \Omega_{\zeta_{\mu} (\,\cdot\,, \nu)}  \subset \mathbb{C}$.  (The fact that the integral is well-defined for $\Re (s) > \overline{\text{dim}}_B (\mu_y, \nu) $ is discussed in Lemma \ref{convergence_1}.)
\end{definition}

\begin{remark}
We remark that this notion is indeed a slight generalized definition of the relative distance zeta function introduced in  \cite{LRZ_book,LRZ,LRZ3,LRZ5,LRZ6a}.  In fact, given an arbiturary bounded Borel set $A \subset \mathbb{R}^d$, if we choose $ \mu $ to be supported on $A$ and $ \nu  = \chi_{  \{ x \in \mathbb{R}^d : d (x, A )  \leq \delta \}  } (x) \, dx $, then we get back the relative distance zeta function given in the aforementioned references.
\end{remark}

Now given a finite positive Borel measure $\nu \in \mathcal{M} (\mathbb{R}^d)$ (i.e. $\nu(1) < \infty $) and for any $\delta > 0$, let us always write a new Borel measure, a $\delta$-tubular-neighborhood measure of  $\mu \in \mathcal{P} (\mathbb{R}^d)$ with respect to $\nu$ as
\[
\nu^{\mu,\delta} := \nu \circ \text{mul}_{ \chi_{  \{ x \in \mathbb{R}^d : d (x, \mu )  \leq \delta \}  } }
\]
where $\text{mul}_{ f } ( g) := f g$ whenever $ f, g $ are Borel functions and $\chi_{A}$ is the indicator function of a Borel set $A$.
Then, following  \cite{LRZ_book,LRZ,LRZ3,LRZ5,LRZ6a}, 
we may now define the relative lower $s$-dimensional Minkowski content of $\mu$ with respect to $\nu$ for $s \geq 0 $:
\[
\mathcal{M}^s_*(\mu, \nu ) := \liminf_{\delta \rightarrow 0^+}  \, \delta^{ s - d} \, \nu^{\mu,\delta} (1) = \liminf_{\delta \rightarrow 0^+}  \, \delta^{ s - d} \,  \nu \left( \{ x \in \mathbb{R}^d : d (x, \mu )  \leq \delta  \} \right)   \,,
\]
and the relative upper $s$-dimensional Minkowski content of $\mu$ with respect to $\nu$ for $s \geq 0$:
\[
\mathcal{M}^{*\,s}(\mu, \nu ) := \limsup_{\delta \rightarrow 0^+}  \, \delta^{ s - d} \, \nu^{\mu,\delta} (1) = \limsup_{\delta \rightarrow 0^+}  \, \delta^{ s - d} \,  \nu \left( \{ x \in \mathbb{R}^d : d (x, \mu )  \leq \delta  \} \right)  \,,
\]
as well as the relative lower box counting dimension of $\mu$ with respect to $\nu$:
\[
\underline{\text{dim}}_B (\mu, \nu) = \inf \{ s \geq 0 \,;\,  \mathcal{M}^s_*(\mu, \nu )  = 0 \} = \sup \{ s \geq 0 \,;\,  \mathcal{M}^s_*(\mu, \nu )  = \infty \}.
\]
and the relative upper box counting dimension of $\mu$ with respect to $\nu$:
\[
\overline{\text{dim}}_B (\mu, \nu) = \inf \{ s \geq 0  \, ; \, \mathcal{M}^{*\,s}(\mu, \nu )  = 0 \} = \sup \{ s \geq 0  \, ; \, \mathcal{M}^{*\,s}(\mu, \nu )  = \infty \}.
\]
If $\underline{\text{dim}}_B (\mu, \nu) = \overline{\text{dim}}_B (\mu, \nu)  $, we write $\text{dim}_B (\mu, \nu) $ the relative box counting dimension of $\mu$ with respect to $\nu$.  $\mu$ is relatively Minkowski nondegenerate with respect to $\nu$ if 
\[
0 <  \mathcal{M}^{\text{dim}_B (\mu, \nu)}_*(\mu, \nu )  \leq  \mathcal{M}^{*\,\text{dim}_B (\mu, \nu)}(\mu, \nu )  < \infty \,,
\]
and the $\text{dim}_B (\mu, \nu)$-dimensional relative Minkowski content of $\mu$ with respect to $\nu$ is 
\[
\mathcal{M}^{\text{dim}_B (\mu, \nu)}(\mu, \nu )  := \mathcal{M}^{\text{dim}_B (\mu, \nu)}_*(\mu, \nu )  = \mathcal{M}^{*\,\text{dim}_B (\mu, \nu)}(\mu, \nu ) 
\]
if they are so equal. 

In what follows, we would like to check whether $\zeta_{\mu} (s, \nu ) $ is well-defined and holomorphic on $\Re (s) > \overline{\text{dim}}_B (\mu, \nu)$, hence
\[
\{ s \in \mathbb{C} \, : \, \Re (s) > \overline{\text{dim}}_B (\mu_y, \nu) \} \subset \Omega_{\zeta_{\mu} (\,\cdot \,, \nu ) }  \,.
\]
In fact, following \cite{LRZ_book,LRZ,LRZ3,LRZ5,LRZ6a}, the following lemma is in force, (the proof here is essentially the same as in the aforementioned references.)

\begin{lemma}
\label{convergence_1}
If $\nu$ is a finite Borel measure with $\text{spt} (\nu) \subset \{ x \in \mathbb{R}^d : d (x, \mu )  \leq \delta \}  $ for some small $ \delta > 0 $, then
$\zeta_{\mu} (s, \nu ) $ is well-defined and holomorphic whenever $\Re (s) > \overline{\text{dim}}_B (\mu, \nu)$, with
\begin{eqnarray}
\zeta_{\mu} (s, \nu ) =   \delta^{s-d}   \nu (1)  -  (s-d) \int_{0}^{\delta}  t^{s-d-1} \, \nu^{\mu,t}(1)  dt 
\label{zeta_communtative_tube}
\end{eqnarray}
and 
\begin{eqnarray}
\partial_s \zeta_{\mu} (s, \nu ) =  \int_{\mathbb{R}^d}  d (x, \mu )^{s-d}  \log(d (x, \mu )) \nu (dx)  \,.
\label{zeta_communtative_derivative}
\end{eqnarray}
Furthermore, if the relative box counting dimension of $\mu$ with respect to $\nu$, i.e. $\text{dim}_B (\mu, \nu) $, exists and $\text{dim}_B (\mu, \nu) < d$, and moreover, if $ \mathcal{M}^{\text{dim}_B (\mu, \nu)}_*(\mu, \nu )  > 0  $, then $\zeta_{\mu} (s, \nu ) \rightarrow \infty $ when $s \rightarrow \text{dim}_B (\mu, \nu)^+$ for $s \in \mathbb{R}$.
\end{lemma}

\begin{remark}
The map $ (s, \nu) \mapsto  \int_{0}^{\delta}  t^{s-d-1} \, \nu^{\mu,t}(1)  dt $  in \eqref{zeta_communtative_tube} can also be referred to as the relative tube zeta function of $\mu $ (with respect to $\nu$) in the spirit of \cite{LRZ_book,LRZ,LRZ3,LRZ5,LRZ6a}.
\end{remark}

\begin{proof}
It is direct to check that, since $\nu$ is finite (in particular $\sigma$-finite), whenever $s \neq d$
\begin{eqnarray*}
\int_{\mathbb{R}^d}  d (x, \mu )^{s-d} \nu (dx) 
& =&  (s-d) \int_{0}^{\delta}  t^{s-d-1}   \left(  \int_{ \{  d (x, \mu ) > t  \}   }  \nu (dx)  \right)  dt \\
& =&   \delta^{s-d}   \nu (1)  -  (s-d) \int_{0}^{\delta}  t^{s-d-1} \,  \nu \left( \{ x \in \mathbb{R}^d : d (x, \mu )  \leq t  \} \right) dt \\
& =&   \delta^{s-d}   \nu (1)  -  (s-d) \int_{0}^{\delta}  t^{s-d-1} \, \nu^{\mu,t}(1)  dt  \,.
\end{eqnarray*}
where the last line is finite whenever $\Re (s) > \overline{\text{dim}}_B (\mu, \nu)$, realizing from definition that
\[ 
\mathcal{M}^{*\, \frac{\Re (s)  + \overline{\text{dim}}_B (\mu, \nu)}{2} }(\mu, \nu )  =  \limsup_{\delta \rightarrow 0^+}  \, \delta^{  \frac{\Re (s)  + \overline{\text{dim}}_B (\mu, \nu)}{2} - d} \, \nu^{\mu,\delta} (1)  = 0 \,,
\]
leads to the existence of $\delta_0 > 0$ such that for all $t < \delta_0$
\[
t^{  \frac{\Re (s) + \overline{\text{dim}}_B (\mu, \nu) }{2} -d}  \nu^{\mu,t} (1)  <  1
\]
and hence
\begin{eqnarray*}
 \left | \int_{0}^{\delta_0}  t^{s-d-1} \,   \nu^{\mu,t}  (1) \, dt \right| \leq  \int_{0}^{\delta_0}  t^{  \frac{\Re (s) - \overline{\text{dim}}_B (\mu, \nu) }{2} -1} \,
\left( t^{  \frac{\Re (s) + \overline{\text{dim}}_B (\mu, \nu) }{2} -d}  \nu^{\mu,t} (1)   \right )\, dt  < \infty \,.
\end{eqnarray*}
We also notice that when $s = d$
\[
\int_{\mathbb{R}^d}  d (x, \mu )^{s-d} \nu (dx)  =  \nu (1)  =  \delta^{s-d}   \nu (1)  -  (s-d) \int_{0}^{\delta}  t^{s-d-1} \, \nu^{\mu,t}(1)  dt  \,.
\]
It is also handy to check via series expansion that, for a fixed $s$ such $\Re(s) > \overline{\text{dim}}_B (\mu, \nu)$, and for all $h \in \mathbb{C}$ such that $|h| < \frac{1}{2} \left( \Re(s) - \overline{\text{dim}}_B (\mu, \nu) \right) $,
\begin{eqnarray*}
&  &\left | \frac{ \zeta_{\mu} (s + h, \nu ) - \zeta_{\mu} (s , \nu ) }{h}  -  \int_{\mathbb{R}^d}  d (x, \mu )^{s-d}  \log(d (x, \mu )) \nu (dx)  \right| \\
&\leq&  C |h|
\int_{  \{ x \in \mathbb{R}^d : d (x, \mu )  \leq \delta \} }  
 \left( \log \left( d (X, \mu )  \right) \right)^2  d (X, \mu ) ^{ |h| }  \,  d (X, \mu ) ^{ \Re(s) - d - 2 |h| } \, \nu ( d x )  \\
&\leq&  C |h| \, \zeta_{\mu} ( \Re(s) - d - 2 |h|  , \nu ) \\
&=& O( h ) \,.
\end{eqnarray*}
The last part of the theorem comes directly from the definition of
\[
0 < \mathcal{M}^{\text{dim}_B (\mu, \nu)}_*(\mu, \nu ) := \liminf_{\delta \rightarrow 0^+}  \, \delta^{ \text{dim}_B (\mu, \nu)- d} \, \nu^{\mu,\delta} (1)   \,,
\]
we have the existence of $C > 0$ such that 
\[
 \nu^{\mu,\delta} (1) \geq C \delta^{ d- \text{dim}_B (\mu, \nu)} \,
\]
and that leads to the observation that if $ d \geq s > \text{dim}_B (\mu, \nu) $, we have
\begin{eqnarray*}
\infty > \zeta_{\mu} (s, \nu ) &=& \delta^{s-d}   \nu (1)  -  (s-d) \int_{0}^{\delta}  t^{s-d-1} \, \nu^{\mu,t}(1)  dt \\
&\geq &  - C (s-d)  \int_{0}^{\delta}  t^{s- \text{dim}_B (\mu, \nu)-1}   dt \\
&=&  - C (s-d)  \frac{ \delta^{s- \text{dim}_B (\mu, \nu)} }{  s- \text{dim}_B (\mu, \nu)   } 
\end{eqnarray*}
and this concludes $\zeta_{\mu} (s, \nu ) \rightarrow \infty $ when $s \rightarrow \text{dim}_B (\mu, \nu)^+$.
\end{proof}

\begin{remark}
Given a finite Borel measure $\nu$ and fixed $\delta > 0$, the tubular measure $\nu^{\mu,\delta} $ itself always satisfies $\text{spt} (\nu^{\mu,\delta} ) \subset \{ x \in \mathbb{R}^d : d (x, \mu )  \leq \delta \}  $, and therefore $\zeta_{\mu} (s, \nu^{\mu,\delta}  ) $ is well-defined and holomorphic whenever $\Re (s) > \overline{\text{dim}}_B (\mu, \nu^{\mu,\delta} )$.   Also, from the above representation of $\zeta_{\mu} (s, \nu^{\mu,\delta})$ in the lemma, we realize that
\[
 h_{\delta_1, \delta_2}(s) := \zeta_{\mu} (s, \nu^{\mu,\delta_2}) -  \zeta_{\mu} (s, \nu^{\mu,\delta_1})
\]
can be extended to an entire function and
\[
\overline{\text{dim}}_B (\mu, \nu^{\mu,\delta_1} ) = \overline{\text{dim}}_B (\mu, \nu^{\mu,\delta_2} )
\]
for any small $\delta_1 \neq \delta_2$ larger than $0$.  In fact this observation comes from the fact that
\[
h_{\delta_1, \delta_2} (s) =  ( \delta_2^{s-d} -  \delta_1^{s-d}  )   \nu (1) -   (s-d) \int_{\delta_1}^{\delta_2}  t^{s-d-1} \, \nu^{\mu,t}(1)  dt 
\]
and hence
\[
| h_{\delta_1, \delta_2} (s) - ( \delta_2^{s-d} -  \delta_1^{s-d}  )   \nu (1)   | \leq 
\begin{cases}
 |s-d| \nu^{\mu,\delta_2} (1)   \left( \frac{ \delta_2^{\Re(s)-d} -  \delta_1^{ \Re(s) -d}  }{ \Re(s) - d } \right)    &  \text{ when } \Re(s)  \neq d \,, \\
 |s-d| \nu^{\mu,\delta_2} (1)   \left( \log(\delta_2) - \log(\delta_1)   \right)  &  \text{ when } \Re(s) = d \,,
\end{cases} 
\]
which leads to the integral being an entire function for $s \in \mathbb{C}$ following the same holomorphicity argument as in  Lemma \ref{convergence_1}.
\end{remark}

Therefore, we may look into the poles of the relative distance zeta function of $\mu$ by a chosen measure $\nu$, i.e. considering the poles of 
\[
s \in  \Omega_{\zeta_{\mu} (\,\cdot \,, \nu ) } \mapsto \zeta_{\mu} (s, \nu ) \in \mathbb{C}  \,.
\]
We would also like to recall the definition of the relative complex dimensions of $\mu$ with respect to $\nu$ and the fractality of $\mu$. 

\begin{definition}
The relative complex dimensions of $\mu$ with respect to $\nu$ is the set of poles of $\zeta_{\mu} (s, \nu ) $, given as denoted as
\[
\mathcal{D}_{\mathbb{C}} ( \mu, \nu) := \left \{ s \in \Omega_{\zeta_{\mu} ( \, \cdot \,, \nu ) } \, : \,  \zeta_{\mu} (s, \nu )   = \infty \right \} \,.
\]
The measure $\mu$ is a fractal (measure) with respect to $\nu$ if 
\[
\mathcal{D}_{\mathbb{C}} ( \mu, \nu ) \backslash \mathbb{R} \neq \emptyset \,.
\]
\end{definition}

\begin{remark}
Given a finite Borel measure $\nu$ and for any small $\delta_1 \neq \delta_2$ larger than $0$, we realize from the previous observations that
\[
\mathcal{D}_{\mathbb{C}} ( \mu, \nu^{\mu,\delta_1} )  = \mathcal{D}_{\mathbb{C}} ( \mu, \nu^{\mu,\delta_2} ) 
\]
\end{remark}

\begin{remark} \label{Product} (Product decomposition)  
Let us  define
\[
d_\infty (x,\mu) :=  \lim_{p \rightarrow \infty}  \left( \int_{\mathbb{R}^d}  d_\infty (x,y)^{-2p} \mu (dy) \right)^{- \frac{1}{2p}} =  \inf \left\{  d_\infty (x, y) : y \in \text{spt} (\mu)  \right\} \,,
\]
where $d_{\infty}(x,y) = \sup_{i} | x_i - y_i | $ is the supremum-norm distance. 
We also define
\[
\nu^{\mu,\delta}_\infty := \nu \circ \text{mul}_{ \chi_{ \{ x \in \mathbb{R}^d : d_\infty (x, \mu )  \leq \delta \} }  } \,, \quad \mathcal{M}^s_{*,\infty}(\mu, \nu ) := \liminf_{\delta \rightarrow 0^+}  \, \delta^{ s - d} \, \nu^{\mu,\delta}_\infty (1) \,, \quad \mathcal{M}^{*\,s}_\infty(\mu, \nu ) := \limsup_{\delta \rightarrow 0^+}  \, \delta^{ s - d} \, \nu^{\mu,\delta}_\infty (1)  \,,
\]
and 
\[
\zeta_{\mu,\infty} (s, \nu) :=  \nu \left ( d_\infty (x, \mu )^{s-d}   \right) \,,
\]
then the following lemma holds, the last of which has been observed in  \cite{LRZ_book,LRZ,LRZ3,LRZ5,LRZ6a}:

\begin{lemma}
Whenever $\nu, \nu_i$ are finite Borel measures with  $\text{spt} (\nu) \subset \{ x \in \mathbb{R}^d : d (x, \mu )  \leq \delta \}  $ and  $\text{spt} (\nu_i) \subset \{ x \in \mathbb{R}^d : d (x, \mu_i )  \leq \delta \}  $ for all $1 \leq i \leq k$,
\begin{enumerate}
\item
it holds that for $s \in \mathbb{R}$,
\begin{eqnarray*}
 d^{-\frac{s-d}{2}} \mathcal{M}^s_{*}(\mu, \nu )   \leq \mathcal{M}^s_{*,\infty}(\mu, \nu )  \leq d^{\frac{s-d}{2}} \mathcal{M}^s_{*}(\mu, \nu )   \,, \\
 d^{-\frac{s-d}{2}} \mathcal{M}^{*\,s}(\mu, \nu )  \leq \mathcal{M}^{*\,s}_\infty(\mu, \nu )  \leq d^{\frac{s-d}{2}} \mathcal{M}^{*\,s}(\mu, \nu ) \,.
\end{eqnarray*}
\item
It holds that for $s_i \in \mathbb{R}$ for all $i$,
\begin{eqnarray*}
 \mathcal{M}^{\sum_{i=1}^k s_i}_{*,\infty}( \otimes_{i=1}^k \mu_i, \otimes_{i=1}^k \nu_i  )  & \geq & \prod_{i=1}^k \mathcal{M}^{s_i}_{*,\infty}(\mu_i, \nu_i )  \text{ with a convention that } 0 \cdot \infty = 0 \,, \\
 \mathcal{M}^{*\, \sum_{i=1}^k s_i}_\infty ( \otimes_{i=1}^k \mu_i, \otimes_{i=1}^k \nu_i  )  & \leq & \prod_{i=1}^k \mathcal{M}^{*\,s_i}_\infty (\mu_i, \nu_i )   \text{ with a convention that } 0 \cdot \infty = \infty \,,
\end{eqnarray*}
and
\begin{eqnarray*}
\underline{\text{dim}}_B ( \otimes_{i=1}^k \mu_i,   \otimes_{i=1}^k  \nu_i)& \geq &\sum_{i=1}^k \underline{\text{dim}}_B (\mu_i, \nu_i) \, , \\
\overline{\text{dim}}_B ( \otimes_{i=1}^k \mu_i,   \otimes_{i=1}^k  \nu_i) &  \leq & \sum_{i=1}^k \overline{\text{dim}}_B  (\mu_i, \nu_i) .
\end{eqnarray*}
Hence if $\underline{\text{dim}}_B (\mu_i, \nu_i) = \overline{\text{dim}}_B (\mu_i, \nu_i)  $ for all $i$, then
\begin{eqnarray*}
\text{dim}_B ( \otimes_{i=1}^k \mu_i,   \otimes_{i=1}^k  \nu_i) = \sum_{i=1}^k \text{dim}_B (\mu_i, \nu_i) \,,
\end{eqnarray*}
Furthermore, if for all $i$, $\mu_i$ is relatively Minkowski nondegenerate with respect to $\nu_i$ and the following $\text{dim}_B (\mu_i, \nu_i)$-dimensional relative Minkowski content in $\infty$-norm
\[ 
\mathcal{M}_\infty^{\text{dim}_B (\mu_i, \nu_i)}(\mu_i, \nu_i ) := \mathcal{M}_{*,\infty}^{\text{dim}_B (\mu_i, \nu_i)}(\mu_i, \nu_i ) = \mathcal{M}_\infty^{*\, \text{dim}_B (\mu_i, \nu_i)}(\mu_i, \nu_i ) 
\]
are all well-defined, then 
\[
 \mathcal{M}^{  \text{dim}_B ( \otimes_{i=1}^k \mu_i,   \otimes_{i=1}^k  \nu_i)  }_{\infty}( \otimes_{i=1}^k \mu_i, \otimes_{i=1}^k \nu_i  )  = \prod_{i=1}^k \mathcal{M}^{\text{dim}_B (\mu_i, \nu_i)}_{\infty}(\mu_i, \nu_i ) \,.
\]
\item
For $s \in \mathbb{R}$ such that $s > \overline{\text{dim}}_B (\mu, \nu)$, we have
\[
 d^{-\frac{s-d}{2}}  \zeta_{\mu} (s, \nu) \leq \zeta_{\mu,\infty} (s, \nu)   \leq  d^{\frac{s-d}{2}}  \zeta_{\mu} (s, \nu) \,.
\]
\item
For all $s_i \in \mathbb{C} $ such that  $ \Re(s_i) > \overline{\text{dim}}_B (\mu_i,  \nu_i)$, we have
\begin{eqnarray*}
&  & \left| \frac{ \zeta_{ \otimes_{i=1}^k \mu_i } ( \sum_{i=1}^k s_i , \otimes_{i=1}^k \nu_i ) -  \delta^{( \sum_{i=1}^k s_i - \sum_{i=1}^k d_i  ) }  \prod_{i=1}^k \nu_{i} (1)  }{ \sum_{i=1}^k s_i - \sum_{i=1}^k d_i } \right|    \\
&\leq  & -  \sum_{i=1}^k    \frac{  \left(  k d_i  \left( \sum_{i=1}^k d_i \right) \right)^{-   \frac{ k ( \Re(s_i) - d_i ) }{2} }   \zeta_{\mu_i^{\otimes^k} } (k \Re(s_i) , \nu_i^{\otimes^k}  ) - \delta^{ (k \Re(s_i) - k d_i )}  [ \nu_{i} (1) ]^k }{ k^2 (\Re(s_i) -  d_i )}    \,.
\end{eqnarray*}

\item
(Shift property)
Whenever the relative box counting dimension $ \text{dim}_B (\mu_2, \nu_2) $  
and the  $\text{dim}_B (\mu_2, \nu_2)$-dimensional relative Minkowski content $\mathcal{M}^{\text{dim}_B (\mu_2, \nu_2)}(\mu_2, \nu_2 ) $ are well-defined,  then for all $s \in \mathbb{R}$ such that $s - \text{dim}_B (\mu_2, \nu_2) > \overline{\text{dim}}_B (\mu_1,  \nu_1 )$,  we have
\begin{eqnarray*}
&    &  -  \frac{1}{2}  d_2^{ - \frac{d_2 -  \text{dim}_B (\mu_2, \nu_2)}{2} }  \left( \frac{ d_1^{\frac{s -  \overline{\text{dim}}_B (\mu_2, \nu_2) -d_1 }{2}}  \zeta_{\mu_1} \left(s  - \text{dim}_B (\mu_2, \nu_2 ) , \nu_1 \right) -  \delta^{(s  - \text{dim}_B (\mu_2, \nu_2 ) - d_1 )}   \nu_{1} (1)  }{ (s - \text{dim}_B (\mu_2, \nu_2 )  - d_1 )} \right) \\
& \leq & 
- \frac{  (d_1 + d_2)^{-\frac{s -d_1 - d_2 }{2}}  \zeta_{\mu_1 \otimes \mu_2} \left(s  , \nu_1 \otimes \nu_2 \right) -  \delta^{(s - d_1 - d_2)}   \nu_{1} (1) \nu_2 (1) }{ \mathcal{M}^{\text{dim}_B (\mu_2, \nu_2 )} (\mu_2, \nu_2 ) (s  - d_1 - d_2)}  \\
& \leq &  - 2 d_2^{  \frac{d_2 -  \text{dim}_B (\mu_2, \nu_2)}{2} }  \left( \frac{ d_1^{-\frac{s -  \overline{\text{dim}}_B (\mu_2, \nu_2) -d_1 }{2}}  \zeta_{\mu_1} \left(s  - \text{dim}_B (\mu_2, \nu_2 ) , \nu_1 \right) -  \delta^{(s  - \text{dim}_B (\mu_2, \nu_2 ) - d_1 )}   \nu_{1} (1)  }{ (s - \text{dim}_B (\mu_2, \nu_2 )  - d_1 )} \right) \,.
\end{eqnarray*}
More precisely, for $ p<a<b<q$,  for all  $s \in \mathbb{R}$ with $s  > \overline{\text{dim}}_B (\mu,  \nu ) + k$, we have
\begin{eqnarray*}
&   &\frac{ \zeta_{\mu \otimes \left( \frac{1}{(b-a)^k} \chi_{[a,b]^k} \mathcal{L}^k \right) ,\infty} \left(s , \nu \otimes \left( \frac{1}{(q-p)^k} \chi_{[q,p]^k} \mathcal{L}^k \right) \right) -  (s  - d - k)   \nu (1)  }{ (s - d - k )} \\
& =&    \left( q-p  \right)^{-k}   \sum_{l=0}^k C_{l}^k    2^{k-l}(b-a)^l  \frac{ \zeta_{\mu ,\infty} \left(s - l , \nu  \right) -  \delta^{(s  - l- d )}   \nu (1)  }{ (s - l - d  )} \,.
\end{eqnarray*}

\end{enumerate}
\end{lemma}

\begin{proof}
The first part and third part of the lemma come from an equivalence of norms in finite dimension, according to which
\[
d^{-\frac{1}{2}} d (x,\mu) \leq d_\infty (x,\mu) \leq d^{\frac{1}{2}}  d (x,\mu) \,,
\]
and therefore,
\[
\{ x \in \mathbb{R}^d : d^{-\frac{1}{2}} d (x, \mu )  \leq  \delta \} \subset \{ x \in \mathbb{R}^d : d_\infty (x, \mu )  \leq \delta \} \subset \{ x \in \mathbb{R}^d :  d^{\frac{1}{2}}  d (x, \mu )  \leq \delta \}
\]
and hence
\[
\nu^{\mu, d^{-\frac{1}{2}} \delta} (1) \leq \nu^{\mu, \delta}_\infty (1) \leq \nu^{\mu, d^{\frac{1}{2}} \delta} (1) \quad \text{ and } \quad d^{-\frac{s-d}{2}}  \zeta_{\mu} (s, \nu) \leq \zeta_{\mu,\infty} (s, \nu)   \leq  d^{\frac{s-d}{2}}  \zeta_{\mu} (s, \nu) \,.
\]

The second part and the fourth parts of the lemma come from
\[
\left \{ x \in \mathbb{R}^d : d_\infty \left((x_1,...,x_k) , \otimes_{i=1}^k \mu_i \right)  \leq \delta \right\}
= \prod_{i=1}^k  \left \{ x \in \mathbb{R}^d : d_\infty \left(x_i, \mu_i \right)  \leq \delta \right \} \,,
\]
and therefore,
\[
(\otimes_{i=1}^k \nu_i )_\infty ^{ \otimes_{i=1}^k \mu_i , \delta} (1)  =   \prod_{i=1}^k  ( \nu_i )_\infty^{ \mu_i ,  \delta} (1)   \,,
\]
as well as the following equality (which comes from a similar proof to \eqref{zeta_communtative_tube})
\[
\zeta_{\mu,\infty} (s, \nu ) =   \delta^{s-d}   \nu (1)  -  (s-d) \int_{0}^{\delta}  t^{s-d-1} \, \nu_\infty^{\mu,t}(1)  dt  \,.
\]
In fact, since for all real $a_i \geq 0 $, $1 \leq i \leq k $, we have
\begin{eqnarray}
\prod_{i=1}^k a_i \leq \left( \frac{ \sum_{i=1}^k a_i  }{k} \right)^{k} \leq  \frac{ \sum_{i=1}^k a_i^k  }{k} 
\label{elementary_inequality}
\end{eqnarray}
we now have, for all $s_i \in \mathbb{C} $ such that $ \Re(s_i) > \overline{\text{dim}}_B (\mu_i,  \nu_i)$,
\begin{eqnarray*}
&  &\left| \frac{ \zeta_{\otimes_{i=1}^k \mu_i ,\infty} ( \sum_{i=1}^k s_i  , \otimes_{i=1}^k \nu_i  ) -  \delta^{(  \sum_{i=1}^k s_i -  \sum_{i=1}^k d_i )}    \prod_{i=1}^k \nu_{i} (1) }{ ( \sum_{i=1}^k s_i -  \sum_{i=1}^k d_i )} \right|  \\
& = &  \left|   \int_{0}^{\delta}  t^{\sum_{i=1}^k s_i  - \sum_{i=1}^k d_i -1} \,  \prod_{i=1}^k ( \nu_i )_\infty^{ \mu_i ,  t} (1)     dt \right|  \\
& \leq &   \sum_{i=1}^k  \frac{1}{k} \int_{0}^{\delta}  t^{k \Re(s_i) -k d_i -1} \, [  ( \nu_i )_\infty^{ \mu_i ,  t} (1)  ]^k   dt   \\
&  =  & -  \sum_{i=1}^k \frac{ \zeta_{\mu_i^{\otimes^k} ,\infty} (k \Re(s_i) , \nu_i^{\otimes^k}  ) -  \delta^{ (k \Re(s_i) - kd_i )}  [ \nu_{i} (1) ]^k }{ k^2 (\Re(s_i) -  d_i )}  
\end{eqnarray*}
and
\begin{eqnarray*}
&  &\left| \frac{ \zeta_{ \otimes_{i=1}^k \mu_i } ( \sum_{i=1}^k s_i , \otimes_{i=1}^k \nu_i ) -  \delta^{( \sum_{i=1}^k s_i - \sum_{i=1}^k d_i  ) }  \prod_{i=1}^k \nu_{i} (1)  }{ \sum_{i=1}^k s_i - \sum_{i=1}^k d_i } \right|  \\
& = &  \left|   \int_{0}^{\delta}  t^{\sum_{i=1}^k s_i  - \sum_{i=1}^k d_i  -1} \,   (\otimes_{i=1}^k  \nu_i  )^{\otimes_{i=1}^k  \mu_i,  t} (1)      dt \right|  \\
& \leq &  \int_{0}^{\delta}  t^{ \sum_{i=1}^k \Re(s_i) - \sum_{i=1}^k d_i -1} \,   ( \otimes_{i=1}^k \nu_i  )^{ \otimes_{i=1}^k \mu_i ,  t} (1)      dt  \\
& \leq &  \int_{0}^{\delta}  t^{\sum_{i=1}^k \Re(s_i) - \sum_{i=1}^k d_i -1} \,   ( \otimes_{i=1}^k \nu_i  )_\infty^{ \otimes_{i=1}^k  \mu_i ,  ( \sum_{i=1}^k d_i )^{\frac{1}{2}} t} (1)      dt  \\
& = &  \int_{0}^{\delta}  t^{ \sum_{i=1}^k \Re(s_i)  - \sum_{i=1}^k d_i -1}  \,   \prod_{i=1}^k  ( \nu_i )_\infty^{ \mu_i , ( \sum_{i=1}^k d_i )^{\frac{1}{2}}  t} (1)      dt  \\
& \leq & \sum_{i=1}^k  \frac{1}{k} \int_{0}^{\delta}  t^{k \Re(s_i) - k d_i -1} \, [  ( \nu_i )_\infty^{ \mu_i , ( \sum_{i=1}^k d_i )^{\frac{1}{2}} t} (1)  ]^k   dt   \\
& \leq & \sum_{i=1}^k  \frac{1}{k} \int_{0}^{\delta}  t^{k \Re(s_i) -k d_i -1} \,   ( \nu_i^{\otimes^k} )^{ \mu_i^{\otimes^k} ,  (k d_i )^{\frac{1}{2}}   ( \sum_{i=1}^k d_i )^{\frac{1}{2}} t} (1)   dt   \\
& \leq & \sum_{i=1}^k  \left(  k d_i  \left( \sum_{i=1}^k d_i \right) \right)^{-   \frac{ k ( \Re(s_i) - d_i ) }{2} }    \frac{1}{k} \int_{0}^{  k^{\frac{1}{2}} d_i^{\frac{1}{2}}  \left( \sum_{i=1}^k d_i \right)^{\frac{1}{2}}  \delta}   t^{k \Re(s_i) -k d_i -1} \,   ( \nu_i^{\otimes^k} )^{ \mu_i^{\otimes^k} ,   t} (1) \,  d   t   \\
&  =  & -  \sum_{i=1}^k    \frac{  \left(  k d_i  \left( \sum_{i=1}^k d_i \right) \right)^{-   \frac{ k ( \Re(s_i) - d_i ) }{2} }   \zeta_{\mu_i^{\otimes^k} } (k \Re(s_i) , \nu_i^{\otimes^k}  ) - \delta^{ (k \Re(s_i) - k d_i )}  [ \nu_{i} (1) ]^k }{ k^2 (\Re(s_i) -  d_i )}   \,.
\end{eqnarray*}

The fifth statement comes from the fact that,
whenever the relative box counting dimension $ \text{dim}_B (\mu_2, \nu_2) $  
as well as the  $\text{dim}_B (\mu_2, \nu_2)$-dimensional relative Minkowski $\mathcal{M}^{\text{dim}_B (\mu_2, \nu_2)}(\mu_2, \nu_2 ) $ are well-defined, we have
\begin{eqnarray*}
& &  d_2^{ - \frac{d_2 -  \text{dim}_B (\mu_2, \nu_2)}{2} } \mathcal{M}^{\text{dim}_B (\mu_2, \nu_2 )} (\mu_2, \nu_2 )    \delta^{ d_2 -  \text{dim}_B (\mu_2, \nu_2) } - o( \delta^{ d_2 - \text{dim}_B (\mu_2, \nu_2) } ) \\
& \leq & \,( \nu_2 )_{\infty}^{\mu_2,\delta} (1) \\
& \leq & d_2^{  \frac{d_2 -  \text{dim}_B (\mu_2, \nu_2)}{2} } \mathcal{M}^{\text{dim}_B (\mu_2, \nu_2 )} (\mu_2, \nu_2 )    \delta^{ d_2 -  \text{dim}_B (\mu_2, \nu_2) } + o( \delta^{ d_2 - \text{dim}_B (\mu_2, \nu_2) } ) \,.
\end{eqnarray*}
Therefore we have for all $s \in \mathbb{R}$ such that $s - \text{dim}_B (\mu_2, \nu_2) > \overline{\text{dim}}_B (\mu_1,  \nu_1 )$,
\begin{eqnarray*}
&    &  \frac{1}{2}  d_2^{ - \frac{d_2 -  \text{dim}_B (\mu_2, \nu_2)}{2} } \mathcal{M}^{\text{dim}_B (\mu_2, \nu_2 )} (\mu_2, \nu_2 ) \int_{0}^{\delta}  t^{s -  \overline{\text{dim}}_B (\mu_2, \nu_2) -d_1 -1} \,   ( \nu_1 )_\infty^{ \mu_1 ,  t} (1)     dt  \\
& \leq &  \int_{0}^{\delta}  t^{s -d_1 - d_2 -1} \,   ( \nu_1 )_\infty^{ \mu_1 ,  t} (1)    ( \nu_2 )_\infty^{ \mu_2 ,  t} (1)    dt \\
& \leq & 2 d_2^{  \frac{d_2 -  \text{dim}_B (\mu_2, \nu_2)}{2} } \mathcal{M}^{\text{dim}_B (\mu_2, \nu_2 )} (\mu_2, \nu_2 )  \int_{0}^{\delta}  t^{s -  \overline{\text{dim}}_B (\mu_2, \nu_2) -d_1 -1} \,   ( \nu_1 )_\infty^{ \mu_1 ,  t} (1)     dt 
\end{eqnarray*}
and thus we have
\begin{eqnarray*}  
&  & - \frac{1}{2}  d_2^{ - \frac{d_2 -  \text{dim}_B (\mu_2, \nu_2)}{2} }    \left(  \frac{ \zeta_{\mu_1,\infty} \left(s  - \text{dim}_B (\mu_2, \nu_2 ) , \nu_1 \right) -  \delta^{(s  - \text{dim}_B (\mu_2, \nu_2 ) - d_1 )}   \nu_{1} (1)  }{ (s - \text{dim}_B (\mu_2, \nu_2 )  - d_1 )}  \right)  \\
&\leq  & - \frac{ \zeta_{\mu_1 \otimes \mu_2 ,\infty} \left( \Re(s)  , \nu_1 \otimes \nu_2 \right) -  \delta^{( \Re(s) - d_1 - d_2)}   \nu_{1} (1) \nu_2 (1) }{ \mathcal{M}^{\text{dim}_B (\mu_2, \nu_2 )} (\mu_2, \nu_2 ) ( \Re(s)  - d_1 - d_2)} \\
& \leq &  - 2  d_2^{  \frac{d_2 -  \text{dim}_B (\mu_2, \nu_2)}{2} }  \left( \frac{ \zeta_{\mu_1,\infty} \left(s  - \text{dim}_B (\mu_2, \nu_2 ) , \nu_1 \right) -  \delta^{(s  - \text{dim}_B (\mu_2, \nu_2 ) - d_1 )}   \nu_{1} (1)  }{ (s - \text{dim}_B (\mu_2, \nu_2 )  - d_1 )}  \right)
\end{eqnarray*}
and hence our conclusion, combining the above with the third part of the lemma.
In particular, for $ p<a<b<q$,  for all $s - k > \overline{\text{dim}}_B (\mu,  \nu )$, $s \in \mathbb{R}$,
\begin{eqnarray*}
&  & \frac{ \zeta_{\mu \otimes \left( \frac{1}{(b-a)^k} \chi_{[a,b]^k} \mathcal{L}^k \right) ,\infty} \left(s , \nu \otimes \left( \frac{1}{(q-p)^k} \chi_{[q,p]^k} \mathcal{L}^k \right) \right) -  (s  - d - k)   \nu (1)  }{ (s - d - k )} \\
& = &     \int_{0}^{\delta}  t^{s - d - k} \,    \nu _\infty^{ \mu ,  t} (1)   \left[ \left( \frac{1}{(b-a)} \chi_{[a,b]} \mathcal{L} \right)_\infty^{ \frac{1}{(q-p)} \chi_{[q,p]} \mathcal{L} ,  t} (1)  \right]^k  dt   \\
& = &    \left( q-p  \right)^{-k}    \int_{0}^{\delta}  t^{s - d -k } \,   \nu _\infty^{ \mu ,  t} (1)  \left(2 t  + (b-a)\right)^k  dt   \\
& = &    \left( q-p  \right)^{-k}   \sum_{l=0}^k C_{l}^k    2^{k-l}(b-a)^l  \int_{0}^{\delta}  t^{s - d  - l } \nu _\infty^{ \mu ,  t} (1)  \,    dt   \\
& = &    \left( q-p  \right)^{-k}   \sum_{l=0}^k C_{l}^k    2^{k-l}(b-a)^l  \frac{ \zeta_{\mu ,\infty} \left(s - l , \nu  \right) -  \delta^{(s  - l- d )}   \nu (1)  }{ (s - l - d  )}  \,.
\end{eqnarray*}
\end{proof}

\end{remark}

% \noindent It remains to check that the above definition of $\mathcal{D}_{\mathbb{C}} ( \mu, \nu) $ is well-defined and is independent of the choice of $\delta$.

\subsection{The noncommutative relative distance zeta functional} \label{sec_2_2}

Let $d < \infty$, and let us consider $\mathcal{A} \langle X \rangle :=  \mathbb{C} \, \langle X_1,...,X_d \rangle $ the universal unital algebra of noncommutative polynomials
in $d$ indeterminates $X_1,...,X_d \in \mathfrak{sa} \,  \mathcal{L} (H) $ where $H$ 
% $H \cong \mathbb{R}^n$ 
is a (complex)
Hilbert space, $\mathcal{L} (H)$ is the space of bounded linear operators on $H$, and $\mathfrak{sa} \,  \mathcal{L} (H) $ are the self-adjoint elements in $\mathcal{L} (H)$.  
(In here, $H$ is not necessarily finite dimensional, nor is it required to be separable.  This flexibility on the choice of $H$ is to include e.g. the universal representation of a $C^*$-algebra.  Meanwhile, please notice that the space $\mathcal{L} (H)$ is already a von Neumann algebra with its predual as the space of trace class operators, and hence we might always restrict ourselves to a closed subspace in the choice of $X_i$.)  
Let us equip $\mathcal{A} \langle X \rangle$ with the family of norms $\| \cdot \|_R$, $R > 0$, as follows. For any monomial $p$ and arbitrary $f \in \mathcal{A} \langle X \rangle $, let $\lambda_p (f)$ be the coefficient of $p$ in the decomposition of $f$, i.e.
\[
f = \sum_{p} \lambda_p(f) \, p
\]
and 
\[
\| f \|_R := \sum_{\text{deg} \, p \, \geq \, 0} |\lambda_p(f)| \, R^{\text{deg} \, p} \,.
\]
Here, we denote $X_i^0 := 1$, the identity element in $\mathcal{A} \langle X \rangle$.
Write $\overline{ \mathcal{A} \langle X \rangle } := \overline{ \mathcal{A} \langle X \rangle }^{\|\cdot \|_R} $ for a choice of large $R > 0$, which we may not emphasize from now on except when necessary.
% (as well as $\mathcal{A} \langle X \rangle^{\text{op}}$ its opposite algebra and $\overline{ \mathcal{A} \langle X \rangle^{\text{op}} } := \overline{ \mathcal{A} \langle X \rangle^{\text{op}} }^{\|\cdot \|_R} $.)  
Notice that if $T \in B_R^{  [\mathfrak{sa} \mathcal{L}(H) ]^d } (0) := \left \{  (T_i )_{i=1}^d \in [ \mathfrak{sa} \, \mathcal{L}(H) ]^d \,:\, \|T_i\|_{\mathcal{L}(H)} \leq R  \right \} $, then it is immediate to see that
\[
\| f(T) \|_{\mathcal{L}(H)} \leq \| f \|_R \,,
\]
and therefore $ \overline{ \mathcal{A} \langle X \rangle }  \subset L^{\infty} \left( B_R^{  [\mathfrak{sa} \mathcal{L}(H) ]^d } (0) \right) $ via the following simple observation
\[
\| f \|_{L^{\infty} \left( B_R^{  [\mathfrak{sa} \mathcal{L}(H) ]^d } (0) \right)} :=  \sup_{T \in B_R^{  [\mathfrak{sa} \mathcal{L}(H) ]^d } (0) } \| f(T) \|_{\mathcal{L}(H)}  \leq \| f \|_R  < \infty \,.
\]

Consider a state $\tau$ of the $C^*$-algebra, $\mathcal{L}(H)$. When $n := \text{dim}_{\mathbb{R}}(H) = 2 \, \text{dim}_{\mathbb{C}}(H)  < \infty$, from the fact that all irreducible states are vector states, we can directly check the following representation of $\tau$ via the Choquet theorem:
\begin{eqnarray*}
\tau :  \mathcal{L}(H) &\rightarrow& \mathbb{C} \\
\tau (X) &=& \int_{  P(H) } \langle w,  X \, w \rangle_{H} \, \mu_\tau ( d w)  \,,
\end{eqnarray*}
where $\mu_\tau \in \mathcal{M} ( P(H) ) $ is a Borel probability measure on $ P(H)  := \partial B_1^{H} (0) / \{ \pm 1 \}$; and we denote  $ \langle w,  X \, w \rangle_{H}  :=  \langle r,  X \, r \rangle_{\tilde{H}} $ whenever $r \in w$.
When $n = \text{dim}_{\mathbb{R}}(H) $ can be possibly $ \infty$, we may on the other hand consider the GNS representation of $\mathcal{L}(H)$ with $\tilde{H} := \oplus_{l \in \mathcal{PS}(\mathcal{L} (H))} \, \mathcal{L}(H) / ( l(a^*a)  = 0)  $ (which is not separable), where we now denote $ \mathcal{PS}(\mathcal{L} (H))$ the set of pure states of $\mathcal{L} (H)$.  With that, we obtain the following representation of $\tau$ via the Choquet theorem:
\begin{eqnarray*}
\tau :  \mathcal{L}(H) &\rightarrow& \mathbb{C} \\
\tau (X) &=& \int_{  P(\tilde{H}) } \langle \xi,  X \, \xi \rangle_{\tilde{H}} \, \mu_\tau( d \xi)  \,,
\end{eqnarray*}
where $\mu_\tau \in \mathcal{M} ( P(\tilde{H}) ) $ is a Radon probability measure on $ P(\tilde{H})  := \partial B_1^{\tilde{H}} (0) / \{ \pm 1 \}$ (which is not a complete separable metric space) with the support $\text{spt} (\mu_\tau)$ containing only the coordinate vectors; and we again denote  $ \langle \xi,  X \, \xi \rangle_{\tilde{H}}  :=  \langle v,  X \, v \rangle_{\tilde{H}} $ whenever $v \in \xi$.

We notice that this representation of $\tau$ coincide with the previous representation of $\tau$ when $n < \infty$ in a canonical manner:  $v := \oplus_{l \in \mathcal{PS}(\mathcal{L} (H))} \, a_l \, \xi_{l} \in \xi \in P(\tilde{H}) $ is a $l_0$-th coordinate vector (i.e. $a_l = 0$ for all $l \neq l_0$) if and only if $v = \pm \xi_{l_0}  = \pm ( r \otimes r ) $ for some $r \in w \in P(H) $, and in a trivial manner,
\begin{eqnarray*}
\langle \xi,  X \, \xi \rangle_{\tilde{H}} = \langle v,  X \, v \rangle_{\tilde{H}} 
& = &  \sum_{ l \in \mathcal{PS}(\mathcal{L} (H)) }  a_l^2 \langle \xi_l,  X \, \xi_l \rangle_{\mathcal{L}(H) / ( l(a^*a)  = 0)  } \\
& = & \langle \xi_{l_0},  X \, \xi_{l_0} \rangle_{\mathcal{L}(H) / ( l_0(a^*a)  = 0)  } \\
& = & \langle (r \otimes r) ,  X \, (r \otimes r) \rangle_{\mathcal{L}(H) / ( l_0(a^*a)  = 0)  } \\
& = & \langle r , X \, r \rangle_{H}    \langle r , r \rangle_{H} \\
& = & \langle r , X \, r \rangle_{H}  \\
& = & \langle w , X \, w \rangle_{H}  \,.
\end{eqnarray*}

Hence,  when $n = \text{dim}_{\mathbb{R}}(H) = 2 \, \text{dim}_{\mathbb{C}}(H) < \infty$, there is a canonical map between the coordinate vectors of $P(\tilde{H})$ and the elements on $P(H)$, and the construction of $\tilde{H}$ is unnecessary (especially noticing that $\tilde{H}$ is way bigger).
However, when $n = \text{dim}_{\mathbb{R}}(H) = \infty$, we cannot construct this map, as it is well-known that there are states $l$ such that $ l( \mathcal{K}(H)) = 0$, where $\mathcal{K}(H)$ is the space of compact operators.  
(Meanwhile, as a possibly slightly confusing remark,
if we start with an arbitrary $C^*$-algebra, we may take $H$ itself as the GNS representation of the $C^*$-algebra right from the start, and a state over that $C^*$-algebra can already be represented as an integral previously described (noting $n = \text{dim}_{\mathbb{R}}(H) = 2 \, \text{dim}_{\mathbb{C}}(H) = \infty$), i.e.,
\[ \tau (X) = \int_{  P(H) } \langle w,  X \, w \rangle_{H} \, \mu_\tau ( d w) \,,  \]
without the need to appeal to the construction of another representation $\tilde{H}$ from $\mathcal{L}(H)$.)

Motivated by the previous discussion as well as by the commutative case, let us \textbf{define} the following:

\begin{definition}
The space $\tilde{\mathcal{M}} ( \overline{ \mathcal{A} \langle X \rangle } ) $ contains all the linear functionals  $\nu$ over $ \overline{ \mathcal{A} \langle X \rangle } =  \overline{ \mathcal{A} \langle X \rangle }^{\|\cdot \|_R} $ such that they can be represented as
\begin{eqnarray*}
\nu : \overline{ \mathcal{A} \langle X \rangle } & \rightarrow& \mathbb{C} \\
\nu(f) &=& \int_{ [ \mathfrak{sa} \, \mathcal{L} (H) ]^d  \times  P (\tilde{H}) }   \, \langle \xi,  f(X) \, \xi \rangle_{\tilde{H}} \, \mu_{\nu} (d \xi, d X) \,,
\end{eqnarray*}
for some $\mu_{\nu} \in \mathcal{M} \left(   B_R^{  [\mathfrak{sa} \mathcal{L}(H) ]^d } (0) \times P(\tilde{H})  \right) $ as a finite Radon measure.   The space $\tilde{\mathcal{P}} ( \overline{ \mathcal{A} \langle X \rangle } ) \subset \tilde{\mathcal{M}} ( \overline{ \mathcal{A} \langle X \rangle } ) $ contains all the states $\tau$ over $ \overline{ \mathcal{A} \langle X \rangle }  $ with this integral representation with $\mu_{\tau} \in \mathcal{P} \left(   B_R^{  [\mathfrak{sa} \mathcal{L}(H) ]^d } (0) \times P(\tilde{H})  \right) $ being a Radon probability measure.
\end{definition}

We note our assumption of $\mu_{\nu} $ being Radon is stronger than that being Borel as it is in general not true that $ P (\tilde{H}) $ is separable.
We also notice that any $\tau$ in $\tilde{\mathcal{P}} ( \overline{ \mathcal{A} \langle X \rangle } )$ is indeed a state, because for any $f \in \mathfrak{sa} \, \overline{ \mathcal{A} \langle X \rangle }  $ such that $f \geq 0$, we have $\tau (f) \geq 0$ and $\tau (1) = 1$ whenever $\mu_{\tau} $ is a probability measure.
Moreover, we notice that while it is clear that, whenever $T \in  [ \mathfrak{sa} \, \mathcal{L} (H) ]^d$ and $w \in P(H) $, the state $f \mapsto   \langle w,  f (T)  \, w \rangle_{H} $ is a pure state of $\overline{ \mathcal{A} \langle X \rangle }$, it is unclear to the author that these are the only pure states in general without further assumption on $\tilde{H}$, and therefore it is unknown that the elements of $\tilde{\mathcal{P}} ( \overline{ \mathcal{A} \langle X \rangle } ) $ are the only states.
In this work, we only restrict our attention to this subset of states $\tilde{\mathcal{P}} ( \overline{ \mathcal{A} \langle X \rangle } ) $.

\begin{remark}
Recall that for any (continuous) linear functional $\nu \in \overline{ \mathcal{A} \langle X \rangle }^*$, 
\[
  \left \|  \nu  \right \|   :=   \sup_{g \in  \overline{ \mathcal{A} \langle X \rangle }  } \frac{\left|  \nu (g) \right|}{ \| g \|_{R}} < \infty \,.
\]
We then have the following elementary lemma:

\begin{lemma}
If $\nu \in \tilde{\mathcal{M}} ( \overline{ \mathcal{A} \langle X \rangle } ) $, then $\| \nu \| = \nu(1) < \infty $ (thus $ \nu \in \overline{ \mathcal{A} \langle X \rangle }^*$).
\end{lemma}
\begin{proof}
We first realize that, by definition,
\[
\nu(1) = \int_{ B_R^{  [\mathfrak{sa} \mathcal{L}(H) ]^d } \times  P (\tilde{H})  }  \,   \langle \xi,  1 \, \xi \rangle_{\tilde{H}}  \, \mu_{\nu} (d \xi, d X)  = \int_{ B_R^{  [\mathfrak{sa} \mathcal{L}(H) ]^d } \times  P (\tilde{H})  } \, \mu_{\nu} (d \xi, d X)  < \infty\,.
\]
Therefore, we have
\begin{eqnarray*}
 | \nu (g) | &=&  \left| \int_{ B_R^{  [\mathfrak{sa} \mathcal{L}(H) ]^d } \times  P (\tilde{H})  } \, \langle \xi,  g(X) \, \xi \rangle_{\tilde{H}} \, \mu_{\nu} (d \xi, d X) \right| \\
&\leq& \left| \int_{ B_R^{  [\mathfrak{sa} \mathcal{L}(H) ]^d } \times  P (\tilde{H})  } \, \mu_{\nu} (d \xi, d X) \right|  \sup_{(X,\xi) \in B_R^{  [\mathfrak{sa} \mathcal{L}(H) ]^d } \times  P (\tilde{H})  } | \langle \xi,  g(X) \, \xi \rangle_{\tilde{H}} | \\
& \leq & \nu (1) \,   \sup_{X \in B_R^{  [\mathfrak{sa} \mathcal{L}(H) ]^d } } \| g(X) \|_{\mathcal{L}(H)} \\
& \leq & \nu (1) \,  \| g \|_R
\end{eqnarray*}
and the equality holds if $g (X) = 1$ .  Therefore, $ \| \nu \|  = \nu (1)  < \infty$, and the result follows.
\end{proof}
\end{remark}

Now for any $\tau \in\tilde{\mathcal{P}} ( \overline{ \mathcal{A} \langle X \rangle } ) $, with the above representation, we can extend the domain of definition of $\tau$ to
\begin{eqnarray*}
\tau :  C \left( B_R^{  [\mathfrak{sa} \mathcal{L}(H) ]^{2d} } (0)  ; [ \mathcal{L}(H) ]^d  \right) & \rightarrow&  C \left( B_R^{  [\mathfrak{sa} \mathcal{L}(H) ]^{d} } (0)  ; \mathbb{C} \right) \\
\left[ \tau(f) \right](Y) &:=& 
 \int_{ [ \mathfrak{sa} \, \mathcal{L} (H) ]^d  \times  P (\tilde{H}) }   \, \langle \xi,  f(X,Y) \, \xi \rangle_{\tilde{H}} \, \mu_{\tau} (d \xi, d X)  \,.
\end{eqnarray*}
Now, for $X$ an indeterminate in $ [ \mathfrak{sa} \, \mathcal{L} (H) ]^d $, let us define $T^*_{X} : \overline{ \mathcal{A} \langle Y \rangle }  \rightarrow \overline{ \mathcal{A} \langle X, Y \rangle }  $  as 
\[
(T^*_X g)(Y) = g (X-Y)  \,.
\]
Let us also write
\[
\mathcal{E} (X) :=  \sum_{i=1}^d  X_i^2 \,.
\]
With these settings in mind, we may now define, for any $\tau \in\tilde{\mathcal{P}} ( \overline{ \mathcal{A} \langle X \rangle } ) $, $X \in  [ \mathfrak{sa} \, \mathcal{L} (H) ]^d  $,
\begin{eqnarray*}
d (X, \tau ) &:= & \lim_{p \rightarrow \infty}  \left( \tau \left(   \left[   (T^*_X \mathcal{E} )( \cdot )   \right]^{-p}  \right) \right)^{- \frac{1}{2p}}   \\
&=& \lim_{p \rightarrow \infty}  \left[  \int_{ [ \mathfrak{sa} \, \mathcal{L} (H) ]^d  \times  P (\tilde{H}) }   \int_{\sigma(\mathcal{E}(X-Y))}  \lambda^{-p} \,  E_{\mathcal{E}(X-Y)}  \left( \xi, \xi \right) (d \lambda )  \, \mu_{\tau} (d \xi, d Y) \right]^{-\frac{1}{2p}} \\
&=&\inf_{ ( \xi, Y ) \in \, \text{spt} \, \mu_\tau }  \left \{  \sqrt{  \inf \left\{  \lambda \in  \text{spt} \left ( E_{\mathcal{E}(X-Y)}  \left( \xi, \xi \right) \right)  \right \} } \right \} \,,
\end{eqnarray*}
where we denote $E_{T}  \left(\cdot,\cdot \right) (d \lambda)$ the projection-valued measure given by the spectral decomposition of $T \in \mathfrak{sa} \, \mathcal{L}(\tilde{H})$ and we denote
$ \text{spt} \left ( E_{\mathcal{E}(X-Y)}  \left( \xi, \xi \right) \right) \subset  \sigma( \mathcal{E}(X-Y))  $ 
the support of the spectral measure for any choice of $\xi$, where for any bounded operator $T \in  \mathcal{L} (H)$, $\sigma(T)$ denotes the spectrum of $T$:
\[
\{ \lambda \in \mathbb{C} \,:\, T- \lambda I \text{ does not have a bounded inverse.}\}
\]
and we recall $\sigma(T) \subset \mathbb{R}$ whenever $T$ is self-adjoint.
%and $H_{T,\lambda}$ is the ``$\lambda$-eigenspace" of $T$ in a sense that $\tilde{H} \cong \int_{R}^{\oplus} H_{T,\lambda} \, \mu_{\lambda} (d\lambda)$ for some $\mu_{\lambda}$ supported on $\sigma(T)$.
We notice also that with this definition, whenever $X - \bigcup \{ Y : (\xi,Y ) \in \text{spt}(\mu_\tau) \} \subset \mathcal{K}(H)$ the space of compact operators, we have $d (X, \tau ) = 0$. 
Moreover, for $\tau \in\tilde{\mathcal{P}} ( \overline{ \mathcal{A} \langle X \rangle } ) $, since $\text{spt}(\mu_\tau) \subset   B_R^{  [\mathfrak{sa} \mathcal{L}(H) ]^d } (0) \times P(\tilde{H}) $, then for all $X \in B_R^{  [\mathfrak{sa} \mathcal{L}(H) ]^d } (0) $, we have $ d (X, \tau )  \leq R $.

% For notational sake, we may also extend the function $d (\cdot, \tau ) $ by zero to the whole space $[ \mathcal{L} (H) ]^d $.

As an example, if $\text{dim}_{\mathbb{R}} \, H = n < \infty$, $A \subset  [ \mathfrak{sa} \, \mathcal{L} (H) ]^d   $ and $ \tau =  \int_{ A }    \frac{1}{n |A|}  \, \text{tr} (\cdot) \,  d X$, then 
\[
d (X,  \tau )   = \inf_{Y \in A} \left\{ \sqrt{ \inf \{ \lambda \in \sigma( \mathcal{E}(X-Y))  \}  } \right\}  \,,
\]
and if 
$ \tau =   \int_{ [ \mathfrak{sa} \, \mathcal{L} (H) ]^d  }  \frac{1}{n}  \, \text{tr} (\cdot) \,  \delta_{Y_0} (d X)$, then 
\[
d (X,  \tau )   =  \sqrt{ \inf \{ \lambda \in \sigma( \mathcal{E}(X-Y_0))  \}  }   \,.
\]
Notice that the level set $\{ d (X,  \tau ) = 0 \} $ is itself of dimension $ d \left( \frac{n (n+1)}{2} - 1 \right)$.  In fact, when $n < \infty$, for any pure state $\tau := \int_{ [ \mathfrak{sa} \, \mathcal{L} (H) ]^d  }  \langle v  ,   ( \cdot  ) \, v \rangle \,   \delta_{Y_0} (d X) $, we have 
\[
\text{dim}_{\mathbb{R}} \left(  \{ d (X,   \tau )  = 0 \} \right) =  d \left( \frac{n (n+1)}{2} - 1 \right) \,.
\]
%Moreover, when $n = \infty$, we notice that 
%\[
%d (X, \tau_Y ) = 0
%\]
%if there exists $Y_0 $ with a compact $X- Y_0$ such that $ \overline{  \text{span} \{ \xi :  ( \xi, Y_0 ) \in \, \text{spt} \, \mu_{Y} \} } = H $, hence you may expect this distance does not tell the difference among compact operators.

With this in mind, given another positive linear functional $\nu$, we may define the following relative distance zeta functional $\zeta_{\tau} (s, \nu)$ for $s \in \mathbb{C}$ and $\Re (s) $ sufficiently large:

\begin{definition}
Given $\tau \in \tilde{\mathcal{P}} ( \overline{ \mathcal{A} \langle X \rangle } ) $ and $\nu \in\tilde{\mathcal{M}} ( \overline{ \mathcal{A} \langle X \rangle } ) $, the relative distance zeta functional $\zeta_{\tau} (s, \nu)$ for $s \in \mathbb{C}$ and $\Re (s) $ sufficiently large is given by
\begin{eqnarray}
\zeta_{\tau} (s, \nu)  :  \overline{ \mathcal{A} \langle X \rangle } & \rightarrow&  \mathbb{C}  \notag  \\
\left[ \zeta_{\tau} (s, \nu) \right] (g) &:=&  \nu \left( d (\cdot, \tau )^{s - d} \, g(\cdot) \right) \, . \label{zeta_non_communtative} 
\end{eqnarray}
For a given (fixed) $g$, whenever possible, we then meromorphically continue $\left[ \zeta_{\tau} (s, \nu) \right] (g)$ to a domain of definition $\Omega_{\left[ \zeta_{\tau} ( \, \cdot \, , \nu ) \right] (g)} $ such that $\{ s \in \mathbb{C} \, : \,  \Re (s) > \overline{\text{dim}}_B (\tau, \nu) \} \subset \Omega_{\left[ \zeta_{\tau} ( \, \cdot \, , \nu ) \right] (g)} \subset \mathbb{C}$, and still denote the continued function as $\left[ \zeta_{\tau} (s, \nu ) \right] (g)$.   (The fact that the functional is well defined for $ \Re (s) > \overline{\text{dim}}_B (\tau, \nu) $ is given in Lemma \ref{convergence_2}. )
\end{definition}

We note that this definition is not restricted to the case when $n := \text{dim}_{\mathbb{R}} (H) < \infty$, and the relative distance zeta functional can be explicitly written as
\begin{eqnarray*}
&  &\nu \left( d (\cdot, \tau )^{s- d } g(\cdot) \right) \\
&=&  \int_{ [ \mathfrak{sa} \, \mathcal{L} (H) ]^d  \times   P (\tilde{H}) }    \left(  \inf_{ ( \xi, Y ) \in \, \text{spt} \, \mu_{\tau} }  \bigg \{  \inf \left\{  \lambda \in  \text{spt} E_{\mathcal{E}(X-Y)}  \left( \xi, \xi \right)  \right \}  \bigg \} \right)^{\frac{s - d }{2}} \langle \xi, g(X) \xi \rangle_{\tilde{H}} \, \mu_{\nu} ( d \xi, d X)
\end{eqnarray*}
The choice of the normalization $s-d$ in the power comes from the suggestion that, when $n := \text{dim}_{\mathbb{R}} (H) < \infty$, we notice
\[
s - d = \left[ s + d \left( \frac{n (n+1)}{2} - 1 \right) \right] - \frac{d (n) (n+1)}{2} 
\]
and it cancels the dimesnional effect from
$\left\{ d (X,   \int_{ [ \mathfrak{sa} \, \mathcal{L} (H) ]^d  }  \langle \xi  ,   ( \cdot  ) \, \xi \rangle \,   \delta_{Y_0} (d X)  )  = 0 \right \} $ in an ambient space of dimension $\frac{d (n) (n+1)}{2} $.

\begin{remark}
We notice that even if $\text{dim}_{\mathbb{R}} \, H = n < \infty$, $d=1$ and $ \tau =  \int_{ [ \mathfrak{sa} \, \mathcal{L} (H) ]^d  }  \frac{1}{n}  \, \text{tr} (\cdot) \, \tilde{\mu}_\tau (d X)$,  $ \nu =  \int_{ [ \mathfrak{sa} \, \mathcal{L} (H) ]^d  }  \frac{1}{n}  \, \text{tr}  (\cdot) \, \tilde{\mu}_\nu (d X)$ for some measures $\tilde{\mu}_\tau,  \tilde{\mu}_\nu $ over $\mathcal{L}(H)$, as long as $n > 1$, we see that in general, for any choice of $g$, 
\[
\left[ \zeta_{\tau } (s, \nu) \right] (g) \neq \zeta_{\tilde{\mu}_\tau} (s, \tilde{\mu}_\nu ) \,.
\]
Therefore our definition of \eqref{zeta_non_communtative} is not (interpreted in any sense as) a distance zeta function treating $\mathcal{L}(H)$ as a metric space.
In fact, when $n$ is finite, our choice of having an infrimum over $\sigma( \mathcal{E}(X-Y)) $ instead of a supremum compensates for the dimensionality effect within one variable $X_i$ from the aforementioned dimensionality argument.  We postulate that we always have $\overline{\text{dim}}_B (\tau, \nu)$ on the strip $\{  0 \leq   \Re (s)  \leq d \}$, irrespective of whether $n$ is finite or infinite.

% Notice that this is also different from the spectral zeta function of an operator even if we take $\tilde{\mu}_\tau = \delta_{X_0}$ for obvious reason.
\end{remark}

From now on, let us always write a $\delta$-tubular-neighborhood linear functional of $\tau \in \tilde{\mathcal{P}} ( \overline{ \mathcal{A} \langle X \rangle } ) $ with respect to $\nu \in \tilde{\mathcal{M}} ( \overline{ \mathcal{A} \langle X \rangle } ) $ as
\[
\nu^{\tau,\delta} :=  \nu \circ  \text{mul}_{\chi_{ \{ X \in [ \mathfrak{sa} \, \mathcal{L} (H) ]^d : d (X, \tau )  \leq t  \}  } } 
\]
We can define the relative lower $s$-dimensional Minkowski content of $\tau$ with respect to $\nu$ for $s \geq 0 $:
\[
\mathcal{M}^s_*(\tau, \nu ) := \liminf_{ \delta \rightarrow 0^+} \,  \delta^{s-d}  \,  \| \nu^{\tau,\delta}  \|      \,,
\]
the relative upper $s$-dimensional Minkowski content of $\tau$ with respect to $\nu$ for $s \geq 0$:
\[
\mathcal{M}^{*\,s}(\tau, \nu ) := \limsup_{\delta \rightarrow 0^+} \, \delta^{s-d}  \,  \|  \nu^{\tau,\delta}   \|  \,,
\]
the relative lower box counting dimension of $\tau$ with respect to $\nu$:
\[
\underline{\text{dim}}_B (\tau, \nu) = \inf \{ s \geq 0 \, , \,  \mathcal{M}^s_*(\tau, \nu )  = 0 \} = \sup \{ s \geq 0 \, , \,  \mathcal{M}^s_*(\tau, \nu )  = \infty \}.
\]
the relative upper box counting dimension of $\tau$ with respect to $\nu$:
\[
\overline{\text{dim}}_B (\tau, \nu) = \inf \{ s \geq 0 \, , \, \mathcal{M}^{* \, s}(\tau, \nu )  = 0 \} = \sup \{ s \geq 0 \, , \, \mathcal{M}^{* \, s}(\tau, \nu )  = \infty \}.
\]
If $\underline{\text{dim}}_B (\tau, \nu) = \overline{\text{dim}}_B (\tau, \nu)  $, we write $\text{dim}_B (\tau, \nu) $ the relative box counting dimension of $\tau $ with respect to $\nu$.    $\tau$ is relatively Minkowski nondegenerate with respect to $\nu$ if 
\[
0 <  \mathcal{M}^{\text{dim}_B (\tau, \nu)}_*(\tau, \nu )  \leq  \mathcal{M}^{*\,\text{dim}_B (\tau, \nu)}(\tau, \nu )  < \infty \,,
\]
and the $\text{dim}_B (\tau, \nu)$-dimensional relative Minkowski content of $\tau$ with respect to $\nu$ is 
\[
\mathcal{M}^{\text{dim}_B (\tau, \nu)}(\tau, \nu )  := \mathcal{M}^{\text{dim}_B (\tau, \nu)}_*(\tau, \nu )  = \mathcal{M}^{*\,\text{dim}_B (\tau, \nu)}(\tau, \nu ) 
\]
if they are so equal. 

With these, we would like to consider the relative distance zeta functional of $\tau \in  \tilde{\mathcal{P}} ( \overline{ \mathcal{A} \langle X \rangle } ) $ by a chosen $\nu \in  \tilde{\mathcal{M}} ( \overline{ \mathcal{A} \langle X \rangle } ) $,  i.e.
\[
s \mapsto [ \zeta_{\tau} (s, \nu) ]  \,.
\]
and check whether the above is well-defined and holomorphic with respect to $s$ whenever $\Re (s) >  \overline{\text{dim}}_B (\tau, \nu)$, i.e. for all $g \in \overline{ \mathcal{A} \langle X \rangle }  $, we have
\[
\{ s \in \mathbb{C} \, : \,  \Re (s) > \overline{\text{dim}}_B (\tau, \nu) \} \subset \Omega_{\left[ \zeta_{\tau} ( \, \cdot \, , \nu ) \right] (g)} \,.
\]

In fact, similar to  \cite{LRZ_book,LRZ,LRZ3,LRZ5,LRZ6a}, we have the following lemma.

\begin{lemma} \label{convergence_2}
If  $\tau \in \tilde{\mathcal{P}} ( \overline{ \mathcal{A} \langle X \rangle } ) $ and $\nu \in \tilde{\mathcal{M}} ( \overline{ \mathcal{A} \langle X \rangle } )  $, then
$\zeta_{\tau} (s, \nu ) $ is a well-defined bounded linear functional and is holomorphic with respect to $s$ whenever $\Re (s) > \overline{\text{dim}}_B (\tau, \nu)$, with some $\delta > 0$ such that
\begin{eqnarray}
[ \zeta_{\tau} (s, \nu) ] (g) =   \delta^{s-d}   \nu (g)  -  (s-d) \int_{0}^{\delta}  t^{s-d-1} \,   \nu^{\tau,t}  (g) \, dt
\label{zeta_non_communtative_tube} 
\end{eqnarray}
for any  $g \in \overline{ \mathcal{A} \langle X \rangle }  $, and 
\begin{eqnarray}
[ \partial_s  \zeta_{\tau} (s, \nu) ] (g) =  [ \zeta_{\tau} (s, \nu) ] \left(  \log(d (\cdot, \tau ))  g (\cdot ) \right)  \,.
\label{zeta_non_communtative_derivative} 
\end{eqnarray}
Furthermore, if the relative box counting dimension of $\tau$ with respect to $\nu$, i.e. $\text{dim}_B (\tau, \nu) $, exists and $\text{dim}_B (\tau, \nu) < d$, and moreover, if $ \mathcal{M}^{\text{dim}_B (\tau, \nu)}_*(\tau, \nu )  > 0  $, then $ \| \zeta_{\tau} (s, \nu ) \|  \rightarrow \infty $ when $s \rightarrow \text{dim}_B (\tau, \nu)^+$ for $s \in \mathbb{R}$.
\end{lemma}

\begin{remark}
In the spirit of  \cite{LRZ_book,LRZ,LRZ3,LRZ5,LRZ6a}, the map $ (s, \nu) \mapsto \left(  g \mapsto  \int_{0}^{\delta}  t^{s-d-1} \,   \nu^{\tau,t}  (g) \, dt \right) $ in  \eqref{zeta_non_communtative_tube} can also be referred to as the relative tube zeta functional of $\tau $.
\end{remark}

\begin{proof}
First notice that since $\tau \in\tilde{\mathcal{P}} ( \overline{ \mathcal{A} \langle X \rangle } ) $, for all $X \in B_R^{  [\mathfrak{sa} \mathcal{L}(H) ]^d } (0) $, we have $ d (X, \tau )  \leq R $. Therefore $\text{spt} (\mu_\nu) \subset   B_R^{  [\mathfrak{sa} \mathcal{L}(H) ]^d } (0) \times P(\tilde{H})  \subset  \{ (X, \xi) \in [ \mathfrak{sa} \, \mathcal{L} (H) ]^d  \times   P (\tilde{H}) : d (X, \tau )  \leq \delta  \} $ for some small $ \delta > 0$.
It is also handy to see that $ \| \nu^{\tau,t}  \|  \leq \left \|  \nu  \right \|  = \nu(1) < \infty$ for all $t > 0$ since $\nu \in \tilde{\mathcal{M}} ( \overline{ \mathcal{A} \langle X \rangle } )  \subset \overline{ \mathcal{A} \langle X \rangle }^* $.  Therefore, with the finiteness of $\mu_{\nu}$, we have
\begin{eqnarray*}
&   &  \int_{ [ \mathfrak{sa} \, \mathcal{L} (H) ]^d  \times   P (\tilde{H}) }    \left( d (X, \tau )\right)^{ s - d } \langle \xi, g(X) \xi \rangle_{\tilde{H}} \, \mu_{\nu} ( d \xi, d X)  \\
& =&  (s-d) \int_{0}^{\delta}  t^{s-d-1}  
\int_{ \{ d (X, \tau ) > t  \} }  \langle \xi, g(X) \xi \rangle_{\tilde{H}} \, \mu_{\nu} ( d \xi, d X)  dt \\
& =&   \delta^{s-d}   \nu (g)  -  (s-d) \int_{0}^{\delta}  t^{s-d-1} \,  \nu \left( \chi_{ \{ X \in [ \mathfrak{sa} \, \mathcal{L} (H) ]^d : d (X, \tau )  \leq t  \}  }(\cdot) \, g (\cdot) \right) \,  dt \\
& =&   \delta^{s-d}   \nu (g)  -  (s-d) \int_{0}^{\delta}  t^{s-d-1} \,   \nu^{\mu,t}  (g) \, dt
\end{eqnarray*}
where the last line is finite by definition whenever $\Re (s) > \overline{\text{dim}}_B (\tau, \nu)$, noticing that
\[ 
\mathcal{M}^{*\, \frac{\Re (s) + \overline{\text{dim}}_B (\tau, \nu) }{2} }(\tau, \nu ) =  \limsup_{ \delta \rightarrow 0^+} \,  \delta^{ \frac{\Re (s) + \overline{\text{dim}}_B (\tau, \nu) }{2} -d}  \,  \| \nu^{\tau,\delta}  \|  = 0 \,,
\]
implies the existence of $\delta_0 > 0$ such that for all $t < \delta_0$
\[
t^{  \frac{\Re (s) + \overline{\text{dim}}_B (\tau, \nu) }{2} -d}  \| \nu^{\tau,t}  \|  <  1
\]
and hence,
\begin{eqnarray*}
 \left | \int_{0}^{\delta_0}  t^{s-d-1} \,   \nu^{\tau,t}  (g) \, dt \right| \leq  \int_{0}^{\delta_0}  t^{  \frac{\Re (s) - \overline{\text{dim}}_B (\tau, \nu) }{2} -1} \,
\left( t^{  \frac{\Re (s) + \overline{\text{dim}}_B (\tau, \nu) }{2} -d}   \| \nu^{\tau,t}  \|  \| g \|_R \right )\, dt  < \infty \,,
\end{eqnarray*}
and
\[
\| \zeta_{\tau} (s, \nu) \| < \infty\,.
\]
Next, we may check via a direct series expansion that for any fixed $s$ such $\Re(s) > \overline{\text{dim}}_B (\tau, \nu)$, and for all $h \in \mathbb{C}$ such that $|h| < \frac{1}{2} \left( \Re(s) - \overline{\text{dim}}_B (\tau, \nu) \right) $,
\begin{eqnarray}
& & \left | \frac{ [\zeta_{\tau} (s + h, \nu ) ] (g) - [ \zeta_{\tau} (s , \nu ) ] (g) }{h}  -  [ \zeta_{\tau} (s, \nu) ] (  \log(d (\cdot, \tau ))  g (\cdot ))   \right |  \notag \\
&\leq& 
\int_{ [ \mathfrak{sa} \, \mathcal{L} (H) ]^d  \times   P (\tilde{H}) }  \left|  \frac{d (X, \tau )^h - 1  }{h} - \log \left( d (X, \tau ) \right)  \right|   d (X, \tau ) ^{ \Re(s) - d } \langle \xi, g(X) \xi \rangle_{\tilde{H}} \, \mu_{\nu} ( d \xi, d X)  \notag  \\
&\leq&  C
\int_{ [ \mathfrak{sa} \, \mathcal{L} (H) ]^d  \times   P (\tilde{H}) }  
|h| \left( \log \left( d (X, \tau ) \right) \right)^2 \left( \sum_{k=0}^{\infty} \frac{|\log(d (X, \tau ) ) h|^k}{k!} \right) d (X, \tau ) ^{ \Re(s) - d } \langle \xi, g(X) \xi \rangle_{\tilde{H}} \, \mu_{\nu} ( d \xi, d X)  \notag  \notag  \\
&\leq&  C |h|
\int_{ [ \mathfrak{sa} \, \mathcal{L} (H) ]^d  \times   P (\tilde{H}) }  
 \left( \log \left( d (X, \tau )  \right) \right)^2  d (X, \tau ) ^{ \Re(s) - d - |h| } \langle \xi, g(X) \xi \rangle_{\tilde{H}} \, \mu_{\nu} ( d \xi, d X)  \notag  \\
& \leq & C |h| [ \zeta_{\tau} (\Re(s) - 2 |h|, \nu) ] \left(  \log(d (\cdot, \tau ))^2  d (\cdot, \tau )^{|h|} g (\cdot ) \right)  \\
& \leq & C |h|  \|  \zeta_{\tau} (\Re(s) - 2 |h|, \nu) \|  \| g \|_{R}  \notag  \\
&=&  O( |h| \| g \|_R) \,. \label{holomorphic_argument}
\end{eqnarray}
Now we come to verify the last part of the theorem.  We first realize, again from the fact that  $ \| \nu^{\tau,t}  \|  \leq \left \|  \nu  \right \| < \infty$, as well as $\zeta_{\tau} (s, \nu ) $ is positive and $ \|  \zeta_{\tau} (s, \nu ) \| < \infty $ whenever  $s$ is a real number with $s > \text{dim}_B (\mu, \nu) $, we have, for all $\delta > 0$ and $s > \text{dim}_B (\mu, \nu) $,
\[
\|\zeta_{\tau} (s, \nu )  \| = [\zeta_{\tau} (s, \nu )] (1) \, , \,  \| \nu \|  = \nu(1)  \, \text{ and } \,  \|  \nu^{\tau,\delta}   \| =  \nu^{\tau,\delta} (1) \,.
\]
Hence, from the definition of
\[
0 < \mathcal{M}^{\text{dim}_B (\tau, \nu)}_*(\tau, \nu ):= \liminf_{\delta \rightarrow 0^+}  \, \delta^{ \text{dim}_B (\tau, \nu)- d} \,  \|  \nu^{\tau,\delta}   \|  =  \liminf_{\delta \rightarrow 0^+}  \, \delta^{ \text{dim}_B (\tau, \nu)- d} \, \nu^{\tau,\delta} (1)   \,,
\]
we have the existence of $C > 0$ such that 
\[
 \nu^{\tau,\delta} (1) \geq C \delta^{ d- \text{dim}_B (\tau, \nu)} \,
\]
and that leads to the observation that if $d \geq s > \text{dim}_B (\tau, \nu) $, we have
\begin{eqnarray*}
\infty > \|  \zeta_{\tau} (s, \nu )  \| =  [ \zeta_{\tau} (s, \nu ) ](1) &=& \delta^{s-d}   \nu (1)  -  (s-d) \int_{0}^{\delta}  t^{s-d-1} \, \nu^{\tau,t}(1)  dt \\
&\geq &  - C (s-d)  \int_{0}^{\delta}  t^{s- \text{dim}_B (\tau, \nu)-1}   dt \\
&=&  - C (s-d)  \frac{ \delta^{s- \text{dim}_B (\tau, \nu)} }{  s- \text{dim}_B (\tau, \nu)   }  \,.
\end{eqnarray*}
This implies that $ \|  \zeta_{\tau} (s, \nu )  \| =  [ \zeta_{\tau} (s, \nu ) ](1)  \rightarrow \infty $ when $s \rightarrow \text{dim}_B (\tau, \nu)^+$.
\end{proof}

\begin{remark}
Given an arbitrary  $\nu \in \tilde{\mathcal{M}} ( \overline{ \mathcal{A} \langle X \rangle } ) $, and for a fixed $\delta > 0$, the $\delta$-tubular-neighborhood linear functional of $\tau \in  \tilde{\mathcal{P}} ( \overline{ \mathcal{A} \langle X \rangle } ) $ with respect to $\nu$ always satisfies
$\text{spt} ( \mu_{\nu^{\tau,\delta} } ) \subset  \{ (X, \xi) \in [ \mathfrak{sa} \, \mathcal{L} (H) ]^d  \times   P (\tilde{H}) : d (X, \tau )  \leq t  \} $, and therefore $\zeta_{\tau} (s, \nu^{\tau,\delta}  ) $ is well-defined and holomorphic whenever $\Re (s) > \overline{\text{dim}}_B (\mu, \nu^{\tau,\delta} )$.  Also, the above representation of $\zeta_{\tau} (s, \nu^{\tau,\delta}) $ in the lemma shows that, given  $g \in \overline{ \mathcal{A} \langle X \rangle }  $,
\[
 h_{\delta_1, \delta_2}(s) := [ \zeta_{\tau} (s, \nu^{\tau,\delta_2}) - \zeta_{\tau} (s, \tau^{\mu,\delta_1}) ] (g)
\]
can be extended to an entire function and hence
\[
\overline{\text{dim}}_B (\tau, \nu^{\tau,\delta_1} ) = \overline{\text{dim}}_B (\tau, \nu^{\tau,\delta_2} )
\]
for any small $\delta_1 \neq \delta_2$ larger than $0$.  This is again coming from the fact that
\[
h_{\delta_1, \delta_2} (s) =  ( \delta_2^{s-d} -  \delta_1^{s-d}  )   \nu (g) -   (s-d) \int_{\delta_1}^{\delta_2}  t^{s-d-1} \, \nu^{\tau,t}(g)  dt 
\]
and hence
\[
| h_{\delta_1, \delta_2} (s) -  ( \delta_2^{s-d} -  \delta_1^{s-d}  )   \nu (g)   | \leq 
\begin{cases}
 |s-d| \nu^{\tau,\delta_2} (g)   \left( \frac{ \delta_2^{\Re(s)-d} -  \delta_1^{ \Re(s) -d}  }{ \Re(s) - d } \right)    &  \text{ when } \Re(s)  \neq d \,, \\
 |s-d| \nu^{\tau,\delta_2} (g)   \left( \log(\delta_2) - \log(\delta_1)   \right)  &  \text{ when } \Re(s) = d \,,
\end{cases} 
\]
which leads to the integral being an entire function for $s \in \mathbb{C}$, again following the same holomorphicity argument as in  Lemma \ref{convergence_2}
\end{remark}

We would also like to define the relative complex dimensions of $\tau$ with respect to $\nu$ and fractality of $\tau$:

\begin{definition}
The (set-valued) relative complex dimension function of $\tau \in  \tilde{\mathcal{P}} ( \overline{ \mathcal{A} \langle X \rangle } ) $  with respect to $\nu \in \tilde{\mathcal{M}} ( \overline{ \mathcal{A} \langle X \rangle } ) $  is defined as a map that brings $g \in \overline{ \mathcal{A} \langle X \rangle }  $ to the set of poles of $[ \zeta_{\tau} (s,  \nu ) ](g) $, i.e.
\begin{eqnarray*}
\mathcal{D}_{\mathbb{C}} ( \tau, \nu) : \overline{ \mathcal{A} \langle X \rangle } & \rightarrow& 2^{\mathbb{C}}  \\
\mathcal{D}_{\mathbb{C}} ( \tau, \nu) [g]  &=& \left \{ s \in \Omega_{\left[ \zeta_{\tau} ( \, \cdot \, , \nu ) \right] (g)}  \, : \,   [ \zeta_{\tau} (s, \nu )](g)  = \infty  \right \} \,.
\end{eqnarray*}
The relative complex dimension of $\tau$ with respect to $\nu$ is then given by
\[
 \mathcal{D}_{\mathbb{C}} ( \tau, \nu) := \bigcup_{g \in \overline{ \mathcal{A} \langle X \rangle } } \mathcal{D}_{\mathbb{C}} ( \tau, \nu) [g] \,. 
\]
$\tau$ is a fractal (state) with respect to $\nu$ if 
\[
 \mathcal{D}_{\mathbb{C}} ( \tau, \nu)  \backslash \mathbb{R} \neq \emptyset \,.  
\]
\end{definition}

\noindent We point out that it is now unclear if the relative complex dimension  $\mathcal{D}_{\mathbb{C}} ( \tau, \nu)$, as a union of $\mathcal{D}_{\mathbb{C}} ( \tau, \nu) [g] $ over all $g \in  \overline{ \mathcal{A} \langle X \rangle } $,  is in general a discrete set, unlike in the usual theory \cite{LRZ_book}.

%Again, it remains to check that the above definition is well-defined and is independent of the choice of $\delta$.

\begin{remark}
Given an arbitrary  $\nu \in \tilde{\mathcal{M}} ( \overline{ \mathcal{A} \langle X \rangle } )  $  and for any small $\delta_1 \neq \delta_2$ larger than $0$, from the previous remark, we can quickly check that
\[
\mathcal{D}_{\mathbb{C}} ( \mu, \nu^{\tau,\delta_1} ) [\cdot] = \mathcal{D}_{\mathbb{C}} ( \mu, \nu^{\tau,\delta_2} ) [\cdot]  \, ,
\]
and therefore,
\[
\mathcal{D}_{\mathbb{C}} ( \mu, \nu^{\tau,\delta_1} )  = \mathcal{D}_{\mathbb{C}} ( \mu, \nu^{\tau,\delta_2} )  \,.
\]
\end{remark}

\section{Geometric and symmetry properties and decomposition rules for the noncommutative relative distance zeta functional} \label{sec_3}

\subsection{Geometric and symmetry properties}  \label{sec_3_1}

Let us consider a Radon map $F : [\mathcal{L} (H) ]^d  \rightarrow  [\mathcal{L} (H) ]^d  $ such that
\begin{eqnarray}
 F ( [\mathfrak{sa} \, \mathcal{L} (H) ]^d ) =  [ (\mathfrak{sa} \, \mathcal{L} (H) ]^d \quad \text{ and } \quad F ( B_R^{  [\mathfrak{sa} \mathcal{L}(H) ]^d } (0) ) \subset B_{ C_F  R}^{  [\mathfrak{sa} \mathcal{L}(H) ]^d } (0)
\label{condition0}
\end{eqnarray}
for some $ C_F < \infty$.  With this, we may denote the pullback of a function $g \in \overline{ \mathcal{A} \langle X \rangle }^{\|\cdot \|_{C_F R}} $ by $F$, $F^* g   \in \overline{ \mathcal{A} \langle X \rangle }^{\|\cdot \|_{R}}  $, as 
\[
[F^*g ](X)  := g ( F(X)) \quad \text{  for all  } X \in [\mathcal{L} (H) ]^d  \,,
\]
as well as the pushforward of a linear functional $\nu \in  \tilde{\mathcal{M}} \left( \overline{ \mathcal{A} \langle X \rangle }^{\|\cdot \|_{R}} \right) $ by $F$,  $F_* \nu \in  \tilde{\mathcal{M}} \left(   \overline{ \mathcal{A} \langle X \rangle }^{\|\cdot \|_{C_F R}} \right) $  , as 
\[
[F_* \nu ] (g) := \nu \left ([F^*(g) ] \right)  \text{  for all  } g  \in \overline{ \mathcal{A} \langle X \rangle }^{\|\cdot \|_{L_F R}} \,,
\]
which we now notice is well defined as $F^*(g)$ is a Radon function and $\mu_{\nu}$ is a Radon measure.  Whenever the context is clear, we again do not specify whether the closure is taken over $\|\cdot \|_{R}$ or $\|\cdot \|_{C_F R}$.
With the above definitions, we readily obtain the following lemma about the transformation rules of the noncommutative relative distance zeta functional.

\begin{lemma}
The following two statements hold:
\begin{enumerate}
\item
Let $\tau  \in  \tilde{\mathcal{P}} ( \overline{ \mathcal{A} \langle X \rangle } ) $ and  $\nu \in \tilde{\mathcal{M}} ( \overline{ \mathcal{A} \langle X \rangle } )  $ be given.  Assume $F$ satisfies \eqref{condition0} and is such that there exists $L_F^2 > 0$ with
\begin{eqnarray}
\mathcal{E} (F(X) - F(Y))  = L_F^2 \, \mathcal{E} (X - Y)  \quad \text{ for all } (X,Y) \in  [\mathfrak{sa} \, \mathcal{L} (H) ]^{2d} \,,
\label{condition1}
\end{eqnarray}
then the following transformation rules hold
\begin{eqnarray}
F_* \left[ d ( \cdot , F_* \tau ) \right]  = L_F \, d ( \cdot , \tau ) \quad \text{ and } \quad  \zeta_{F_* \tau} (s, F_* \nu)  =  L_F^{s - d} \, F_* [ \zeta_{\tau} (s, \nu) ]  \,.
\label{transformation_rule}
\end{eqnarray}
\item
Suppose  $\tau  \in  \tilde{\mathcal{P}} ( \overline{ \mathcal{A} \langle X \rangle } ) $ is furthermore tracial. Assume $F$ satisfies \eqref{condition0} and is such that there exist $L_F^2 > 0$ and a $\mu_\tau^{\otimes^2}$-measurable map $G: [\mathcal{L} (H) ]^{2d} \rightarrow \mathcal{L} (H) $ satisfying
\begin{eqnarray}
\mathcal{E} (F(X) - F(Y)) = L_F^2 \, G(X,Y) \, \mathcal{E} (X - Y)\,  [ G(X,Y) ]^{-1} \quad  \text{ for all } (X,Y) \in  [\mathfrak{sa} \, \mathcal{L} (H) ]^{2d}  \,,
\label{condition2}
\end{eqnarray}
then \eqref{transformation_rule} holds.
\end{enumerate}

\end{lemma}

\begin{proof}
We notice that, from definition, we have
\begin{eqnarray*}
d (F(X), F_* \tau ) &=&  \lim_{p \rightarrow \infty}  \left( F_* \tau \left(   \left[   (T^*_{F(X)} \mathcal{E} )( \cdot )   \right]^{-p}  \right) \right)^{- \frac{1}{2p}}  \\
& = &  \lim_{p \rightarrow \infty}  \left(  \tau \left( F^*  \left[   (T^*_{F(X)} \mathcal{E} )( \cdot )   \right]^{-p}  \right) \right)^{- \frac{1}{2p}}  \\
& = &   \lim_{p \rightarrow \infty}  \left(  \tau \left(  \left[   (T^*_{F(X)} \mathcal{E} )( F ( \cdot ) )   \right]^{-p}  \right) \right)^{- \frac{1}{2p}}  \,.
\end{eqnarray*}
\begin{enumerate}
\item
Whenever \eqref{condition1} holds, we have $(T^*_{F(X)} \mathcal{E} )( F ( Y ) ) = L_F^2 \,(T^*_{X} \mathcal{E} )(Y ) $ for $ X,Y \in  [\mathfrak{sa} \, \mathcal{L} (H) ]^{2d}$. Hence for all $ X\in  [\mathfrak{sa} \, \mathcal{L} (H) ]^{2d}$, we have
\begin{eqnarray*}
d (F(X), F_* \tau ) =  \lim_{p \rightarrow \infty}  \left(  \tau \left(  \left[ L_F^2  (T^*_{X} \mathcal{E} )( \cdot )   \right]^{-p}  \right) \right)^{- \frac{1}{2p}}  = L_F \, d (X, \tau )  \,,
\end{eqnarray*}
and therefore by definition, we have, for all $g \in \overline{ \mathcal{A} \langle X \rangle }^{\|\cdot \|_{C_F R}} $, that
\begin{eqnarray*}
\left[ \zeta_{F_* \tau} (s, F_* \nu) \right]  (g)
= \left[ \zeta_{F_* \tau} (s, F_* \nu) \right]  (g) 
&=&  F_*\nu \left( d ( \cdot, F_* \tau )^{s - d} \, g( \cdot ) \right)  \\
&=&  \nu \left( d ( F( \cdot) , F_* \tau )^{s - d} \, g( F( \cdot) ) \right)  \\
&=&  L_F^{s - d} \, \nu \left(  d (X, \tau )^{s - d}   \, g( F( \cdot) ) \right)  \\
&=&  L_F^{s - d} \, [ \zeta_{\tau} (s, \nu) ] ( F^* g ) \\
&=&  L_F^{s - d} \, \left( F_* [ \zeta_{\tau} (s, \nu) ] \right) (  g  ) 
\end{eqnarray*}
which results in  \eqref{transformation_rule}, proving the first part of the lemma.
\item
Suppose $\tau$ is tracial and \eqref{condition2} holds, we have $(T^*_{F(X)} \mathcal{E} )( F ( Y ) ) = L_F^2 \,    G(X,Y) \, \mathcal{E} (X - Y) \, [ G(X,Y) ]^{-1}$ for $ X,Y \in  [\mathfrak{sa} \, \mathcal{L} (H) ]^{2d}$. Hence for all $ X\in  [\mathfrak{sa} \, \mathcal{L} (H) ]^{2d}$, we have
\begin{eqnarray*}
d (F(X), F_* \tau ) &=&  \lim_{p \rightarrow \infty}  \left(  \tau \left(  \left[ L_F^2  (T^*_{X} \mathcal{E} )( \cdot )   \right]^{-p}  \right) \right)^{- \frac{1}{2p}}   \\
&=&  \lim_{p \rightarrow \infty}  \left(  \tau \left(  \left[ L_F^2 \, G(X,\cdot) \, (T^*_{X} \mathcal{E} )( \cdot ) \, [ G(X,\cdot) ]^{-1}  \right]^{-p}  \right) \right)^{- \frac{1}{2p}}   \\
&=& L_F \lim_{p \rightarrow \infty}  \left(  \tau \left(   \, [G(X,\cdot)]^{-1} \,\left[  (T^*_{X} \mathcal{E} )( \cdot ) \right]^{-p}  \, [ G(X,\cdot) ]   \right) \right)^{- \frac{1}{2p}} \\
&=& L_F \lim_{p \rightarrow \infty}  \left(  \tau \left(   \,\left[  (T^*_{X} \mathcal{E} )( \cdot ) \right]^{-p} \right) \right)^{- \frac{1}{2p}} \\
&=& L_F \, d (X, \tau )  \,,
\end{eqnarray*}
where the second last equality comes from the fact that $\tau$ is tracial.  The conclusion \eqref{transformation_rule} then follows from the same argument as above.
\end{enumerate}
\end{proof}

\begin{corollary}  \label{transformation_lemma}
Given $\tau  \in  \tilde{\mathcal{P}} ( \overline{ \mathcal{A} \langle X \rangle } ) $ and  $\nu \in \tilde{\mathcal{M}} ( \overline{ \mathcal{A} \langle X \rangle } ) $, the following transformation rules hold:
\begin{enumerate}
\item
Let $X_0 \in [\mathfrak{sa} \, \mathcal{L} (H) ]^d $ be given. Write (with an abuse of notation) $T_{X_0} :  [\mathcal{L} (H) ]^d  \rightarrow  [\mathcal{L} (H) ]^d  $, $X \mapsto  X - X_0$, then
\[
\zeta_{(T_{X_0})_* \tau} (s, (T_{X_0})_* \nu)  =  (T_{X_0})_* [ \zeta_{\tau} (s, \nu) ]  \,.
\]
\item
Let $k > 0$ be given. Write (with an abuse of notation) $ k :  [\mathcal{L} (H) ]^d  \rightarrow  [\mathcal{L} (H) ]^d  $, $X \mapsto k X$, then
\[
\zeta_{ k_* \tau} (s, k_* \nu)  = k^{s - d}  \,  k_* [ \zeta_{\tau} (s, \nu) ] \,.
\]
\item
Let $O = (O_{ij}) \in SO(d) $ be given. Write (with an abuse of notation) $ O :  [\mathcal{L} (H) ]^d  \rightarrow  [\mathcal{L} (H) ]^d  $, $(X_i )_{i=1}^d \mapsto ( \sum_{j=1}^d O_{ij} X_j  )_{i=1}^d$, then
\[
\zeta_{ O_* \tau} (s, O_* \nu)  =  \,  O_* [ \zeta_{\tau} (s, \nu) ]  \,.
\]
\item
Let $U \in U(H) $ be given. Write $ \text{inn}_{U} :  [\mathcal{L} (H) ]^d  \rightarrow  [\mathcal{L} (H) ]^d  $, $(X_i )_{i=1}^d \mapsto ( U X_i U^* )_{i=1}^d$, then whenever $\tau$ is tracial, we have
\[
\zeta_{ ( \text{inn}_{U} )_* \tau} (s, (\text{inn}_{U})_* \nu)  =  \,  (\text{inn}_{U})_* [ \zeta_{\tau} (s, \nu) ]  \,.
\]
\end{enumerate}
\end{corollary}
\begin{proof}
It is direct to notice that all the maps given are Radon, and that \eqref{condition0} holds for all the aforementioned maps, with
\[
C_{T_{X_0}} = \|X_0\| \,, \, C_{k} = k \,, \, C_{O} \leq \sqrt{d} \,, \, C_{\text{inn}_{U}} = 1.
\]
It is also strict forward to check \eqref{condition1} holds for $T_{X_0}$, $k$ and $C_O$ with
\[
L_{T_{X_0}} = 1  \,, \, L_{k} = k \,, \, L_{O} = 1
\]
and \eqref{condition2} holds for $\text{inn}_{U}$ with
\[
L_{\text{inn}_{U}}  = 1  \, , \, G(X,Y) = U
\]
\end{proof}

\begin{remark}
From the fact that $F^*(1) = 1$ for any $F$, the above gives that 
\begin{eqnarray*}
[\zeta_{(T_{X_0})_* \tau} (s, (T_{X_0})_* \nu) ] (1)  &=& [ \zeta_{\tau} (s, \nu) ](1)  \\ 
\quad [\zeta_{ k_* \tau} (s, k_* \nu) ] (1)  &=& k^{s - d}  \, [\zeta_{\tau} (s, \nu) ](1)  \\
\quad [\zeta_{ O_* \tau} (s, O_* \nu)] (1) & =&  [ \zeta_{\tau} (s, \nu) ] (1) \\
\quad [ \zeta_{ ( \text{inn}_{U} )_* \tau} (s, (\text{inn}_{U})_* \nu) ](1)  &=&   [ \zeta_{\tau} (s, \nu) ] (1) \,.
\end{eqnarray*}
which corresponds to and generalizes the more familiar transformation rules and symmetry properties in the commutative cases obtained in \cite{LRZ_book}.  They are helpful for us to obtain the functional equations in Section \ref{sec_4}.
\end{remark}

\begin{remark}
The notation of $T^*_{X} : \overline{ \mathcal{A} \langle Y \rangle }  \rightarrow \overline{ \mathcal{A} \langle X, Y \rangle }  $ with $X$ as an indeterminate in $ [ \mathfrak{sa} \, \mathcal{L} (H) ]^d$ in the previous section aligns with the notation of the pullback of a function $g$ by $T_{X_0} :  [\mathcal{L} (H) ]^d  \rightarrow  [\mathcal{L} (H) ]^d  $ for a given self-adjoint element $X_0 \in [\mathfrak{sa} \, \mathcal{L} (H) ]^d$, as we realize that $\text{eva}_{X_0} [T^*_{X} (g)] =T^*_{X_0} (g) $.
\end{remark}

\subsection{Decomposition rules}  \label{sec_3_2}

In this subsection, we would like to consider decomposition rules which help explicit computation of the noncommutative relatives distance zeta functional, where we note the second part of the lemma is inspired by Theorem 4.9 in \cite{Hoffer}.  They are helpful for us to obtain the functional equations in Section \ref{sec_4}.

\begin{lemma} \label{decomposition_lemma}
The following decomposition rules are in force.
\begin{enumerate}
\item
For all $ \tau \in \tilde{\mathcal{P}} ( \overline{ \mathcal{A} \langle X \rangle } ) $, $\nu_1, \nu_2 \in \tilde{\mathcal{M}} ( \overline{ \mathcal{A} \langle X \rangle } ) $ and $a_1, a_2 \geq 0$, it holds that
\begin{eqnarray}
\zeta_{\tau} (s, a_1 \nu_1 + a_2 \nu_2) = a_1  \zeta_{\tau} (s,  \nu_1 )  + a_2 \zeta_{\tau} (s, \nu_2) \,.
\label{decomposition1}
\end{eqnarray}

\item
Suppose $ \tau \in \tilde{\mathcal{P}} ( \overline{ \mathcal{A} \langle X \rangle } ) $ can be decomposed into $ \tau = a_1 \tau_1 + a_2 \tau_2$ with $ \tau_1, \tau_2 \in \tilde{\mathcal{P}} ( \overline{ \mathcal{A} \langle X \rangle } ) $ and $a_1 + a_2 = 1 $.  Then, for $\nu \in \tilde{\mathcal{M}} ( \overline{ \mathcal{A} \langle X \rangle } ) $, then whenever $\overline{\text{dim}}_B (\tau_2, \nu)  <  \overline{\text{dim}}_B (\tau, \nu)  $, and for all $ s \in \Omega_{\left[ \zeta_{\tau_1} ( \, \cdot \, , \nu ) \right] (g)} \bigcap \{ s \in \mathbb{C} \, : \,  \Re (s) > \overline{\text{dim}}_B (\tau_2, \nu) \} $,
\begin{eqnarray}
[\zeta_{\tau} (s, \nu)](g) =  [\zeta_{\tau_1} (s,  \nu)](g) + [h(s)](g) \,,
\label{decomposition2_1}
\end{eqnarray}
where $s \mapsto [h(s)](g)$ is a holomorphic function on $ \{ s \in \mathbb{C} \, : \,  \Re (s) > \overline{\text{dim}}_B (\tau_2, \nu) \} $.

\item
Suppose  $\nu \in \tilde{\mathcal{M}} ( \overline{ \mathcal{A} \langle X \rangle } )  $ is given, and $ \tau \in \tilde{\mathcal{P}} ( \overline{ \mathcal{A} \langle X \rangle } ) $ can be decomposed into $ \tau = a_1 \tau_1 + a_2 \tau_2$ with $ \tau_1, \tau_2 \in \tilde{\mathcal{P}} ( \overline{ \mathcal{A} \langle X \rangle } ) $ and $a_1 + a_2 = 1 $ such that for some $\epsilon > 0$,
\begin{eqnarray}
\{ X \in   [\mathfrak{sa} \, \mathcal{L} (H) ]^d \, : \, d (X , \tau_2) \leq \epsilon \}  \bigcap \text{spt} \left(  \int_{  P (\tilde{H}) } \langle \xi, g(X) \xi \rangle_{\tilde{H}} \, \mu_{\nu} ( d \xi, d X) \right) = \emptyset \,.
\label{assumption_strong}
\end{eqnarray}
Then we have, for $ s \in \Omega_{\left[ \zeta_{\tau_1} ( \, \cdot \, , \nu ) \right] (g)} $,
\begin{eqnarray}
[\zeta_{\tau} (s, \nu)](g) =  [\zeta_{\tau_1} (s,  \nu)](g) + [h(s)](g) \,.
\label{decomposition2_1}
\end{eqnarray}
where $s \mapsto [h(s)](g)$ is an entire function on $\mathbb{C}$.

\end{enumerate}

\end{lemma}

\begin{proof}
The first property is trivial via linearity that
\[
 \nu \left( d (\cdot, \tau )^{s - d} \, g(\cdot) \right)  =  a_1 \nu_1 \left( d (\cdot, \tau )^{s - d} \, g(\cdot) \right)  + a_2 \nu_2 \left( d (\cdot, \tau )^{s - d} \, g(\cdot) \right)  \,.
\]

The second and third properties come from the fact that
\begin{eqnarray*}
d(X, \tau) &=&   \inf_{ ( \xi, Y ) \in \, \text{spt} \, \mu_{\tau} }  \bigg \{  \inf \left\{  \lambda \in  \text{spt} E_{\mathcal{E}(X-Y)}  \left( \xi, \xi \right)  \right \}  \bigg \} \\
& = &  \min_{i=1,2}  \left \{  \inf_{ ( \xi, Y ) \in \, \text{spt} \, \mu_{\tau_i} }  \bigg \{  \inf \left\{  \lambda \in  \text{spt} E_{\mathcal{E}(X-Y)}  \left( \xi, \xi \right)  \right \}  \bigg \}  \right \} \\
& = &  \min_{i=1,2}  \{ d(X, \tau_i) \} \,.
\end{eqnarray*}
Now simplifying the notation and writing
\[
 [d (\cdot , \tau) ]^{-1} ( [0,t] )  :=  \{ X \in [ \mathfrak{sa} \, \mathcal{L} (H) ]^d : d (X, \tau )  \leq t  \}\,,
\]
the above property gives
\[
 [d (\cdot , \tau) ]^{-1} ( [0,t] ) =  [d (\cdot , \tau_1) ]^{-1} ( [0,t] ) \bigcup  [d (\cdot , \tau_2) ]^{-1} ( [0,t] )
\]
and therefore
\[
0 \leq \chi_{  [d (\cdot , \tau) ]^{-1} ( [0,t] )  } - \chi_{  [d (\cdot , \tau_1) ]^{-1} ( [0,t] ) } \leq \chi_{  [d (\cdot , \tau_2) ]^{-1} ( [0,t] ) }
\]
or
\[
 0 \leq \nu^{\tau,t} - \nu^{\tau_1,t}  \leq  \nu^{\tau_2,t} \,.
\]

To get the second property, we now obtain the following identity from \eqref{zeta_non_communtative_tube} that, for some $\delta > 0$, 
\begin{eqnarray*}
&  &[ \zeta_{\tau} (s, \nu) ] (g) \\
&=&   \delta^{s-d}   \nu (g)  -  (s-d) \int_{0}^{\delta}  t^{s-d-1} \,   \nu^{\tau,t}  (g) \, dt \\
&=&   [ \zeta_{\tau_1} (s,  \nu) ](g) -  (s-d) \int_{0}^{\delta}  t^{s-d-1} \,    \left( \nu^{\tau,t} - \nu^{\tau_1,t} \right)  (g) \, dt \\
&:=&   [ \zeta_{\tau_1} (s,  \nu) ](g) + [ h(s) ] (g) 
\end{eqnarray*}
to get
\[
 | [ h(s) ] (g)  | \leq |(s-d)| \int_{0}^{\delta}  t^{\Re(s)-d-1} \,   \| \nu^{\tau_2,t}  \|  \| g \|_R  \, dt  \leq \infty \,,
\]
where the right hand side is finite whenever $\Re (s) > \overline{\text{dim}}_B (\tau_2, \nu) $, noticing that
\[
t^{  \frac{\Re (s) + \overline{\text{dim}}_B (\tau_2, \nu) }{2} -d}  \| \nu^{\tau_2,t}  \|  <  1
\]
whenever $t < \delta_0$ for some $\delta_0$.  The holomorphicity of $s \mapsto [h(s)](g)$ on $ \{ s \in \mathbb{C} \, : \,  \Re (s) > \overline{\text{dim}}_B (\tau_2, \nu) \} $ now follows from the same argument as in  Lemma \ref{convergence_2}.

Now the third property comes directly from the fact that, whenever $t < \epsilon$, we have
\begin{eqnarray*}
 | \nu^{\tau_2,t} (g) |  & \leq  & \left |   \int_{  [d (\cdot , \tau_2) ]^{-1} ( [0,\epsilon] ) \times  P (\tilde{H})  } \, | \langle \xi,  g(X) \, \xi \rangle_{\tilde{H}}  | \, \mu_{\nu} (d \xi, d X)  \right | \\
&  \leq  &   \int_{ \left(  [d (\cdot , \tau_2) ]^{-1} ( [0,\epsilon] )  \bigcap \text{spt} \left(  \int_{  P (\tilde{H}) } \langle \xi, g(X) \xi \rangle_{\tilde{H}} \, \mu_{\nu} ( d \xi, d X) \right) \right)  \times  P (\tilde{H}) } \,  \mu_{\nu} (d \xi, d X)  \,   \| g \|_R \\
&  \leq  &   \int_{ \emptyset \times  P (\tilde{H}) } \,  \mu_{\nu} (d \xi, d X)  \,   \| g \|_R \\
& = & 0
\end{eqnarray*}
where we know apriori that the right hand side of the second inequality is finite because $ \left \|  \nu  \right \| < \infty $.  Therefore, we have
\[
 | [ h(s) ] (g)  | \leq |(s-d)| \int_{\epsilon}^{\delta}  t^{\Re(s)-d-1} \,   \| \nu^{\tau_2,t}  \|  \| g \|_R  \, dt  \leq 
\begin{cases}
 |s-d| \nu^{\tau,\delta_2} (g)   \left( \frac{ \delta^{\Re(s)-d} -  \epsilon^{ \Re(s) -d}  }{ \Re(s) - d } \right)    &  \text{ when } \Re(s)  \neq d \,, \\
 |s-d| \nu^{\tau,\delta_2} (g)   \left( \log(\delta) - \log(\epsilon)   \right)  &  \text{ when } \Re(s) = d \,,
\end{cases} 
\]
which again leads to the fact that $s \mapsto [h(s)](g)$ extends to an entire function for $s \in \mathbb{C}$ following the same holomorphicity argument in  Lemma \ref{convergence_2}.
Therefore the result follows.
\end{proof}

\subsection{Zeta functional over tensor states and decomposition rules}  \label{sec_3_3}

In this subsection, we first consider $\overline{\mathcal{A}(X_i)}$ where $X_i \in [\mathfrak{sa} \,  \mathcal{L} (H)]^{d_i}$ are the indeterminates for $i = 1,2$.
% We realize it is not necessary to consider $\mathcal{L} (H_i)$ for different $H_i$, as $\mathcal{L}(H_i) \subset \mathcal{L}(H_1 \oplus H_2) $, and hence we can always represent the $C^*$-algebras with the same (possibly larger) Hilbert space $H$.
Now let us recall the spatial/minimal tensor product $\overline{\mathcal{A}(X_1)} \otimes_{\min} \overline{\mathcal{A}(X_2)} $ as the norm closure of the finite sums of $f_1 \otimes f_2 \in \overline{\mathcal{A}(X_1)} \otimes \overline{\mathcal{A}(X_2)}$ in the algebraic tensor products under the norm 
\[
\| f_1 \otimes f_2 \|_{R} := \| f_1 \|_R \| f_2 \|_R \,.
\]
We notice that for any chosen $T_i \in   B_R^{  [\mathfrak{sa} \,  \mathcal{L} (H)]^{d_i} } (0)  \,, i = 1,2 $, we have
\[
\| f_1(T_1) \otimes  f_2(T_2)  \|_{\mathcal{L}(H \otimes H )}  = \| f_1(T_1)   \|_{\mathcal{L}(H )} \| f_2(T_2)  \|_{\mathcal{L}( H )}   \leq \| f_1 \|_R \| f_2 \|_R   = \| f_1 \otimes f_2 \|_{R}  \,,
\]
where the first line comes the definition of the tensor product of operators,
and therefore again
\[
\| f_1 \otimes f_2 \|_{L^{\infty} \left( B_R^{  [\mathfrak{sa} \mathcal{L}(H) ]^{d_1+ d_2} } (0) ; H \otimes H  \right)}  \leq  \| f_1 \otimes f_2 \|_{R}    < \infty \,.
\]
and hence,
\[
\overline{\mathcal{A}(X_1)} \otimes_{\min} \overline{\mathcal{A}(X_2)} \subset  L^{\infty} \left( B_R^{  [\mathfrak{sa} \mathcal{L}(H) ]^{d_1 + d_2}  } (0) ; H \otimes H \right) \,.
\]
However, we would like to emphasize that $ \overline{\mathcal{A}(X_1)} \otimes_{\min} \overline{\mathcal{A}(X_2)} $ has no relation with $ \overline{ \mathcal{A}(X_1, X_2)} $. 

Similarly, we have 
\[
{ \otimes_{\min} }_{i=1}^k  \overline{\mathcal{A}(X_i)}   \subset  L^{\infty} \left( B_R^{  [\mathfrak{sa} \mathcal{L}(H) ]^{\sum_{i=1}^k d_i}  } (0) ; H^{\otimes^k}\right) \,,
\]
However, again we would like to emphasize that $ { \otimes_{\min} }_{i=1}^k  \overline{\mathcal{A}(X_i)}  $ has no relation with $ \overline{ \mathcal{A}(X_1, ..., X_k)} $. 

Now, consider linear functionals $\nu_i \in \overline{\mathcal{A}(X_i)}^*$ where $X_i \in [\mathfrak{sa} \,  \mathcal{L} (Hi)]^{d_i}$ for $i=1,2$.
Then the tensor functional $\nu_1 \otimes \nu_2$ is a unique functional over $\overline{\mathcal{A}(X_1)} \otimes_{\min} \overline{\mathcal{A}(X_2)} $ such that for any $f_1 \otimes f_2 \in \overline{\mathcal{A}(X_1)} \otimes \overline{\mathcal{A}(X_2)}$, we have
\[
( \nu_1 \otimes  \nu_2 ) (f_1 \otimes f_2 ) :=  \nu_1 (f_1) \, \nu_2 ( f_2 )
\]
If it is furthermore true that $\tau_i \in  \tilde{\mathcal{M}} (  \overline{\mathcal{A}(X_i)} )$, then 
\begin{eqnarray*}
&  & ( \nu_1 \otimes  \nu_2 ) (f_1 \otimes f_2 ) \\
&=& \int_{ \prod_{i=1}^2 \left( [ \mathfrak{sa} \, \mathcal{L} (H) ]^{d_i}  \times  P (\tilde{H})\right) }   \,  \prod_{i=1}^2 \langle \xi_i,  f_i(X_i) \, \xi_i \rangle_{\tilde{H}} \,  \otimes_{i=1}^2 \mu_{\nu_i} (d \xi_1,d X_1, d \xi_2 , d X_2)  \\
&=& \int_{  [ \mathfrak{sa} \, \mathcal{L} (H) ]^{d_1 + d_2}  \times  [ P (\tilde{H}) ]^2 }   \,   \left \langle (f_1 \otimes f_2) (X_1, X_2) \otimes_{i=1}^2  \xi_i  \, \, , \, \otimes_{i=1}^2  \xi_i  \right\rangle_{ \tilde{H}^{\otimes^2} } \, ( \otimes_{i=1}^2 \mu_{\nu_i} ) (d \xi_1, d \xi_2 , d X_1, d X_2) \,,
\end{eqnarray*}
and hence we have for all $f \in \overline{\mathcal{A}(X_1)} \otimes_{\min} \overline{\mathcal{A}(X_2)} $ 
\begin{eqnarray*}
 ( \tau_1 \otimes  \tau_2 ) (f )  = \int_{  [ \mathfrak{sa} \, \mathcal{L} (H) ]^{d_1 + d_2}  \times  [ P (\tilde{H}) ]^2 }   \,    \left \langle f (X_1, X_2) \otimes_{i=1}^2  \xi_i  \, \, , \, \otimes_{i=1}^2  \xi_i  \right\rangle_{ \tilde{H}^{\otimes^2} }  \,  (\otimes_{i=1}^2 \mu_{\nu_i} ) (d \xi_1, d \xi_2 , d X_1, d X_2) \,,
\end{eqnarray*}
where $ \otimes_{i=1}^2 \mu_{\nu_i}  \in \mathcal{M} \left(   B_R^{  [\mathfrak{sa} \mathcal{L}(H) ]^{d_1 + d_2} } (0) \times [ P(\tilde{H}) ]^2  \right) $ is a finite Radon measure. 
Likewise, if $\nu_i \in  \tilde{\mathcal{M}} (  \overline{\mathcal{A}(X_i)} )$ for $i = 1,...,k$, then the tensor functional   $\otimes_{i=1}^k \nu_i $ can be represented as
\begin{eqnarray*}
&  & ( \otimes_{i=1}^k \tau_i ) (f )  \\
&=&\int_{  [ \mathfrak{sa} \, \mathcal{L} (H) ]^{\sum_{i=1}^k d_i}  \times  [ P (\tilde{H}) ]^k }   \,  
 \left \langle f (X_1,..., X_k)  \otimes_{i=1}^k  \xi_i  \, \, , \, \otimes_{i=1}^k  \xi_i  \right\rangle_{ \tilde{H}^{\otimes^k} } \,  ( \otimes_{i=1}^k \mu_{\nu_i} ) (d \xi_1, ... d \xi_k , d X_1, ..., d X_k) \,,
\end{eqnarray*}
where $ \otimes_{i=1}^k \mu_{\nu_i}  \in \mathcal{M} \left(   B_R^{  [\mathfrak{sa} \mathcal{L}(H) ]^{\sum_{i=1}^k d_i} } (0) \times [ P(\tilde{H}) ]^k  \right) $ is a finite Radon measure. 
This now motivates us to \textbf{define} the following:

\begin{definition}
The space
$\tilde{\mathcal{M}}  \left(   { \otimes_{\min} }_{i=1}^k  \overline{\mathcal{A}(X_i)} \right) \subset  { \left(  { \otimes_{\min} }_{i=1}^k  \overline{\mathcal{A}(X_i)} \right)   }^{*}  $ contains all the linear functionals $\nu$ over $ { \otimes_{\min} }_{i=1}^k  \overline{\mathcal{A}(X_i)} $ such that they can be represented as
\begin{eqnarray*}
 \nu (f )  = \int_{  [ \mathfrak{sa} \, \mathcal{L} (H) ]^{\sum_{i=1}^k d_i}  \times  [ P (\tilde{H}) ]^k }   \,  
 \left \langle f (X_1,..., X_k) \, \otimes_{i=1}^k  \xi_i  \, \, , \, \otimes_{i=1}^k  \xi_i  \right\rangle_{ \tilde{H}^{\otimes^k} } \,  \mu_{\nu} (d \xi_1, ... d \xi_k , d X_1, ..., d X_k) \,,
\end{eqnarray*}
where $ \mu_{\nu}  \in \mathcal{M} \left(   B_R^{  [\mathfrak{sa} \mathcal{L}(H) ]^{\sum_{i=1}^k d_i} } (0) \times [ P(\tilde{H}) ]^k  \right)  $ is a finite Radon measure.  
Meanwhile, the space $\tilde{\mathcal{P}}  \left(   { \otimes_{\min} }_{i=1}^k  \overline{\mathcal{A}(X_i)} \right) \subset \tilde{\mathcal{M}}  \left(   { \otimes_{\min} }_{i=1}^k  \overline{\mathcal{A}(X_i)} \right) $ contains all the states $\tau$ over $  { \otimes_{\min} }_{i=1}^k  \overline{\mathcal{A}(X_i)} $ with the above integral representation with $\mu_{\tau} \in \mathcal{P} \left(   B_R^{  [\mathfrak{sa} \mathcal{L}(H) ]^{\sum_{i=1}^k d_i} } (0) \times [ P(\tilde{H}) ]^k  \right)   $ being a Radon probability measure.

\end{definition}

It is clear that whenever $\mu_i \in  \tilde{\mathcal{M}} (  \overline{\mathcal{A}(X_i)} )$ for $i = 1,...,k$, we have $\otimes_{i=1}^k \nu_i   \in \tilde{\mathcal{M}}  \left(   { \otimes_{\min} }_{i=1}^k  \overline{\mathcal{A}(X_i)} \right) $; and whenever $\tau_i \in  \tilde{\mathcal{P}} (  \overline{\mathcal{A}(X_i)} )$ for $i = 1,...,k$, we have $\otimes_{i=1}^k \tau_i   \in \tilde{\mathcal{P}}  \left(   { \otimes_{\min} }_{i=1}^k  \overline{\mathcal{A}(X_i)} \right) $.

Now for any $\tau \in  \tilde{\mathcal{P}}  \left(   { \otimes_{\min} }_{i=1}^k  \overline{\mathcal{A}(X_i)} \right) $, with the above representation, we can again extend the domain of definition of $\tau$ to
\begin{eqnarray*}
\tau :  C \left( B_R^{  [\mathfrak{sa} \mathcal{L}(H) ]^{2 \sum_{i=1}^k d_i} } (0)   ; [ \mathcal{L}(H) ]^{ \sum_{i=1}^k d_i}   \right) & \rightarrow&  C \left( B_R^{  [\mathfrak{sa} \mathcal{L}(H) ]^{ \sum_{i=1}^k d_i } } (0)  ; \mathbb{C} \right) 
\end{eqnarray*}
where whenever $( Y_1,..., Y_k ) \in B_R^{  [\mathfrak{sa} \mathcal{L}(H) ]^{ \sum_{i=1}^k d_i } } (0)$, we have
\begin{eqnarray*}
&    & \left[ \tau(f) \right]( Y_1,..., Y_k ) \\
& := &  \int_{  [ \mathfrak{sa} \, \mathcal{L} (H) ]^{\sum_{i=1}^k d_i}  \times  [ P (\tilde{H}) ]^k }   \,  \left \langle f (X_1,..., X_k, Y_1,..., Y_k ) \, \otimes_{i=1}^k  \xi_i  \, \, , \, \otimes_{i=1}^k  \xi_i  \right\rangle_{ \tilde{H}^{\otimes^k} } \,  \mu_{\tau} (d \xi_1, ... d \xi_k , d X_1, ..., d X_k)  \,.
\end{eqnarray*}
Now whenever $\tau \in \tilde{\mathcal{P}}  \left(   { \otimes_{\min} }_{i=1}^k  \overline{\mathcal{A}(X_i)} \right)$ and $(X_1,...,X_k) \in B_R^{  [\mathfrak{sa} \mathcal{L}(H) ]^{ \sum_{i=1}^k d_i } } (0)$, we may write
\begin{eqnarray*}
d \left((X_1,...,X_k), \tau  \right) :=  \lim_{p \rightarrow \infty}  \left( \tau \left(   \left[  \sum_{i=1}^k  \left( \otimes_{ j = 1}^{i-1} 1_j \right) \otimes (T^*_{X_i} \mathcal{E} )( \cdot )  \otimes \left( \otimes_{ j = i+1}^{k} 1_j \right)  \right]^{-p}  \right) \right)^{- \frac{1}{2p}}   \\
\end{eqnarray*}
where $1_i$ is the identity element in $\mathcal{A} \langle X_i \rangle$.
With these notations at hand, we may define the following:

\begin{definition}
Given $ \tau \in   \tilde{\mathcal{P}}  \left(   { \otimes_{\min} }_{i=1}^k  \overline{\mathcal{A}(X_i)} \right)  $ and $\nu \in   \tilde{\mathcal{M}}  \left(   { \otimes_{\min} }_{i=1}^k  \overline{\mathcal{A}(X_i)} \right) $, the relative distance zeta functional  $\zeta_{ \tau } (s,   \nu ) $ for $s \in \mathbb{C}$ with $\Re (s) $ sufficiently large is given by
\begin{eqnarray*}
\zeta_{ \tau } (s, \nu )  :   { \otimes_{\min} }_{i=1}^k  \overline{\mathcal{A}(X_i)}   
%  { \color{blue} \overline{ \mathcal{A}(X_1, ..., X_k)}   }  
 & \rightarrow&  \mathbb{C}  \notag  \\
\zeta_{ \tau } (s,   \nu ) ( g ) &:=&   \nu \left( d (\cdot,  \tau )^{s - \sum_{i=1}^k d_i  } \, g (\cdot ) \right) \, . 
\end{eqnarray*}
For a given (fixed) $g$, we then meromorphically continue $\left[ \zeta_{\tau} (s, \nu) \right] (g)$ to a domain of definition $\Omega_{\left[ \zeta_{\tau} ( \, \cdot \, , \nu ) \right] (g)}$ such that  $\{\Re (s) > \overline{\text{dim}}_B (\tau, \nu\} \subset \Omega_{\left[ \zeta_{\tau} ( \, \cdot \, , \nu ) \right] (g)} \subset \mathbb{C}$, and still denote the continued function as $\left[ \zeta_{\tau} (s, \nu ) \right] (g)$. (The fact that the functional is well defined for $ \Re (s) > \overline{\text{dim}}_B (\tau, \nu) $ is stated in Lemma \ref{convergence_3}.)
\end{definition}

And similarly, we again have the following definition.
\begin{definition}
The (set-valued) relative complex dimension function of $\tau $ with respect to $\nu$ is defined as a map that brings $g \in  { \otimes_{\min} }_{i=1}^k  \overline{\mathcal{A}(X_i)}   $ to the set of poles of $[ \zeta_{\tau} (s,  \nu ) ](g) $, i.e.
\begin{eqnarray*}
\mathcal{D}_{\mathbb{C}} ( \tau, \nu) :  { \otimes_{\min} }_{i=1}^k  \overline{\mathcal{A}(X_i)}  & \rightarrow& 2^{\mathbb{C}}  \\
\mathcal{D}_{\mathbb{C}} ( \tau, \nu) [g]  &=& \left \{ s \in \Omega_{\left[ \zeta_{\tau} ( \, \cdot \, , \nu ) \right] (g)}  \, : \,   [ \zeta_{\tau} (s, \nu )](g)  = \infty  \right \} \,.
\end{eqnarray*}
The set of relative complex dimensions of $\tau$ with respect to $\nu$ is then given by
\[
 \mathcal{D}_{\mathbb{C}} ( \tau, \nu) := \bigcup_{g \in  { \otimes_{\min} }_{i=1}^k  \overline{\mathcal{A}(X_i)}  } \mathcal{D}_{\mathbb{C}} ( \tau, \nu) [g] \,.
\]
A state $\tau$ is a fractal (state) with respect to $\nu$ if 
\[
 \mathcal{D}_{\mathbb{C}} ( \tau, \nu)  \backslash \mathbb{R} \neq \emptyset \,.
\]
\end{definition}

We also generalize and define, for $ \tau \in   \tilde{\mathcal{P}}  \left(   { \otimes_{\min} }_{i=1}^k  \overline{\mathcal{A}(X_i)} \right)  $ and $\nu \in   \tilde{\mathcal{M}}  \left(   { \otimes_{\min} }_{i=1}^k  \overline{\mathcal{A}(X_i)} \right) $, that
\begin{eqnarray*}
 \nu^{ \tau ,\delta}  :=  \nu \circ  \text{mul}_{\chi_{ \left \{ X \in [ \mathfrak{sa} \, \mathcal{L} (H) ]^{\sum_{i=1}^k d_i} : d \left(  (X_1,...,X_k), \tau  \right)  \leq \delta  \right \}  } } 
\end{eqnarray*}
as well as
\begin{eqnarray*}
 \mathcal{M}^s_{*}  ( \tau ,  \nu )   := \liminf_{\delta \rightarrow 0^+}  \, \delta^{ s - \sum_{i=1}^k d_i  } \,  \| \nu^{ \tau ,\delta}  \|  \,, \quad  \mathcal{M}^{*\,s}   ( \tau,  \nu)   := \limsup_{\delta \rightarrow 0^+}  \, \delta^{ s - \sum_{i=1}^k d_i  } \,  \| \nu^{ \tau,\delta}  \|   \,,
\end{eqnarray*}
and
\begin{eqnarray*}
\underline{\text{dim}}_B (\tau, \nu) &=& \inf \{ s \geq 0 \, , \,  \mathcal{M}^s_*(\tau, \nu )  = 0 \} = \sup \{ s \geq 0 \, , \,  \mathcal{M}^s_*(\tau, \nu )  = \infty \} \\
\overline{\text{dim}}_B (\tau, \nu) & =& \inf \{ s \geq 0 \, , \, \mathcal{M}^{* \, s}(\tau, \nu )  = 0 \} = \sup \{ s \geq 0 \, , \, \mathcal{M}^{* \, s}(\tau, \nu )  = \infty \}.
\end{eqnarray*}
and $\underline{\text{dim}}_B (\tau, \nu) = \overline{\text{dim}}_B (\tau, \nu)  $, we write $\text{dim}_B (\tau, \nu) $ the relative box counting dimension of $\tau $ with respect to $\nu$.  The state  $\tau$ is relatively Minkowski nondegenerate with respect to $\nu$ if 
\[
0 <  \mathcal{M}^{\text{dim}_B (\tau, \nu)}_*(\tau, \nu )  \leq  \mathcal{M}^{*\,\text{dim}_B (\tau, \nu)}(\tau, \nu )  < \infty \,,
\]
and the $\text{dim}_B (\tau, \nu)$-dimensional relative Minkowski content of $\tau$ with respect to $\nu$ is 
\[
\mathcal{M}^{\text{dim}_B (\tau, \nu)}(\tau, \nu )  := \mathcal{M}^{\text{dim}_B (\tau, \nu)}_*(\tau, \nu )  = \mathcal{M}^{*\,\text{dim}_B (\tau, \nu)}(\tau, \nu ) 
\]
if they are equal.

We then again have the following lemma, that we now state without proof (as the proof is the same as Lemma \ref{convergence_2} in the previous subsection).

\begin{lemma} \label{convergence_3}
If  $\tau \in \tilde{\mathcal{P}}  \left(   { \otimes_{\min} }_{i=1}^k  \overline{\mathcal{A}(X_i)} \right)  $ and $\nu \in  \tilde{\mathcal{M}}  \left(   { \otimes_{\min} }_{i=1}^k  \overline{\mathcal{A}(X_i)} \right)  $, then
$\zeta_{\tau} (s, \nu ) $ is a well-defined bounded linear functional and is holomorphic with respect to $s$ whenever $\Re (s) > \overline{\text{dim}}_B (\tau, \nu)$, with some $\delta > 0$ such that
\begin{eqnarray}
[ \zeta_{\tau} (s, \nu) ] (g) =   \delta^{s-d}   \nu (g)  -  (s-d) \int_{0}^{\delta}  t^{s-d-1} \,   \nu^{\tau,t}  (g) \, dt
\label{zeta_non_communtative_tube_tensor} 
\end{eqnarray}
for any  $g \in \overline{ \mathcal{A} \langle X \rangle }  $, and 
\begin{eqnarray}
[ \partial_s  \zeta_{\tau} (s, \nu) ] (g) =  [ \zeta_{\tau} (s, \nu) ] \left(  \log(d (\cdot, \tau ))  g (\cdot ) \right)  \,.
\label{zeta_non_communtative_derivative_tensor} 
\end{eqnarray}
Furthermore, if the relative box counting dimension of $\tau$ with respect to $\nu$, i.e. $\text{dim}_B (\tau, \nu) $, exists and $\text{dim}_B (\tau, \nu) < d$, and moreover, if $ \mathcal{M}^{\text{dim}_B (\tau, \nu)}_*(\tau, \nu )  > 0  $, then $ \| \zeta_{\tau} (s, \nu ) \|  \rightarrow \infty $ when $s \rightarrow \text{dim}_B (\tau, \nu)^+$ for $s \in \mathbb{R}$.
\end{lemma}

\begin{remark}
We may again name the map
$ (s, \nu) \mapsto \left(  g \mapsto  \int_{0}^{\delta}  t^{s-d-1} \,   \nu^{\tau,t}  (g) \, dt \right) $
in  \eqref{zeta_non_communtative_tube_tensor} in this case as the relative tube zeta functional of $\tau $ in the spirit of  \cite{LRZ_book,LRZ,LRZ3,LRZ5,LRZ6a}.
\end{remark}

Now we aim to better understand  $d \left((X_1,...,X_k), \tau  \right) $ when $\tau = \otimes_{i=1}^k \tau_i  $ with $\tau_i \in  \tilde{\mathcal{P}} (  \overline{\mathcal{A}(X_i)} )$ for all $i = 1,...,k $.
In fact, we first realize that, via functional calculus, we have
\begin{eqnarray*}
&  & 
\left \langle  \left[  \sum_{i=1}^k  \left( \otimes_{ j = 1}^{i-1} 1_j \right) \otimes (T^*_{X_i} \mathcal{E} )( \cdot )  \otimes \left( \otimes_{ j = i+1}^{k} 1_j \right)  \right]^{-p}   \left( \otimes_{i=1}^k \xi_i \right)  \, \, , \, \left( \otimes_{i=1}^k  \xi_i \right)  \right\rangle_{ \tilde{H}^{\otimes^k} }  \\ 
&=& \int_{\sigma\left(  \sum_{i=1}^k  \left( \otimes_{ j = 1}^{i-1} 1_j \right) \otimes (T^*_{X_i} \mathcal{E} )( \cdot )  \otimes \left( \otimes_{ j = i+1}^{k} 1_j \right)  \right) }  \lambda^{-p} \,  E_{   \sum_{i=1}^k  \left( \otimes_{ j = 1}^{i-1} 1_j \right) \otimes (T^*_{X_i} \mathcal{E} )( \cdot )  \otimes \left( \otimes_{ j = i+1}^{k} 1_j \right)  }  \left(  \otimes_{i=1}^k  \xi_i  , \otimes_{i=1}^k  \xi_i  \right) (d \lambda )  \\
&=& \int_{ \prod_{i=1}^k \sigma\left( (T^*_{X_i} \mathcal{E} )( \cdot )  \right) }  \left( \sum_{i=1}^k \lambda_i \right)^{-p} \, \left( \otimes_{i=1}^k E_{   (T^*_{X_i} \mathcal{E} )( \cdot ) }  \left( \xi_i, \xi_i \right)\right) \left(d \lambda_1,...d \lambda_k \right) \,.
\end{eqnarray*}
(With this, transformation rules that preserves $T^*_{X_i} \mathcal{E} (\cdot)$ or their spectrum always induce transformation rules over the zeta functionals as in the previous section.  We skip them for the sake of brevity.)
Moreover, we may directly compute
\begin{eqnarray*}
& &d \left((X_1,...,X_k), \otimes_{i=1}^k \tau_i  \right)\\
&= & \lim_{p \rightarrow \infty}  \left( \left( \otimes_{i=1}^k \tau_i  \right) \left(   \left[  \sum_{i=1}^k  \left( \otimes_{ j = 1}^{i-1} 1_j \right) \otimes (T^*_{X_i} \mathcal{E} )( \cdot )  \otimes \left( \otimes_{ j = i+1}^{k} 1_j \right)  \right]^{-p}  \right) \right)^{- \frac{1}{2p}}  \\
&=& \lim_{p \rightarrow \infty}  \bigg[  
\int \limits_{  [ \mathfrak{sa} \, \mathcal{L} (H) ]^{\sum_{i=1}^k d_i}  \times  [ P (\tilde{H}) ]^k }   \,  
 \int \limits_{ \prod_{i=1}^k \sigma\left( (T^*_{X_i} \mathcal{E} )( \cdot )  \right) } \\
&  & \quad \quad \quad \quad \left( \sum_{i=1}^k \lambda_i \right)^{-p} \,  \left( \otimes_{i=1}^k E_{   (T^*_{X_i} \mathcal{E} )( \cdot ) } \right)  \left( \xi_i, \xi_i \right) \left(d \lambda_1,...d \lambda_k \right) 
\,  ( \otimes_{i=1}^k \mu_{\tau_i} ) (d \xi_1, ... d \xi_k , d X_1, ..., d X_k) 
 \bigg]^{-\frac{1}{2p}} \\
&=&  \sqrt{  \inf_{ \scriptsize \begin{matrix} ( \xi_i, Y_i ) \in \, \text{spt} \, \mu_{\tau_i}, \\ i = 1,...,k \end{matrix} }  \left \{   \inf \left\{ \sum_{i=1}^k \lambda_i  \, : \,  \lambda_i \in  \text{spt} \left ( E_{\mathcal{E}(X_i-Y_i)}  \left( \xi_i, \xi_i \right) \right), i = 1,...,k \right \} \right \}  } \,.
\end{eqnarray*}
In particular, for any $s \in \mathbb{C}$, whenever $\tau_i \in  \tilde{\mathcal{P}} (  \overline{\mathcal{A}(X_i)} )$ and $\nu_i \in  \tilde{\mathcal{M}} (  \overline{\mathcal{A}(X_i)} )$ for all $i = 1,...,k $, we observe that $\zeta_{ \otimes_{i=1}^k \tau_i } (s,   \otimes_{i=1}^k \nu_i ) $ is the unique functional such that
\begin{eqnarray*}
\zeta_{ \otimes_{i=1}^k \tau_i } (s,  \otimes_{i=1}^k \nu_i  )  :   { \otimes_{\min} }_{i=1}^k  \overline{\mathcal{A}(X_i)}   
%  { \color{blue} \overline{ \mathcal{A}(X_1, ..., X_k)}   }  
 & \rightarrow&  \mathbb{C}  \notag  \\
\zeta_{ \otimes_{i=1}^k \tau_i } (s,   \otimes_{i=1}^k \nu_i ) (  \otimes_{i=1}^k g_i  ) &=& \left(  \otimes_{i=1}^k \nu_i  \right) \left( d (\cdot,  \otimes_{i=1}^k \tau_i )^{s - \sum_{i=1}^k d_i  } \, (  \otimes_{i=1}^k g_i  )(\cdot ) \right) \,,
\end{eqnarray*}
which we may explicitly evaluate as
{\small
\begin{eqnarray*}
&  & [ \zeta_{ \otimes_{i=1}^k \tau_i } (s,   \otimes_{i=1}^k \nu_i ) ] (  \otimes_{i=1}^k g_i  ) \\
& = & 
\int \limits_{  [ \mathfrak{sa} \, \mathcal{L} (H) ]^{\sum_{i=1}^k d_i}  \times  [ P (\tilde{H}) ]^k }    \left( \inf_{ \scriptsize \begin{matrix} ( \xi_i, Y_i ) \in \, \text{spt} \, \mu_{\tau_i}, \\ i = 1,...,k \end{matrix} }  \left \{   \inf \left\{ \sum_{i=1}^k \lambda_i  \, : \,  \lambda_i \in  \text{spt} \left ( E_{\mathcal{E}(X_i-Y_i)}  \left( \xi_i, \xi_i \right) \right), i = 1,...,k \right \} \right \}   \right)^{\frac{s - \sum_{i=1}^k d_i }{2}}  \,  \\
&  & \quad \quad \quad \quad \quad
 \prod_{i=1}^k \langle \xi_i,  g_i(X_i) \, \xi_i \rangle_{\tilde{H}}
\,  ( \otimes_{i=1}^k \mu_{\nu_i} )  (d \xi_1, ... d \xi_k , d X_1, ..., d X_k) \,.
\end{eqnarray*}
}To better understand the above functional $\zeta_{ \otimes_{i=1}^k \tau_i } (s,   \otimes_{i=1}^k \nu_i ) $, inspired by Remark \ref{Product}, we define the following auxillary function (which is not an infinity norm in any sense)
\begin{eqnarray*}
& & \widetilde{d_\infty} \left((X_1,...,X_k), \otimes_{i=1}^k \tau_i  \right) \\
&:=& \inf_{ \scriptsize \begin{matrix} ( \xi_i, Y_i ) \in \, \text{spt} \, \mu_{\tau_i}, \\ i = 1,...,k \end{matrix} }   \left \{   \inf \left\{ \left( \sup_i  \sqrt{ \lambda_i }  \right) \, : \,  \lambda_i \in  \text{spt} \left ( E_{\mathcal{E}(X_i-Y_i)}  \left( \xi_i, \xi_i \right) \right)  \text{ for } i = 1,...,k \right \} \right \} \,.
\end{eqnarray*}
We also denote
\begin{eqnarray*}
\widetilde{ (\otimes_{i=1}^k \nu_i  ) } ^{ \otimes_{i=1}^k \tau_i  ,\delta}_\infty  :=  (\otimes_{i=1}^k \nu_i  )  \circ  \text{mul}_{\chi_{ \left\{ X \in [ \mathfrak{sa} \, \mathcal{L} (H) ]^{\sum_{i=1}^k d_i} : \widetilde{d_\infty} \left(  (X_1,...,X_k), \otimes_{i=1}^k \tau_i  \right)  \leq \delta  \right\}  } } 
\end{eqnarray*}
and write
\begin{eqnarray*}
\widetilde{ \mathcal{M}^s_{*,\infty} } (  \otimes_{i=1}^k \tau_i ,  \otimes_{i=1}^k \nu_i )   &:=& \liminf_{\delta \rightarrow 0^+}  \, \delta^{ s - \sum_{i=1}^k d_i  } \,  \| \widetilde{(\otimes_{i=1}^k \nu_i)}_\infty^{\otimes_{i=1}^k \tau_i,\delta}  \|  \,,  \\
 \quad \widetilde{ \mathcal{M}^{*\,s}_\infty }  (  \otimes_{i=1}^k \tau_i ,  \otimes_{i=1}^k \nu_i )   &:=& \limsup_{\delta \rightarrow 0^+}  \, \delta^{ s - \sum_{i=1}^k d_i  } \,  \| \widetilde{(\otimes_{i=1}^k \nu_i)}_\infty^{\otimes_{i=1}^k \tau_i,\delta}  \|   
\end{eqnarray*}
for a given $s \geq 0$, and
\begin{eqnarray*}
 \widetilde{ \mathcal{M}^s_{\infty} } (  \otimes_{i=1}^k \tau_i ,  \otimes_{i=1}^k \nu_i )  :=
\widetilde{ \mathcal{M}^s_{*,\infty} } (  \otimes_{i=1}^k \tau_i ,  \otimes_{i=1}^k \nu_i )  =  \widetilde{ \mathcal{M}^{*\,s}_\infty }  (  \otimes_{i=1}^k \tau_i ,  \otimes_{i=1}^k \nu_i )  
\end{eqnarray*}
when they are so equal.
We moreover write the following auxilliary zeta(-like) functional:
\[
[ \widetilde{\zeta_{ \otimes_{i=1}^k \tau_i , \infty}} (s,   \otimes_{i=1}^k \nu_i ) ] (  \otimes_{i=1}^k g_i  )  := 
\left(  \otimes_{i=1}^k \nu_i  \right) \left( \widetilde{ d_\infty } (\cdot,  \otimes_{i=1}^k \tau_i )^{s  - \sum_{i=1}^k d_i } \, (  \otimes_{i=1}^k g_i  )(\cdot ) \right) \,.
\]
Then the following lemma holds.

\begin{lemma}
Whenever $\tau_i \in  \tilde{\mathcal{P}} (  \overline{\mathcal{A}(X_i)} )$ and $\nu_i \in  \tilde{\mathcal{M}} (  \overline{\mathcal{A}(X_i)} )$ for $i = 1,...,k$, we have
\begin{enumerate}
\item
It holds that
\begin{eqnarray*}
 k^{-\frac{s-d}{2}} \mathcal{M}^s_{*}( \otimes_{i=1}^k \tau_i  , \otimes_{i=1}^k \nu_i  )   \leq \widetilde{ \mathcal{M}^s_{*,\infty} }( \otimes_{i=1}^k \tau_i  ,  \otimes_{i=1}^k \nu_i  )   \leq k^{\frac{s-d}{2}} \mathcal{M}^s_{*}(\otimes_{i=1}^k \tau_i  , \otimes_{i=1}^k \nu_i  )   \,, \\
 k^{-\frac{s-d}{2}} \mathcal{M}^{*\,s}(\otimes_{i=1}^k \tau_i ,  \otimes_{i=1}^k \nu_i  )  \leq  \widetilde{ \mathcal{M}^{*\,s}_\infty }(\otimes_{i=1}^k \tau_i ,  \otimes_{i=1}^k \nu_i  )   \leq k^{\frac{s-d}{2}} \mathcal{M}^{*\,s}(\otimes_{i=1}^k \tau_i ,  \otimes_{i=1}^k \nu_i  ) \,.
\end{eqnarray*}
\item
It holds that
\begin{eqnarray*}
\widetilde{ \mathcal{M}^{\sum_{i=1}^k s_i}_{*,\infty} }( \otimes_{i=1}^k \tau_i, \otimes_{i=1}^k \nu_i  )  & \geq & \prod_{i=1}^k \mathcal{M}^{s_i}_{*}(\tau_i, \nu_i )  \text{ with a convention that } 0 \cdot \infty = 0 \,, \\
\widetilde{ \mathcal{M}^{*\, \sum_{i=1}^k s_i}_\infty } ( \otimes_{i=1}^k \tau_i, \otimes_{i=1}^k \nu_i  )  & \leq & \prod_{i=1}^k \mathcal{M}^{*\,s_i} (\tau_i, \nu_i )   \text{ with a convention that } 0 \cdot \infty = \infty \,,
\end{eqnarray*}
and
\begin{eqnarray*}
\underline{\text{dim}}_B ( \otimes_{i=1}^k \tau_i,   \otimes_{i=1}^k  \nu_i)& \geq &\sum_{i=1}^k \underline{\text{dim}}_B (\tau_i, \nu_i) \, , \\
\overline{\text{dim}}_B ( \otimes_{i=1}^k \tau_i,   \otimes_{i=1}^k  \nu_i) &  \leq & \sum_{i=1}^k \overline{\text{dim}}_B  (\tau_i, \nu_i) .
\end{eqnarray*}
Hence if $\underline{\text{dim}}_B (\tau_i, \nu_i) = \overline{\text{dim}}_B (\tau_i, \nu_i)  $ for all $i$, then
\begin{eqnarray*}
\text{dim}_B ( \otimes_{i=1}^k \tau_i,   \otimes_{i=1}^k  \nu_i) = \sum_{i=1}^k \text{dim}_B (\tau_i, \nu_i) \,,
\end{eqnarray*}
Furthermore, if for all $i$, $\tau_i$ is relatively Minkowski nondegenerate with respect to $\nu_i$ and the $\text{dim}_B (\tau_i, \nu_i)$-dimensional relative Minkowski content are all well-defined, then the auxillary variable satisfies
\[
 \widetilde{ \mathcal{M}^{  \text{dim}_B ( \otimes_{i=1}^k \tau_i,   \otimes_{i=1}^k  \nu_i)  }_{\infty} } ( \otimes_{i=1}^k \tau_i, \otimes_{i=1}^k \nu_i  )  = \prod_{i=1}^k \mathcal{M}^{\text{dim}_B (\tau_i, \nu_i)} (\tau_i, \nu_i )
\]
\item
For $s > \overline{\text{dim}}_B (  \otimes_{i=1}^k \tau_i , \otimes_{i=1}^k \nu_i  )$, $s \in \mathbb{R}$, and $g_i \geq 0$, we have
\[
 k^{-\frac{s-d}{2}}  [ \zeta_{ \otimes_{i=1}^k \tau_i } (s, \otimes_{i=1}^k \nu_i ) ] ( \otimes_{i=1}^k g_i ) \leq [ \widetilde{ \zeta_{\otimes_{i=1}^k \tau_i ,\infty} } (s, \otimes_{i=1}^k \nu_i)  ] ( \otimes_{i=1}^k g_i )  \leq  k^{\frac{s-d}{2}}  [ \zeta_{ \otimes_{i=1}^k \tau_i } (s, \otimes_{i=1}^k \nu_i ) ] ( \otimes_{i=1}^k g_i ) 
\]

\item
For all $s_i \in \mathbb{C} $ such that  $ \Re(s_i) > \overline{\text{dim}}_B (\tau_i,  \nu_i)$,  $g_i \geq 0$, we have
\begin{eqnarray*}
&  & \left| \frac{ [ \widetilde{ \zeta_{ \otimes_{i=1}^k \tau_i } } ( \sum_{i=1}^k s_i , \otimes_{i=1}^k \nu_i ) ] ( \otimes_{i=1}^k g_i ) -  \delta^{( \sum_{i=1}^k s_i - \sum_{i=1}^k d_i  ) }  \prod_{i=1}^k \nu_{i} (g_i)  }{ \sum_{i=1}^k s_i - \sum_{i=1}^k d_i } \right|    \\
&\leq  & -  \sum_{i=1}^k    \frac{  \left(  k d_i  \left( \sum_{i=1}^k d_i \right) \right)^{-   \frac{ k ( \Re(s_i) - d_i ) }{2} }    [ \widetilde{ \zeta_{\tau_i^{\otimes^k} } } (k \Re(s_i) , \nu_i^{\otimes^k}  ) ] (  g_i^{\otimes^k} )  - \delta^{ (k \Re(s_i) - k d_i )}  [ \nu_{i} (g_i) ]^k }{ k^2 (\Re(s_i) -  d_i )}    \,.
\end{eqnarray*}

\item
(Shift property) Whenever the relative box counting dimension $ \text{dim}_B (\tau_2, \nu_2) $  
and the  $\text{dim}_B (\tau_2, \nu_2)$-dimensional relative Minkowski content $\mathcal{M}^{\text{dim}_B (\tau_2, \nu_2)}(\tau_2, \nu_2 ) $ are well-defined,  then for all $s \in \mathbb{R}$ such that $s - \text{dim}_B (\tau_2, \nu_2) > \overline{\text{dim}}_B (\tau_1,  \nu_1 )$,  $g \geq 0$, we have
\begin{eqnarray*}  
&  & - \frac{1}{2}  d_2^{ - \frac{d_2 -  \text{dim}_B (\tau_2, \nu_2)}{2} }    \left(  \frac{ [ \zeta_{\tau_1} \left(s  - \text{dim}_B (\tau_2, \nu_2 ) , \nu_1 \right) ] (g) -  \delta^{(s  - \text{dim}_B (\tau_2, \nu_2 ) - d_1 )}   \nu_{1} (g)  }{ (s - \text{dim}_B (\tau_2, \nu_2 )  - d_1 )}  \right)  \\
&\leq  & - \frac{ [ \widetilde{ \zeta_{\mu_1 \otimes \mu_2 ,\infty} } \left( \Re(s)  , \nu_1 \otimes \nu_2 \right) ] ( g \otimes 1 )  -  \delta^{( \Re(s) - d_1 - d_2)}   \tau_{1} (g) \tau_2 (1) }{ \mathcal{M}^{\text{dim}_B (\tau_2, \nu_2 )} (\mu_2, \nu_2 ) ( \Re(s)  - d_1 - d_2)} \\
& \leq &  - 2  d_2^{ - \frac{d_2 -  \text{dim}_B (\tau_2, \nu_2)}{2} }    \left(  \frac{ [ \zeta_{\tau_1} \left(s  - \text{dim}_B (\tau_2, \nu_2 ) , \nu_1 \right) ] (g) -  \delta^{(s  - \text{dim}_B (\tau_2, \nu_2 ) - d_1 )}   \nu_{1} (g)  }{ (s - \text{dim}_B (\tau_2, \nu_2 )  - d_1 )}  \right) 
\end{eqnarray*}

\end{enumerate}
\end{lemma}

\begin{proof}
The first part and third part of the lemma comes from an equivalence of norm in finite dimension
\[
k^{-\frac{1}{2}} \sqrt{ \sum_{i=1}^k \lambda_i } \leq \sup_i  \sqrt{ \lambda_i } \leq k^{\frac{1}{2}} \sqrt{ \sum_{i=1}^k \lambda_i }
\]
that
\[
k^{-\frac{1}{2}}  \widetilde{d_\infty} \left((X_1,...,X_k), \otimes_{i=1}^k \tau_i  \right)   \leq   d \left((X_1,...,X_k), \otimes_{i=1}^k \tau_i  \right)  \leq k^{\frac{1}{2}}   \widetilde{d_\infty} \left((X_1,...,X_k), \otimes_{i=1}^k \tau_i  \right)  \,,
\]
and therefore
\begin{eqnarray*}
&  &  \{ (X_1,...,X_k) \in  [\mathfrak{sa} \, \mathcal{L} (H) ]^{\sum_{i=1}^k d_i } : d \left((X_1,...,X_k), \otimes_{i=1}^k \tau_i  \right) \leq k^{-\frac{1}{2}} \delta \}  \\
&\subset & \{ (X_1,...,X_k) \in  [\mathfrak{sa} \, \mathcal{L} (H) ]^{\sum_{i=1}^k d_i }  : d_\infty  \left((X_1,...,X_k), \otimes_{i=1}^k \tau_i  \right) \leq \delta \} \\
&\subset&  \{ (X_1,...,X_k) \in  [\mathfrak{sa} \, \mathcal{L} (H) ]^{\sum_{i=1}^k d_i }  :   d \left((X_1,...,X_k), \otimes_{i=1}^k \tau_i  \right)  \leq k^{\frac{1}{2}}  \delta \}
\end{eqnarray*}
and therefore for all positive $g_i \geq 0$ \,
\[
( \otimes_{i=1}^k  \nu_i)^{\otimes_{i=1}^k \tau_i , k^{-\frac{1}{2}} \delta} ( \otimes_{i=1}^k g_i ) \leq  (  \otimes_{i=1}^k  \nu_i )^{ \otimes_{i=1}^k \tau_i  , \delta}_\infty ( \otimes_{i=1}^k g_i ) \leq ( \otimes_{i=1}^k  \nu_i )^{ \otimes_{i=1}^k \tau_i  , k^{\frac{1}{2}} \delta} ( \otimes_{i=1}^k g_i ) \,,
\]
and in particular we may choose $g_i = 1$.

The second part and the fourth of the lemma comes from
\begin{eqnarray*}
&    & \left \{(X_1,...,X_k) \in  [\mathfrak{sa} \, \mathcal{L} (H) ]^{\sum_{i=1}^k d_i }  : \widetilde{d_\infty} \left((X_1,...,X_k) , \otimes_{i=1}^k \tau_i \right)  \leq \delta \right\} \\
&= & \prod_{i=1}^k  \left \{ (X_1,...,X_k) \in  [\mathfrak{sa} \, \mathcal{L} (H) ]^{\sum_{i=1}^k d_i }  : d_\infty \left(x_i, \mu_i \right)  \leq \delta \right \} \,,
\end{eqnarray*}
and therefore
\[
[ \widetilde{ (\otimes_{i=1}^k \nu_i )_\infty}^{ \otimes_{i=1}^k \tau_i , \delta} ]  (   \otimes_{i=1}^k g_i  )  =   \prod_{i=1}^k  ( \nu_i )^{ \tau_i ,  \delta} (  g_i  )    \,. \,,
\]
Together with a similar proof to \eqref{zeta_non_communtative_tube_tensor}), we arrive at
\[
[ \widetilde{\zeta_{\otimes_{i=1}^k \tau_i ,\infty} } (s,  (\otimes_{i=1}^k \nu_i ) ]  (   \otimes_{i=1}^k g_i  ) =   \delta^{s- \sum_{i=1}^k d_i} \prod_{i=1}^k  \nu_i (g_i)  +  (s-\sum_{i=1}^k d_i ) \int_{0}^{\delta}  t^{s-  \sum_{i=1}^k d_i -1} \,   \prod_{i=1}^k  ( \nu_i )^{ \tau_i ,  \delta} (  g_i  )     dt 
\]
Moreover, together with \eqref{elementary_inequality}, we now have, for all $s_i \in \mathbb{C} $ such that $ \Re(s_i) > \overline{\text{dim}}_B (\tau_i,  \nu_i)$, we have
\begin{eqnarray*}
&  &\left| \frac{ [ \widetilde{ \zeta_{\otimes_{i=1}^k \mu_i ,\infty} } ( \sum_{i=1}^k s_i  , \otimes_{i=1}^k \nu_i  ) ] ( \otimes_{i=1}^k g_i  ) -  \delta^{(  \sum_{i=1}^k s_i -  \sum_{i=1}^k d_i )}    \prod_{i=1}^k \nu_{i} (g_i) }{ ( \sum_{i=1}^k s_i -  \sum_{i=1}^k d_i )} \right|  \\
& = &  \left|   \int_{0}^{\delta}  t^{\sum_{i=1}^k s_i  - \sum_{i=1}^k d_i -1} \,  \prod_{i=1}^k ( \nu_i )^{ \tau_i ,  t} (g_i)     dt \right|  \\
& \leq &   \sum_{i=1}^k  \frac{1}{k} \int_{0}^{\delta}  t^{k \Re(s_i) -k d_i -1} \, [  ( \nu_i )^{ \tau_i ,  t} (g_i)  ]^k   dt   \\
&  =  & -  \sum_{i=1}^k \frac{  [ \widetilde{ \zeta_{\tau_i^{\otimes^k} ,\infty} } (k \Re(s_i) , \nu_i^{\otimes^k}  ) ] (  g_i^{\otimes^k}  ) -  \delta^{ (k \Re(s_i) - kd_i )}  [ \nu_{i} (g_i) ]^k }{ k^2 (\Re(s_i) -  d_i )}  
\end{eqnarray*}

The fifth statement comes from the fact that,
whenever the relative box counting dimension $ \text{dim}_B (\tau_2, \nu_2) $  
as well as the  $\text{dim}_B (\tau_2, \nu_2)$-dimensional relative Minkowski $\mathcal{M}^{\text{dim}_B (\tau_2, \nu_2)}(\tau_2, \nu_2 ) $ are well-defined, we have
\begin{eqnarray*}
& &  d_2^{ - \frac{d_2 -  \text{dim}_B (\tau_2, \nu_2)}{2} } \mathcal{M}^{\text{dim}_B (\tau_2, \nu_2 )} (\tau_2, \nu_2 )    \delta^{ d_2 -  \text{dim}_B (\tau_2, \nu_2) } - o( \delta^{ d_2 - \text{dim}_B (\tau_2, \nu_2) } ) \\
& \leq & \,( \nu_2 )_{\infty}^{\tau_2,\delta} (1) \\
& \leq & d_2^{  \frac{d_2 -  \text{dim}_B (\tau_2, \nu_2)}{2} } \mathcal{M}^{\text{dim}_B (\tau_2, \nu_2 )} (\tau_2, \nu_2 )    \delta^{ d_2 -  \text{dim}_B (\tau_2, \nu_2) } + o( \delta^{ d_2 - \text{dim}_B (\tau_2, \nu_2) } ) \,.
\end{eqnarray*}
Therefore we have for all $s \in \mathbb{R}$ such that $s - \text{dim}_B (\tau_2, \nu_2) > \overline{\text{dim}}_B (\tau_1,  \nu_1 )$,
\begin{eqnarray*}
&    &  \frac{1}{2}  d_2^{ - \frac{d_2 -  \text{dim}_B (\tau_2, \nu_2)}{2} } \mathcal{M}^{\text{dim}_B (\tau_2, \nu_2 )} (\tau_2, \nu_2 ) \int_{0}^{\delta}  t^{s -  \overline{\text{dim}}_B (\tau_2, \nu_2) -d_1 -1} \,   \nu_1^{ \mu_1 ,  t} (g)     dt  \\
& \leq &  \int_{0}^{\delta}  t^{s -d_1 - d_2 -1} \,    \nu_1 ^{ \tau_1 ,  t} (g)     \nu_2^{ \tau_2 ,  t} (1)    dt \\
& \leq & 2 d_2^{  \frac{d_2 -  \text{dim}_B (\tau_2, \nu_2)}{2} } \mathcal{M}^{\text{dim}_B (\tau_2, \nu_2 )} (\tau_2, \nu_2 )  \int_{0}^{\delta}  t^{s -  \overline{\text{dim}}_B (\tau_2, \nu_2) -d_1 -1} \,   \nu_1^{ \tau_1 ,  t} (g)     dt 
\end{eqnarray*}
and thus we have
\begin{eqnarray*}  
&  & - \frac{1}{2}  d_2^{ - \frac{d_2 -  \text{dim}_B (\tau_2, \nu_2)}{2} }    \left(  \frac{ [ \zeta_{\tau_1} \left(s  - \text{dim}_B (\tau_2, \nu_2 ) , \nu_1 \right) ] (g) -  \delta^{(s  - \text{dim}_B (\tau_2, \nu_2 ) - d_1 )}   \nu_{1} (g)  }{ (s - \text{dim}_B (\tau_2, \nu_2 )  - d_1 )}  \right)  \\
&\leq  & - \frac{ [ \widetilde{ \zeta_{\mu_1 \otimes \mu_2 ,\infty} } \left( \Re(s)  , \nu_1 \otimes \nu_2 \right) ] ( g \otimes 1 )  -  \delta^{( \Re(s) - d_1 - d_2)}   \tau_{1} (g) \tau_2 (1) }{ \mathcal{M}^{\text{dim}_B (\tau_2, \nu_2 )} (\mu_2, \nu_2 ) ( \Re(s)  - d_1 - d_2)} \\
& \leq &  - 2  d_2^{ - \frac{d_2 -  \text{dim}_B (\tau_2, \nu_2)}{2} }    \left(  \frac{ [ \zeta_{\tau_1} \left(s  - \text{dim}_B (\tau_2, \nu_2 ) , \nu_1 \right) ] (g) -  \delta^{(s  - \text{dim}_B (\tau_2, \nu_2 ) - d_1 )}   \nu_{1} (g)  }{ (s - \text{dim}_B (\tau_2, \nu_2 )  - d_1 )}  \right) 
\end{eqnarray*}
and hence our conclusion.

\end{proof}

\section{Examples} \label{sec_4}
In this section, we would like to explore several examples in both the commutative and noncommutative case to have a better idea what properties the zeta functions are characterizing.

\begin{example} (A commutative example.)
Let $d=1$, $n=1$, and write $H =\mathbb{R} $, and hence $\mathfrak{sa}\,  \mathcal{L} (H)  = \mathbb{R}$.
Let $\mu = \mu_{\mathcal{C}} (dy)$, $\nu =\eta(x) dx$ for a compactly supported $\eta$ that covers $[-1,2]$.  We may compare this with the results of \cite{LRZ_book} about the Cantor set (\cite{LRZ_book}, Section 1.1, Example 1.1.2) and the graph of the Cantor function (\cite{LRZ_book}, Section 5.5.4, Example 5.5.14), and their associate relative fractal drums.

\begin{figurehere} \centering
\includegraphics[width=5cm,height=4cm]{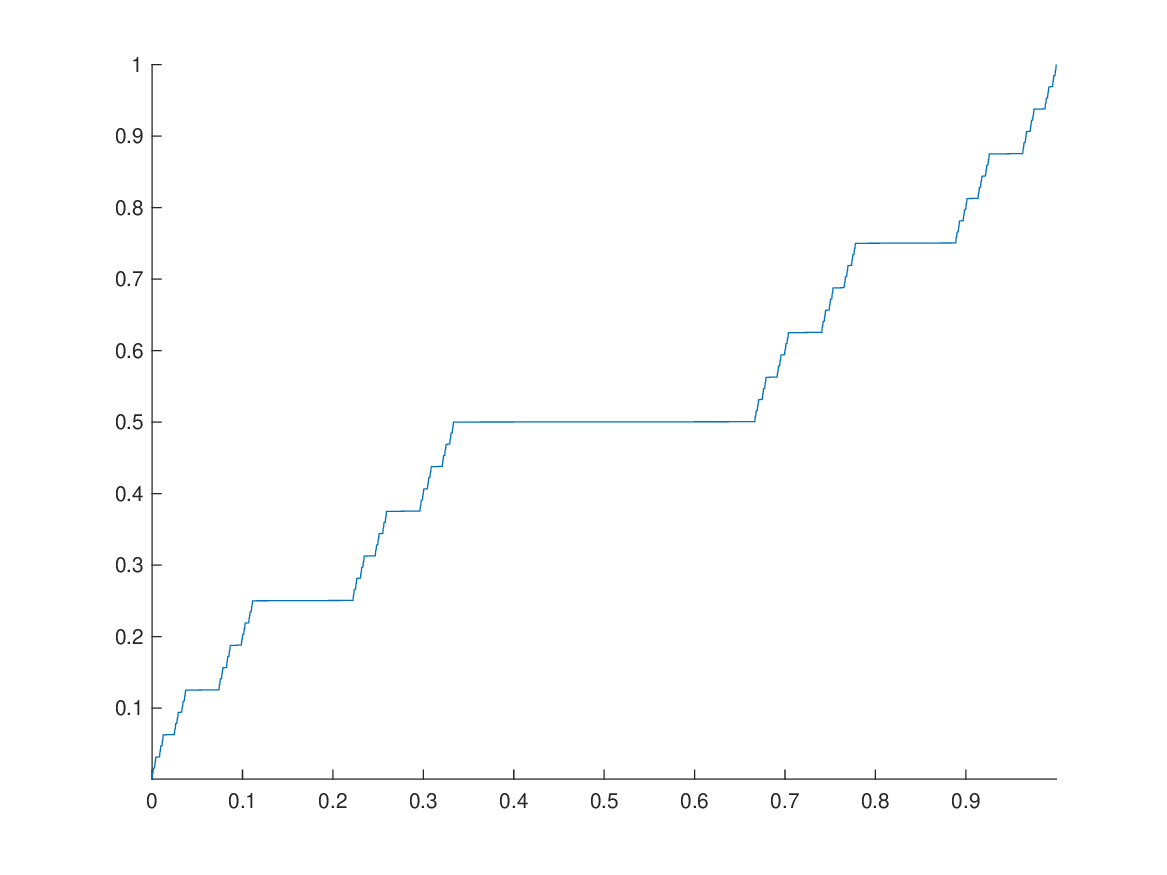}
\caption{The Cantor function, $F_{\mathcal{C}} (x) := \int_{(-\infty,x]} \mu_{\mathcal{C}} (dt)$, for illustrative purpose of the Cantor set} \label{test0A}
\end{figurehere}
\noindent Instantly, we have
\begin{eqnarray*}
 \zeta_{\mu} (s, \nu) =   \int_{ \mathbb{R} }  d (x,\mathcal{C})^{s - 1}  \, \eta(x) dx 
& = & 2 \int_{ \mathbb{R}_+ }  x^{s - 1}  \, \eta(x) dx + C  + \sum_{i = 0}^{\infty}  2^{i+1}  \int_0^{\frac{1}{(2)3^i} }  x^{s - 1}   dx    \\
&= &  C(s)  + \frac{ 2^{1-s} }{s} \sum_{i = 0}^{\infty}   \left( \frac{2}{3^{s}}   \right)^i   \\
&=& C(s) +  \frac{ 2^{1-s} 3^s }{s  \left(   3^{s} -2  \right) } 
\end{eqnarray*}
where $C$ is a constant and $C(s)$ only has a pole at $s = 0$.  Therefore, we have
\[
\mathcal{D}_{\mathbb{C}} ( \mu, \nu ) =  \left \{ \frac{\log(2)}{\log(3)} + i n \frac{2\pi}{\log(3)} \, : \, n \in \mathbb{N} \right \} \bigcup  \{ 0 \}  \,,
\]
noticing the obvious fact that the residues of $ \zeta_{\mu} ( \cdot , \nu)$ at the corresponding poles are non-zero.
For illustrative purpose, the norm of the values of the zeta function is shown in Figure \ref{test0B}.

\begin{figurehere} \centering
\includegraphics[width=5cm,height=4cm]{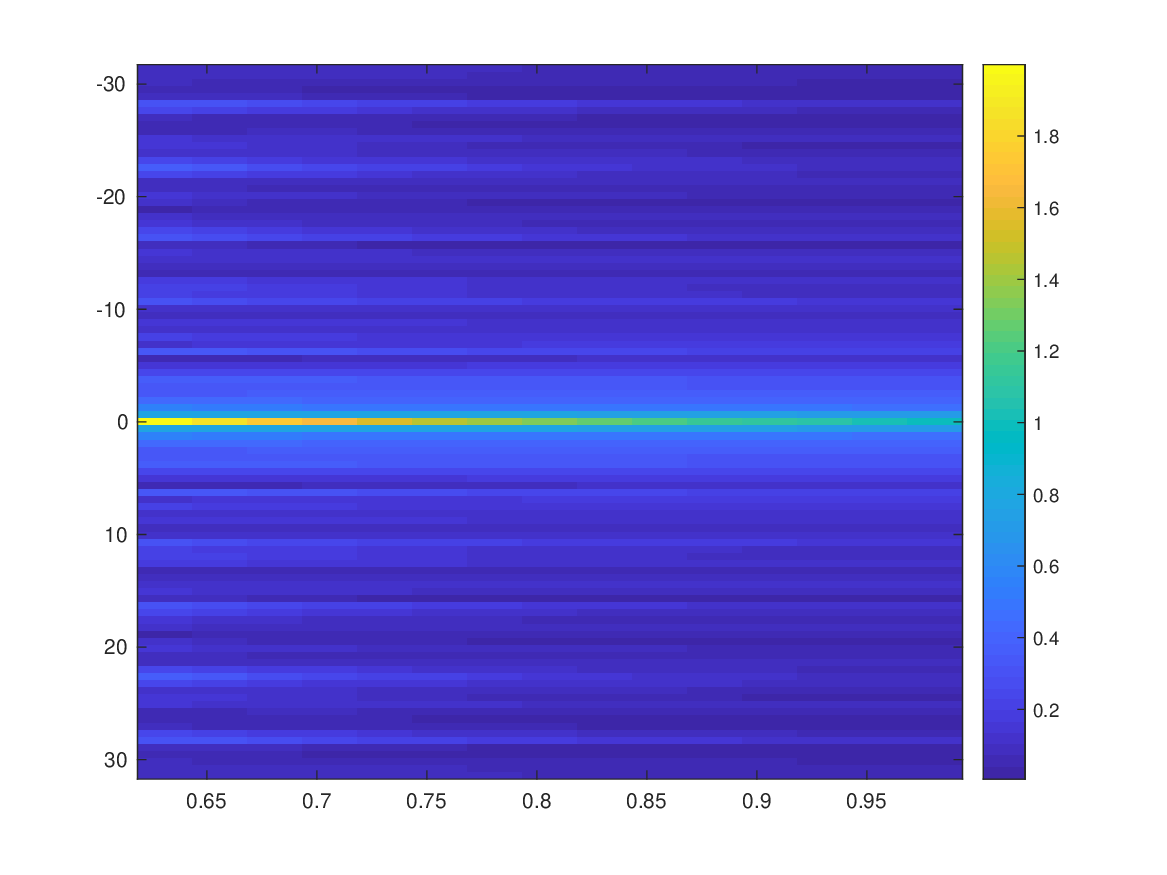}
\caption{$| \zeta_{\mu} (s, \nu) |$  for $\frac{\log(2)}{ \log(3)} \leq \Re (s) $ with $\eta(x) = e^{- 0.001 x^2}$.} \label{test0B}
\end{figurehere}

\end{example}

\begin{example} (A commutative example.)
Let $d=1$, $n=2$, and $H = \mathbb{R}^2$, then
\[
\mathfrak{sa} \mathcal{L} (H) = \left \{ \begin{pmatrix} a & b \\ b & c \end{pmatrix}  :  a, b, c \in \mathbb{R} \right \} \,.
\]
Let
$ \tau  :=  \int_{ [ \mathfrak{sa} \, \mathcal{L} (H) ]^d  }  \frac{1}{2}  \, \text{tr} (\cdot) \,  \tilde{\mu}_{\tau} (d X) :=
 \int_{\mathbb{R}^3}  \frac{1}{2}  \, \text{tr} (\cdot) \,  \mu_\mathcal{C} ( d b) \,  \delta_0( da ) \, \delta_0( dc ) $ where $\mu_\mathcal{C}$ is the Cantor measure, and $ \nu =  \int_{ [ \mathfrak{sa} \, \mathcal{L} (H) ]^d  }  \frac{1}{2}  \, \text{tr} (\cdot) \,  \tilde{\mu}_{\nu} (d X) = \int_{ \mathbb{R}^3  }  \frac{1}{2}  \, \text{tr} (\cdot) \,  \eta(a,b,c) \,  da \, d b \, dc  $
for a smooth $\eta$ with compact support to be chosen later.
\noindent Then it is clear that 
\[
 \text{spt} \, \tilde{\mu}_\tau = \left \{ \begin{pmatrix} 0 & b \\ b & 0 \end{pmatrix}  :  b \in \mathcal{C}  \right \} \,
\]
where $\mathcal{C}$ is the Cantor set.  Notice that for $\delta > 0$
\begin{eqnarray*}
[d (\cdot , \tau) ]^{-1} ( \delta ) 
& =  &  \left \{ (a,b,c) \in \mathbb{R}^3  \, :  \,   \inf_{p \in \mathcal{C} } \min_{s=0,1}  \left| \frac{a+c}{2} + (-1)^{s} \sqrt{  \left( \frac{a+c}{2}  \right)^2 + \bigg(  (b-p )^2 - ac \bigg)  } \right| < \delta \right \}  \\
& =  & \left \{ (a,b,c) \in \mathbb{R}^3 \, :  \,   \inf_{p \in \mathcal{C} }  \left|  \left| \frac{a+c}{2} \right|  - \sqrt{  \left( \frac{a-c}{2}  \right)^2 +  (b-p)^2  } \right| < \delta \right \} \,.
\end{eqnarray*}
To further understand this set $ [d (\cdot , \tau) ]^{-1} ( \delta )$ (from this example, and similar sets in the next two examples) we realize that, fixing any $p \in \mathbb{R}$ for all $\delta > 0 $ and $R_0 > 0 $, we have $ 0 < \underline{\delta}_{R_0} <  \overline{\delta}_{R_0} $ such that
\begin{eqnarray*}
& & B_{R_0} \left( (0, p, 0)  \right) \bigcap \bigcup_{   \left \{ (a,b,c) \in  \mathbb{R}^3 \, :  \, ( b - p )^2 = a c  \right \}   }   B_{  \underline{\delta}_{R_0} } \left( (a, b, c)  \right)  \\
& \subset &  \left \{ (a,b,c) \in  B_{R_0} \left( (0, p, 0)  \right) \, :  \,   \left|  \left| \frac{a+c}{2} \right|  - \sqrt{  \left( \frac{a-c}{2}  \right)^2 +  ( b-p )^2  } \right| < \delta \right \}  \\
& \subset & B_{R_0} \left( (0, p, 0)  \right) \bigcap \bigcup_{  \left \{ (a,b,c) \in  \mathbb{R}^3 \, :  \, ( b - p )^2 = a c  \right \}    }   B_{  \overline{\delta}_{R_0} } \left( (a, b, c)  \right) \,,
\end{eqnarray*}
and therefore, inside any Euclidean ball of finite radius $R$, the set in the middle is containing and is contained in two Euclidean neighborhoods of the elliptic cone $ \{  (b - p)^2 = a c  \}$.

Hence, if we only focus on the subset inside any Euclidean ball of finite radius $R>0$ (which we can, by making our choice of the smooth cutoff function $\eta$ with compact support,) $ [d (\cdot , \tau) ]^{-1} ( \delta )$ can be effectively regarded as a union over $p \in \mathcal{C} $ of Euclidean neighborhoods of the elliptic cones $ \{  ( b - p)^2 =  ac  \} =  \{  (b - p)^2 +  n^2 = m^2  \}$ after a change of variable $(a,c) \mapsto (m,n ) := (\frac{a+c}{2}, \frac{a-c}{2})$.
Moreover, we notice that $   [d (\cdot , \tau) ]^{-1} ( 0)   $ is the following union of the aforementioned elliptic cones 
\begin{eqnarray*}
[d (\cdot , \tau) ]^{-1} ( 0) 
& = &   \bigcup_{ p \in \mathcal{C} } \left \{ (a,b,c) \in  \mathbb{R}^3  \, :  \, (b - p)^2 = a c  \right \} \\
&= &  \bigcup_{ ac \geq 0}  \left \{ (a,b,c)  \, :  \,   b \in  \mathcal{C}  +  \partial B_{\sqrt{a c} }(0)  \right \} \,, \\
&=&  [d (\cdot , \tau_1) ]^{-1} ( 0)  \bigcup  [d (\cdot , \tau_2) ]^{-1} ( 0)  
\end{eqnarray*}
which is again a union of two non-disjoint sets, where we write
\[
\tau = \frac{1}{2} \tau_1 + \frac{1}{2} \tau_2  := \frac{1}{2} \left(\frac{1}{3}\right)_* \tau + \frac{1}{2} \left( T_{ - X_0 } \right)_*  \left(\frac{1}{3}\right)_* \tau
\]
with $X_0 = \begin{pmatrix} 0 & \frac{2}{3} \\ \frac{2}{3} & 0 \end{pmatrix} $.  

\begin{figurehere} \centering
\includegraphics[width=5cm,height=4cm]{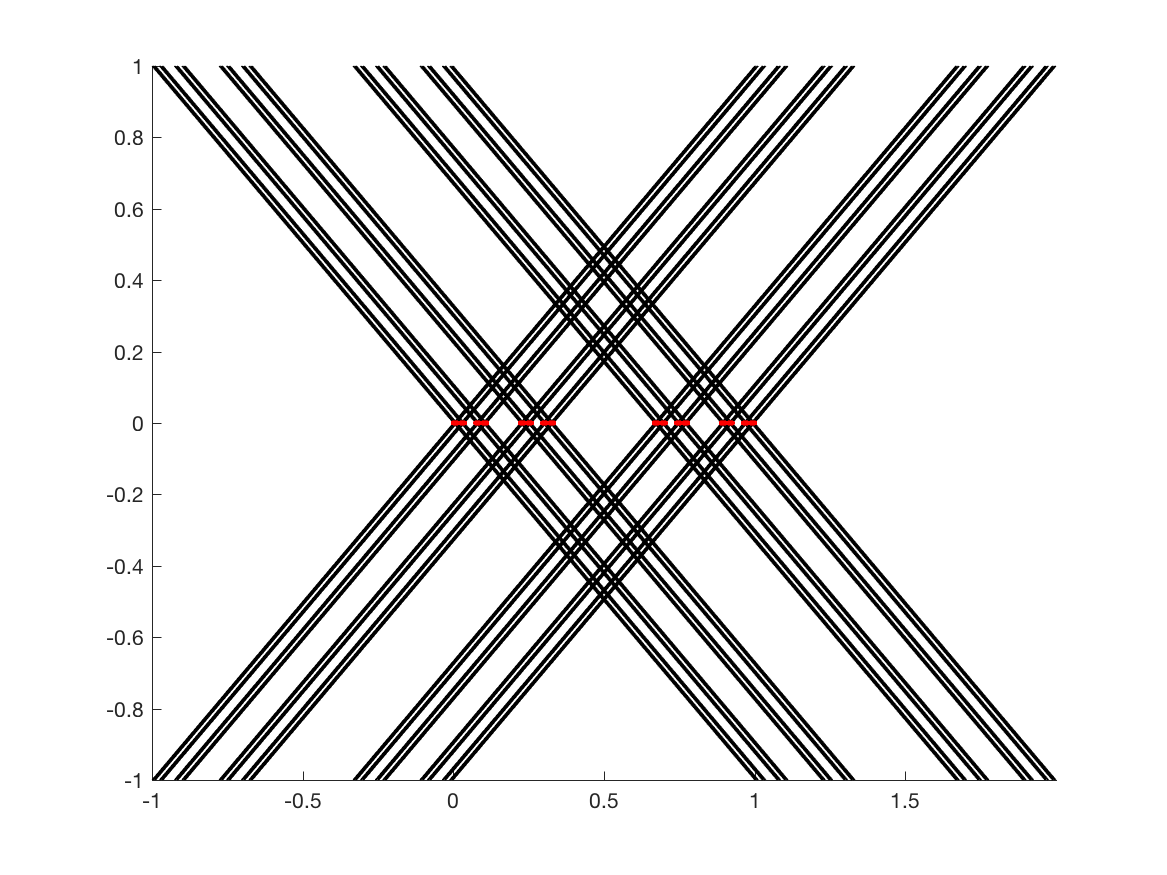}
\caption{The cross-section of $[d (\cdot , \tau) ]^{-1} ( 0) $ along the plane $a + c = 0$ on the variables $(a, b ,c) \in \mathbb{R}^3$, (the Cantor set is colored red.) } \label{test1A}
\end{figurehere}

Meanwhile, there exists a function $\eta$ of compact support such that
\[
\eta  =  \eta_1 + \eta_2 + \varphi := \left(\frac{1}{3}\right)^* \eta + T_{(0, \frac{2}{3},0)}^*  \left(\frac{1}{3}\right)^* \eta  + \varphi
\]
where $\varphi > 0$ is defined above and is also smooth and of compact support, and that for some $\epsilon > 0$
\[
\text{spt}\left( \eta_1 \right) \bigcap  [d (\cdot , \tau_2) ]^{-1} ( \epsilon)  = \text{spt}\left(  \eta_2 \right)  \bigcap  [d (\cdot , \tau_1) ]^{-1} ( \epsilon )  = \emptyset\,.
\]
Then
\[
 \nu = \nu_1 + \nu_2 + \nu_0 := 3^{-3} \left(\frac{1}{3}\right)_* \nu +   3^{-3}  \left( T_{ - X_0 } \right)_*  \left(\frac{1}{3}\right)_* \nu +  \int_{  \mathbb{R}^2 \times \mathbb{C} }   \frac{1}{2} \text{tr} (\cdot) \, \varphi \, da  \, db  \, dc  \,.
\]
Together with the transformation and decomposition rules Lemmas \ref{transformation_lemma} and \ref{decomposition_lemma}, writing $k = \frac{1}{3} $ we have
\begin{eqnarray*}
& & [\zeta_{ \tau } (s, \nu) ] (1)  \\
& = & [\zeta_{ \tau } (s, \nu_1) ] (1)  + [\zeta_{ \tau } (s, \nu_2) ] (1) + [\zeta_{ \tau } (s, \nu_0) ] (1)  \\
& = & [\zeta_{ \tau_1 } (s, \nu_1) ] (1)  + [\zeta_{ \tau_2 } (s, \nu_2) ] (1) + [\zeta_{ \tau } (s, \nu_0) ] (1) + [h(s)](1)  \\
& = & k^{3} [\zeta_{ k_* \tau} (s, k_* \nu) ] (1)  + k^{3} [\zeta_{ T_{- X_0}^* k_* \tau} (s,  T_{- X_0}^* k_* \nu) ] (1)  + [\zeta_{ \tau } (s, \nu_3) ] (1) + [h(s)](1)  \\
&=& 2 k^{s - 1 + 3 }  \, [\zeta_{\tau} (s, \nu) ](1)  + [\zeta_{ \tau } (s, \nu_0) ] (1) + [h(s)](1) \,,
\end{eqnarray*}
with $h$ holomorphic, which is a functional equation similar to the form introduced in \cite{Hoffer}.
Or
\begin{eqnarray*}
\left(1 -2 k^{s - 1 + 3 }  \right) [\zeta_{ \tau } (s, \nu) ] (1)  =  [\zeta_{ \tau } (s, \nu_0) ] (1) + [h(s)](1) \,,
\end{eqnarray*}
Or
\begin{eqnarray*}
[\zeta_{ \tau } (s, \nu) ] (1)  =  \frac{ [\zeta_{ \tau } (s, \nu_0) ] (1) + [h(s)](1) } {  1 -2  \cdot 3^{-s -2} } \,.
\end{eqnarray*}
and therefore we speculate that it will have poles satisfying
\[
[ \mathcal{D}_{\mathbb{C}} ( \tau, \nu) ](1) \bigcap \left \{ \Re(s) \geq  \frac{\log(2)}{\log(3)} - 2 \right\} \subset \left \{ \frac{\log(2)}{\log(3)} - 2 + i n \frac{2\pi}{\log(3)} \, : \, n \in \mathbb{N} \right \}  \,,
\]
We note that the main difficulty to rigorously establish the above statement boils down to the need for a more detailed analysis of $ [\zeta_{ \tau } (s, \nu_0) ] (1) $ so that an explicit bound can be established such that $\overline{\text{dim}}_B (\tau, \nu_0)  <  \overline{\text{dim}}_B (\tau, \nu) $ can be concluded.  Since the aim of this work is to introduce our definition of relative distance zeta functional and fractality in this context instead of diving into specific details of a particular example, we will defer this detailed analysis to a next work.
Meanwhile, for illustrative purpose and a better understanding of our zeta functional, in Figure \ref{test1B} we plot the norms of the value of our zeta functional, which we can compare with Figure \ref{test1A}.

\begin{figurehere} \centering
\includegraphics[width=5cm,height=4cm]{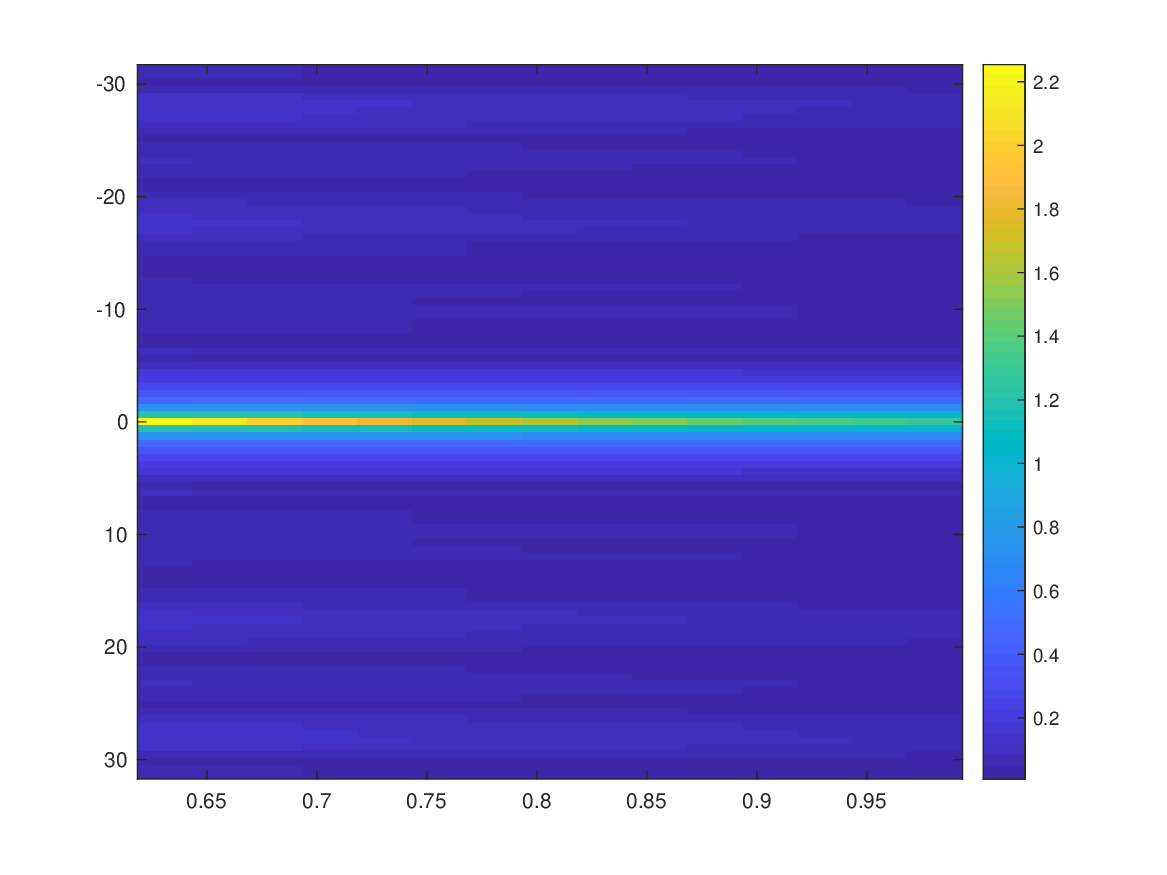}
\caption{$| \left[ \zeta_{\tau} (s, \nu) \right] (1) |$ for $\frac{\log(2)}{ \log(3)} \leq \Re (s) $  with $\eta(a,b,c) =e^{- 0.001 (a^2 + b^2 + c^2)}$.} \label{test1B}
\end{figurehere}

\noindent We notice that, even though this is not the most exciting example with the set $ \text{spt} \, \tilde{\mu}_\tau $ in fact simultaneously diagonizable, i.e.:
\[
\begin{pmatrix} 0 & a \\  a & 0 \end{pmatrix}   \begin{pmatrix} 0 & b \\ b & 0 \end{pmatrix}
=  \begin{pmatrix} a  b & 0 \\ 0 &  a b \end{pmatrix} 
= \begin{pmatrix} 0 & b \\ b & 0 \end{pmatrix}   \begin{pmatrix} 0 & a \\ a & 0 \end{pmatrix} \,,
\]
the set $ [d (\cdot , \tau) ]^{-1} ( 0)  $ is already geometrically intriguing as a unions of cones $\{ ( b - p)^2 = a c \}$ over $p\in \mathcal{C}$.

\end{example}

\begin{example} (A noncommutative example.)
Let us still keep $d=1$, $n=2$. We now let
$ \tau  = \int_{ [ \mathfrak{sa} \, \mathcal{L} (H) ]^d  } \frac{1}{2}  \, \text{tr}  (\cdot)  \, \tilde{\mu}_\tau (d X) := \frac{1}{2}  \, \int_{ \mathbb{R}^3 } \text{tr} (\cdot)  \, \delta_c( da )  \mu_\mathcal{C} (d b)  \mu_\mathcal{C} (dc) \,  $ where $\mu_\mathcal{C}$ is again the Cantor measure, and again take $ \nu =  \int_{ [ \mathfrak{sa} \, \mathcal{L} (H) ]^d  } \frac{1}{2}  \, \text{tr} (\cdot)  \tilde{\mu}_{\nu} (d X)  = \int_{\mathbb{R}^3 } \frac{1}{2}  \, \text{tr} (\cdot) \, \eta(a,b,c)
\,  da  \, db  \, dc $ 
where $\eta$ is smooth and to be defined later.
Then it is clear that 
\[
 \text{spt} \, \tilde{\mu}_\tau = \left \{ \begin{pmatrix} a & b \\ b & a \end{pmatrix}  :  a, b \in \mathcal{C}  \right \} \, ,
\]
and it is then straightforward to check for $\delta > 0$
{\small
\begin{eqnarray*}
&   &   [d (\cdot , \tau) ]^{-1} ( \delta ) \\
& =&   \left \{ (a,b,c)  \in \mathbb{R}^3  \, :  \,   \inf_{p , q \in \mathcal{C} } \min_{s=0,1}  \left| \frac{a+c}{2} - q + (-1)^{s} \sqrt{  \left( \frac{a+c}{2}  - q  \right)^2 + \bigg(  ( b-p )^2 - (a - q) (c -q ) \bigg)  } \right| < \delta \right \}  \\
& =&   \left \{ (a,b,c)  \in \mathbb{R}^3  \, :  \,   \inf_{p , q \in \mathcal{C} }  \left| \left| \frac{a+c}{2} - q \right| - \sqrt{  \left( \frac{a-c}{2} \right)^2 +   ( b-p )^2   } \right| < \delta \right \} 
\end{eqnarray*}
}which, if we only focus on the subset inside any Euclidean ball of finite radius $R > 0 $, can be effectively regarded as a union over $(p,q) \in \mathcal{C} \times  \mathcal{C} $ of Euclidean neighborhoods of the elliptic cones $ \{  ( b - p )^2 =  (a -q) (c - q)   \} =  \{  ( b - p )^2 +  (n - q)^2 = m^2  \}.$
Moreover, we observe that
\[
[d (\cdot , \tau) ]^{-1} ( 0)   =  \bigcup_{ p,q \in \mathcal{C}  }   \left \{ (a,b,c) \in \mathbb{R}^3   \, :  \, ( b-p )^2 = (a -q) (c - q) \right \} \,.
\]
\begin{figurehere} \centering
\includegraphics[width=5cm,height=4cm]{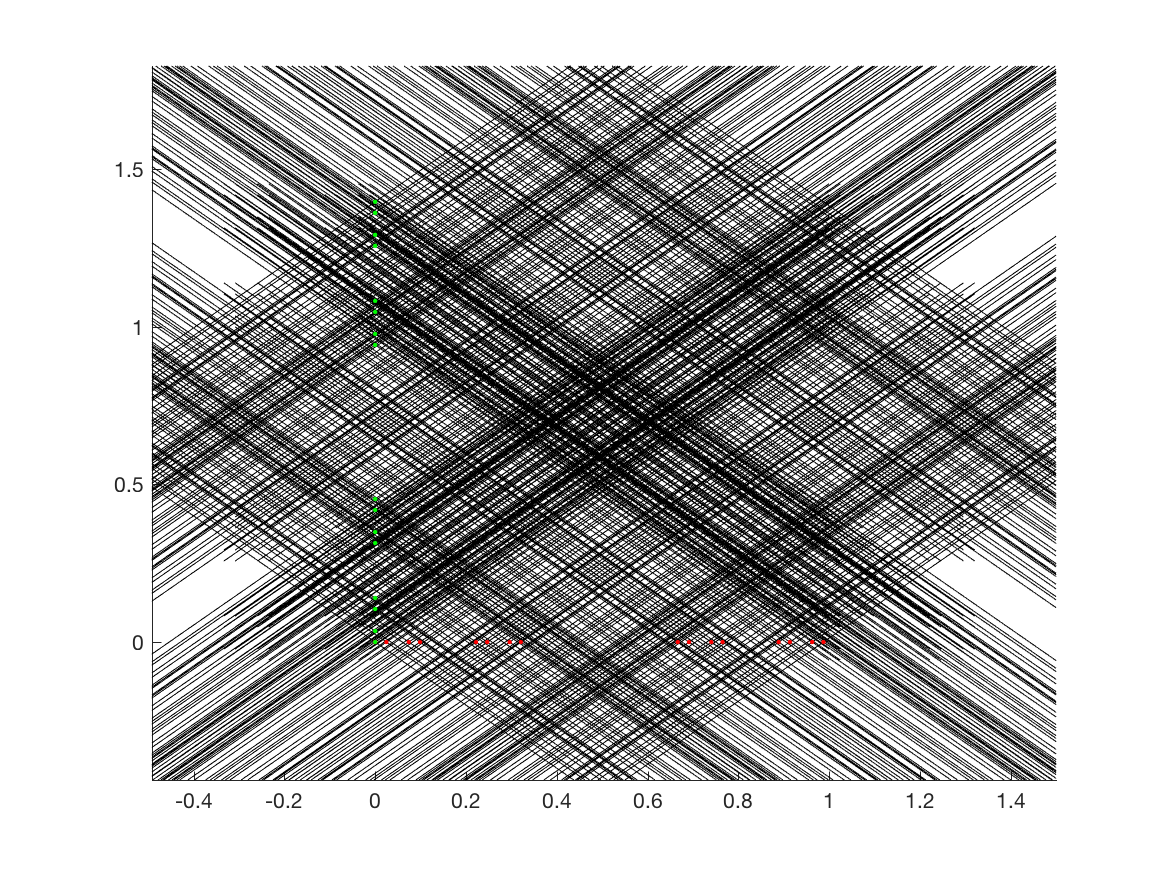}
\caption{The cross-section of $[d (\cdot , \tau) ]^{-1} ( 0) $ along the plane $a + c = 0$ with the variables $(a, b ,c) \in \mathbb{R}^3$, (the Cantor sets $\mathcal{C}$ and $\sqrt{2} \, \mathcal{C}$ are colored red and green respectively.)    } \label{test1C}
\end{figurehere}

Write
\begin{eqnarray*}
[d (\cdot , \tau) ]^{-1} ( 0) 
=  \bigcup_{i=1}^4 [d (\cdot , \tau_i) ]^{-1} ( 0) 
\end{eqnarray*}
as a union of four non-disjoint sets, where
\[
\tau = \frac{1}{4} \sum_{i=1}^d \tau_4 = \frac{1}{4} \left(\frac{1}{3}\right)_* \tau + \frac{1}{4} \left( T_{ - X_1 } \right)_*  \left(\frac{1}{3}\right)_* \tau +\frac{1}{4}   \left( T_{ - X_2 } \right)_*  \left(\frac{1}{3}\right)_* \tau  + \frac{1}{4}  \left( T_{ - X_1 - X_2 } \right)_*  \left(\frac{1}{3}\right)_* \tau
\]
where we write $X_1 = \begin{pmatrix} 0 & \frac{2}{3} \\ \frac{2}{3} & 0 \end{pmatrix} $, $X_2 = \begin{pmatrix} \frac{2}{3}  & 0 \\ 0 & \frac{2}{3}  \end{pmatrix} $.  
Meanwhile, there exists a function $\eta$ of compact support such that
\[
\eta  =  \frac{1}{4} \sum_{i=1}^d \eta_i  + \varphi : =  \left(\frac{1}{3}\right)^* \eta + T_{(0, \frac{2}{3},0)}^*  \left(\frac{1}{3}\right)^* \eta + T_{(\frac{2}{3}, 0,\frac{2}{3})}^*  \left(\frac{1}{3}\right)^* \eta + T_{(\frac{2}{3}, \frac{2}{3},\frac{2}{3})}^*  \left(\frac{1}{3}\right)^* \eta    + \varphi
\]
where $\varphi > 0$ is defined above and is also smooth and of compact support, and that for all $i\neq j$
\[
\text{spt}\left( \eta_i \right) \bigcap  [d (\cdot , \tau_j) ]^{-1} ( \epsilon)  = \emptyset\,.
\]
for some $\epsilon > 0$.
Then
\begin{eqnarray*}
&  &  \nu \\
&= & \sum_{i=1}^{4} \nu_i + \nu_0 \\
& := &3^{-3} \left(\frac{1}{3}\right)_* \nu +   3^{-3}  \left( T_{ - X_1 } \right)_*  \left(\frac{1}{3}\right)_* \nu +   3^{-3}  \left( T_{ - X_2 } \right)_*  \left(\frac{1}{3}\right)_* \nu +   3^{-3}  \left( T_{ - X_1- X_2 } \right)_*  \left(\frac{1}{3}\right)_* \nu \\
&  & + \int_{ \mathbb{R}^3}  \frac{1}{2} \text{tr} (\cdot )\, \varphi \, da  \, db \, dc
\end{eqnarray*}
Together with the transformation and decomposition rules Lemmas \ref{transformation_lemma} and \ref{decomposition_lemma}, writing $k = \frac{1}{3} $ we have
\begin{eqnarray*}
& & [\zeta_{ \tau } (s, \nu) ] (1)  \\
& = & \sum_{i=1}^4 [\zeta_{ \tau } (s, \nu_i) ] (1)  +  [\zeta_{ \tau } (s, \nu_0) ] (1)  \\
& = & \sum_{i=1}^4 [\zeta_{ \tau_i } (s, \nu_i) ] (1)  +  [\zeta_{ \tau } (s, \nu_0) ] (1) + [h(s)](1)  \\
&=& 4 k^{s - 1+3}  \, [\zeta_{\tau} (s, \nu) ](1)  + [\zeta_{ \tau } (s, \nu_0) ] (1) + [h(s)](1) \,.
\end{eqnarray*}
Or
\begin{eqnarray*}
\left(1 -4 k^{s - 1+3}  \right) [\zeta_{ \tau } (s, \nu) ] (1)  =  [\zeta_{ \tau } (s, \nu_0) ] (1) + [h(s)](1) \,,
\end{eqnarray*}
with $h$ holomorphic, which is a functional equation similar to the one introduced in \cite{Hoffer}.
Or
\begin{eqnarray*}
[\zeta_{ \tau } (s, \nu) ] (1)  =  \frac{ [\zeta_{ \tau } (s, \nu_0) ] (1) + [h(s)](1) } {  1 -4 \cdot 3^{-s - 2} } \,.
\end{eqnarray*}
and therefore we speculate it has poles
\[
[ \mathcal{D}_{\mathbb{C}} ( \tau, \nu) ](1) \bigcap \left \{ \Re(s) \geq  2 \frac{\log(2)}{\log(3)} - 2 \right \}  \subset \left \{ 2 \frac{\log(2)}{\log(3)} - 2 + i n \frac{2\pi}{\log(3)} \, : \, n \in \mathbb{N} \right \}  \,.
\]
Again, we defer the more detailed analysis of $ [\zeta_{ \tau } (s, \nu_0) ] (1) $ so as to rigorously establish the above statement to a next work.
Meanwhile, so as to have a better understanding of our zeta functional (for comparison with the more conventional zeta function), we show the norms of the values of our zeta functions in Figure \ref{test1E} for illustrative purpose and comparison with Figure \ref{test1A}.

\begin{figurehere} \centering
\includegraphics[width=5cm,height=4cm]{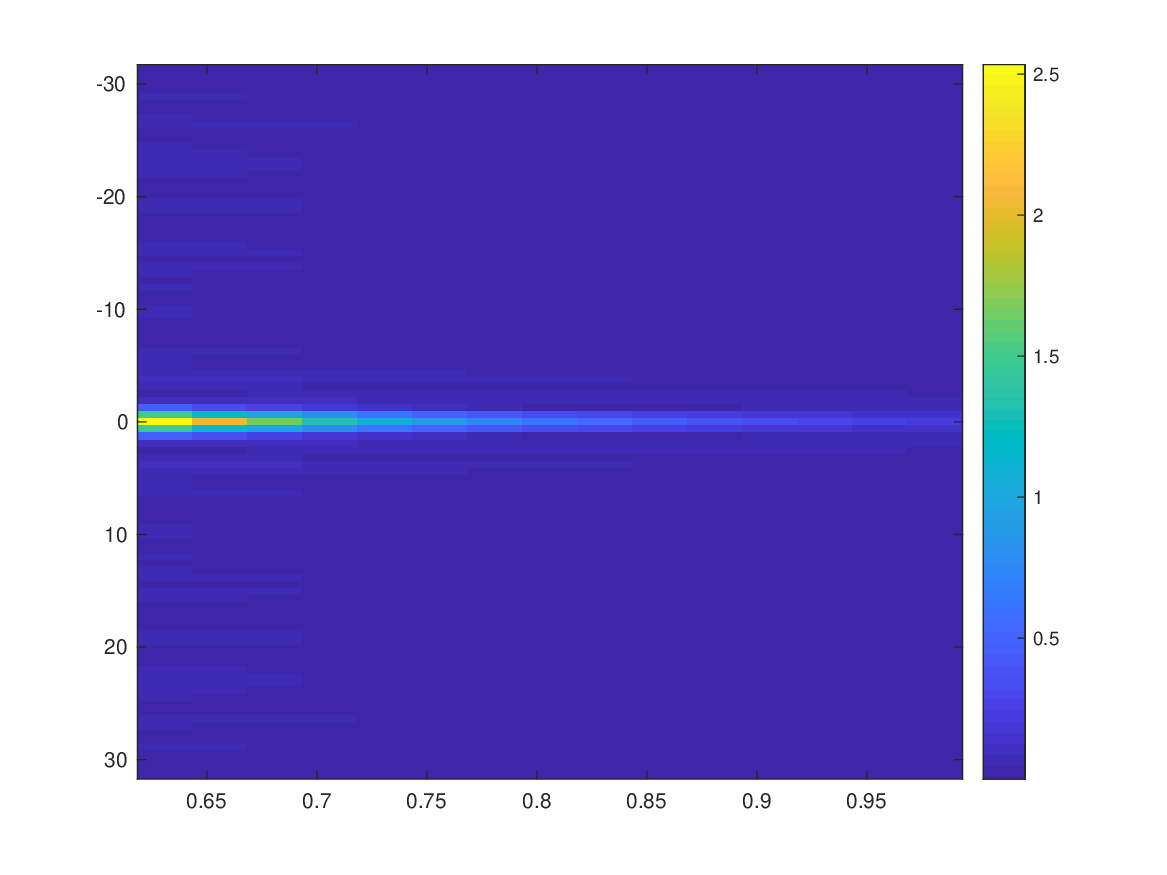}
\caption{$| \left[ \zeta_{\tau} (s, \nu ) \right] (1) |$ for $\frac{\log(2)}{ \log(3)} \leq \Re (s) $ with $\eta(a,b,c) =e^{- 0.001 (a^2 + b^2 + c^2)}$.} \label{test1E}
\end{figurehere}

\noindent  We remark that the set $ [d (\cdot , \tau) ]^{-1} ( 0)  $ is not simultaneously diagonizable, and is therefore the first noncommutative example of interest.

\end{example}

\begin{example} (A noncommutative example.)
Still, in this example, let us keep $d=1$, $n=2$.
Let
$ \tau  = \int_{ [ \mathfrak{sa} \, \mathcal{L} (H) ]^d  }  \frac{1}{2}  \, \text{tr} (\cdot)  \tilde{\mu}_{\tau} (d X) :=  \int_{ \mathbb{R}^3 } \frac{1}{2}  \, \text{tr} (\cdot) \, \mu_\mathcal{C} ( d b) \,  \chi_{\left[-\frac{1}{2},\frac{1}{2}\right]}( da ) \, \chi_{\left[-\frac{1}{2},\frac{1}{2}\right]} ( dc )$, and again $ \nu =  \int_{ [ \mathfrak{sa} \, \mathcal{L} (H) ]^d  }  \frac{1}{2}  \, \text{tr} (\cdot) \tilde{\mu}_{\nu} (dX)  = \int_{ \mathbb{R}^3 } \frac{1}{2}  \, \text{tr} (\cdot) \, \eta(a,b,c) 
\,  da  \, db  \, dc $
for a smooth function $\eta$ with compact support to be chosen later.

Then it is clear that 
\[
 \text{spt} \, \tilde{\mu}_\tau = \left \{ \begin{pmatrix} a & b \\ b & c \end{pmatrix}  :  b \in \mathcal{C} , a,c \in \left[-\frac{1}{2},\frac{1}{2}\right] \right \} \, ,
\]
and we may check that for $\delta > 0$
{\footnotesize
\begin{eqnarray*}
&    & [d (\cdot , \tau) ]^{-1} ( \delta )  \\
& = &   \left \{ (a,b,c)  \in \mathbb{R}^3  \, :  \,   \inf_{ \tiny \begin{matrix} p \in \mathcal{C}, \\  q,r \in \left[-\frac{1}{2},\frac{1}{2}\right]  \end{matrix} } \min_{s=0,1}  \left| \frac{a +c - q - r }{2} + (-1)^{s} \sqrt{  \left( \frac{a +c -q - r}{2}  \right)^2 + \bigg(  ( b-p )^2 - (a - q) (c -r ) \bigg)  } \right| < \delta \right \}  \\
& = &   \left \{ (a,b,c)  \in \mathbb{R}^3   \, :  \,   \inf_{ \tiny \begin{matrix} p \in \mathcal{C}, \\  q,r \in \left[-\frac{1}{2},\frac{1}{2}\right]  \end{matrix} }  \left| \left|  \frac{a +c - q - r }{2} \right| - \sqrt{  \left( \frac{a -c -q + r}{2}  \right)^2 + ( b-p )^2  \bigg)  } \right| < \delta \right \} 
\end{eqnarray*}
}which again, if we only focus on the subset inside any Euclidean ball of finite radius $R > 0 $, can be effectively regarded as a union over $(p,q,r) \in \mathcal{C} \times \left[-\frac{1}{2},\frac{1}{2}\right]^2 $ of Euclidean neighborhoods of the elliptic cones $ \{ ( b - p )^2 =  (a -q) (c - r)   \} =  \{  ( b - p )^2 +  (n - \frac{q + r}{2} )^2 = (m  - \frac{q - r}{2}  )^2  \}.$
Moreover, we have
\begin{eqnarray*}
&  &  [d (\cdot , \tau) ]^{-1} ( 0)   \\
&=&  \bigcup_{ p \in \mathcal{C}, \,  q,r \in \left[-\frac{1}{2},\frac{1}{2}\right]   }   \left \{ (a,b,c)  \in \mathbb{R} \times \mathbb{C} \times \mathbb{R}   \, :  \, ( b-p )^2 = (a -q) (c - r) \right \} \\
&=&  [d (\cdot , \tau_1) ]^{-1} ( 0)  \bigcup  [d (\cdot , \tau_2) ]^{-1} ( 0) \bigcup  [d (\cdot , \tau_0) ]^{-1} ( 0)  
\end{eqnarray*}
as a union of three non-disjoint set, where
\[
\tau = \frac{1}{2} \tau_1 + \frac{1}{2} \tau_2  + \tau_0  := 3^{-2} \frac{1}{2}  \left(\frac{1}{3}\right)_* \tau +  3^{-2} \frac{1}{2} \left( T_{ - X_0 } \right)_*  \left(\frac{1}{3}\right)_* \tau  + \tau_0
\]
where we write $X_0 = \begin{pmatrix} 0 & \frac{2}{3} \\ \frac{2}{3} & 0 \end{pmatrix} $.  
Meanwhile, there exists $\eta$ of compact support such that
\[
\eta  =  \eta_1 + \eta_2 + \varphi := \left(\frac{1}{3}\right)^* \eta + T_{(0, \frac{2}{3},0)}^*  \left(\frac{1}{3}\right)^* \eta  + \varphi
\]
where $\varphi > 0$ is defined above and is also smooth and of compact support, and that for some $\epsilon > 0 $,
\[
\text{spt}\left( \eta_1 \right) \bigcap  [d (\cdot , \tau_2) ]^{-1} ( \epsilon )  = \text{spt}\left(  \eta_2 \right)  \bigcap  [d (\cdot , \tau_1) ]^{-1} ( \epsilon )  = \emptyset\,.
\]
Then
\[
 \nu = \nu_1 + \nu_2 + \nu_3 := 3^{-3} \left(\frac{1}{3}\right)_* \nu +   3^{-3}  \left( T_{ - X_0 } \right)_*  \left(\frac{1}{3}\right)_* \nu +  \int_{ \mathbb{R}^2 \times \mathbb{C}  } \frac{1}{2} \text{tr} (\cdot) \, \varphi \, da  \, db \, dc 
\]
Hence, with the transformation and decomposition rules Lemmas \ref{transformation_lemma} and \ref{decomposition_lemma}, writing $k = \frac{1}{3} $ we have
\begin{eqnarray*}
& & [\zeta_{ \tau } (s, \nu) ] (1)  \\
& = & [\zeta_{ \tau } (s, \nu_1) ] (1)  + [\zeta_{ \tau } (s, \nu_2) ] (1) + [\zeta_{ \tau } (s, \nu_0) ] (1)  \\
& = & [\zeta_{ \tau_1 } (s, \nu_1) ] (1)  + [\zeta_{ \tau_2 } (s, \nu_2) ] (1) + [\zeta_{ \tau } (s, \nu_0) ] (1) + [h(s)](1)  \\
& = & k^{3-2} [\zeta_{ k_* \tau} (s, k_* \nu) ] (1)  + k^{3-2} [\zeta_{ T_{- X_0}^* k_* \tau} (s,  T_{- X_0}^* k_* \nu) ] (1)  + [\zeta_{ \tau } (s, \nu_3) ] (1) + [h(s)](1)  \\
&=& 2 k^{s -1 + 3 - 2}  \, [\zeta_{\tau} (s, \nu) ](1)  + [\zeta_{ \tau } (s, \nu_0) ] (1) + [h(s)](1) \,.
\end{eqnarray*}
with $h$ holomorphic, which is again a functional equation similar to that is introduced in \cite{Hoffer}.  Or
\begin{eqnarray*}
\left(1 -2 k^{s - 1 + 3 -2 }  \right) [\zeta_{ \tau } (s, \nu) ] (1)  =  [\zeta_{ \tau } (s, \nu_0) ] (1) + [h(s)](1) \,,
\end{eqnarray*}
Or
\begin{eqnarray*}
[\zeta_{ \tau } (s, \nu) ] (1)  =  \frac{ [\zeta_{ \tau } (s, \nu_0) ] (1) + [h(s)](1) } {  1 -2  \cdot 3^{-s} } \,.
\end{eqnarray*}
and therefore we suspect they will have poles
\[
[ \mathcal{D}_{\mathbb{C}} ( \tau, \nu) ](1) \bigcap \left \{ \Re(s) \geq  \frac{\log(2)}{\log(3)} \right \}  \subset \left \{ \frac{\log(2)}{\log(3)} + i n \frac{2\pi}{\log(3)} \, : \, n \in \mathbb{N} \right \}  \,.
\]
We again remark that we shall defer the more detailed analysis of $ [\zeta_{ \tau } (s, \nu_0) ] (1) $ so as to rigorously establish the above statement to a next work.
\end{example}

\begin{remark}
We note that in fact, with the above procedure and a similar argument as in \cite{Hoffer}, it is expected that we shall obtain a general functional equation using Lemmas \ref{transformation_lemma} and \ref{decomposition_lemma} whenever a general self-similar structure arises.  
Similar functional equations satisfied by the geometric, fractal and distance zeta function with rescaling were previously also introduced and used earlier in different context, including \cite{KL93,LalLap1,LalLap2,Lal88,Lal89,L3_ref,LF1,LF2,LF_book,LapPe1,LapPe2,LapPeWi1,LapPeWi2,LRZ_book,LRZ,LRZ4}.
This is a direction worth pursuing in our next work, especially rigorously establishing concrete criteria so as to guarantee an associated scaling zeta function similar to Definition 4.2 in \cite{Hoffer} and in \cite{LapPe2,LapPeWi1,LRZ_book} alone may recover all the poles of our relative distance zeta functional.
\end{remark}

\section{Conclusion and future directions}

In this work, we pioneered in defining a relative distance (and tube) zeta functional for a class of states over a C* algebra as a generalization of the relative distance zeta function introduced in \cite{LF1,LF2,LF_book, LRZ0,LRZ_book,LRZ,LRZ3,LRZ4,LRZ5,LRZ6a,LRZ6b,LRZ7}.  With this, we are able to define relative Minkowski dimensions and relative complex dimensions in this context, which provides essential tools to describe the notion of fractality in the noncommutative setting.  We discussed the geometric and transformation properties, decomposition rules and properties that respects tensor products for this newly proposed relative distance zeta functional, and then looked into some intriguing examples with this new zeta functional by exploring possible functional equations similar to \cite{Hoffer,KL93,LalLap1,LalLap2,Lal88,Lal89,L3_ref,LapPe1,LapPe2,LapPeWi1,LapPeWi2}. 

Our work provides a possible choice of set of tools to discuss what constitute a fractal in a noncommutative setting, more concretely in the context of the C*-algebraic formulation of quantum mechanics where their self-adjoint elements represent physical observables and their states (as a mapping from physical observables to their expected measurement outcomes) represent physically relevant mixed ``quantum states" of the system.  In this context, our characterization of fractality can now be understood as the periodicity in the logarithmic scale/multiplicative periodicity observed when the set of physical observables from where information for expected measurement outcomes can be obtained grows.

Important next steps to explore include thoroughly investigating more concrete (and perhaps physically relevant) examples in which a detailed understanding of how $t \mapsto \nu^{\tau,t}(g)$ grows can be established, rigorously establishing functional equations for our relative distance zeta functional whenever general self-similarity arises, examining concrete criteria that guarantees associated scaling zeta function alone may recover all the relative complex dimensions in our setting.

\section*{Acknowledgment}
The authors would like to thank (in alphabetical order of family names) S. Das, M. L. Lapidus, S. Rankin, A. Richardson, C. W. Ho and W. Hoffer for many tremendously helpful discussions and suggestions, which leads to a substantial improvement of the manuscript.

\end{document}